\documentclass[11pt]{article}
\usepackage{multirow}
\usepackage{color,url,eurosym,graphicx,amsfonts,amsmath,algorithm,algorithmic,bbm,xspace}
\usepackage{amsthm}
\usepackage{tikz}
\usetikzlibrary{arrows,shapes,positioning}

\usepackage[most]{tcolorbox}
 \usepackage{empheq}
 \usepackage{varwidth}

 \usepackage{MnSymbol}%
\usepackage{wasysym}% 

\newtcbox{\mymath}[1][]{%
    nobeforeafter, math upper, tcbox raise base,
    enhanced, colframe=blue!30!black,
    colback=blue!30, boxrule=1pt,
    #1}

\definecolor{ForestGreen}{rgb}{0.1333,0.5451,0.1333}
\usepackage[letterpaper,
            colorlinks,linkcolor=ForestGreen,citecolor=ForestGreen,
            bookmarks,bookmarksopen,bookmarksnumbered]
           {hyperref}
\usepackage[margin=1.5in]{geometry}

\newtheorem*{theorem*}{Theorem}

\newtheorem{problem}{Problem}
\newtheorem*{problem*}{Problem}

\newtheorem*{claim*}{Claim}
\newtheorem{corollary}{Corollary}
\newtheorem*{corollary*}{Corollary}

\newcommand{\field}[1]{\mathbb{#1}} % requires amsfonts

\newcommand{\Prob}[1]{\ensuremath{{\bf{Pr}}\left[{#1}\right]}}

\newcommand{\NPhard}{{\ensuremath{\mathbf{NP}}-hard}\xspace}

\newcommand{\NegDSD}{{\sc{Neg-DSD}}\xspace}

\newcommand{\spara}[1]{\smallskip\noindent{\bf #1}}

\newtheorem{theorem}{Theorem}

\newcommand{\hide}[1]{}
\newcommand{\squishlist}{
 \begin{list}{$\bullet$}
  {  \setlength{\itemsep}{0pt}
     \setlength{\parsep}{3pt}
     \setlength{\topsep}{3pt}
     \setlength{\partopsep}{0pt}
     \setlength{\leftmargin}{2em}
     \setlength{\labelwidth}{1.5em}
     \setlength{\labelsep}{0.5em}
} }
\newcommand{\squishlisttight}{
 \begin{list}{$\bullet$}
  { \setlength{\itemsep}{0pt}
    \setlength{\parsep}{0pt}
    \setlength{\topsep}{0pt}
    \setlength{\partopsep}{0pt}
    \setlength{\leftmargin}{2em}
    \setlength{\labelwidth}{1.5em}
    \setlength{\labelsep}{0.5em}
} }

\newcommand{\squishdesc}{
 \begin{list}{}
  {  \setlength{\itemsep}{0pt}
     \setlength{\parsep}{3pt}
     \setlength{\topsep}{3pt}
     \setlength{\partopsep}{0pt}
     \setlength{\leftmargin}{1em}
     \setlength{\labelwidth}{1.5em}
     \setlength{\labelsep}{0.5em}
} }

\newcommand{\squishend}{
  \end{list}
}

\newcommand{\squishlistt}{
 \begin{list}{---}
  {  \setlength{\itemsep}{0pt}
     \setlength{\parsep}{3pt}
     \setlength{\topsep}{3pt}
     \setlength{\partopsep}{0pt}
     \setlength{\leftmargin}{2em}
     \setlength{\labelwidth}{1.5em}
     \setlength{\labelsep}{0.5em}
} }

\begin{document}

\title{Novel Dense Subgraph Discovery Primitives: \\
\Large{Risk Aversion and Exclusion Queries} }
\author{
  Charalampos E. Tsourakakis \thanks{Boston University,
    \texttt{ctsourak@bu.edu}} 
  \and
  Tianyi Chen \thanks{Boston University,
    \texttt{ctony@bu.edu}} 
  \and 
  Naonori Kakimura \thanks{Keio University
 \texttt{kakimura@global.c.u-tokyo.ac.jp}} 
  \and 
  Jakub Pachocki \thanks{OpenAI \texttt{merettm@gmail.com}}
}

\date{}
\maketitle

\begin{abstract}
In the densest subgraph problem, given a weighted undirected graph $G(V,E,w)$, with non-negative edge weights $w:E \rightarrow \field{R}$,  we are asked to find a subset of nodes $S\subseteq V$ that maximizes the degree density $w(S)/|S|$, where $w(S)$ is the sum of the edge weights induced by $S$. This problem is a well studied problem, known as the {\em densest subgraph problem}, and is solvable in polynomial time. But what happens when the edge weights are negative, i.e., $w:E \rightarrow \field{R}$? Is the problem still solvable in polynomial time? Also, why should we care about the densest subgraph problem in the presence of negative weights?

In this work we answer the aforementioned question. Specifically, we provide two novel graph mining primitives that are applicable to a wide variety of applications. Our primitives can be used to answer questions such as ``how can we find a dense subgraph in Twitter with lots of replies and mentions but no follows?'', ``how do we extract a dense subgraph with high expected reward and low risk from an uncertain graph''? We formulate both problems mathematically as special instances of dense subgraph discovery in graphs with negative weights. We study the hardness of the problem, and we prove that the problem in general is \NPhard. We design an efficient approximation algorithm that works well in the presence of small negative weights, and also an effective heuristic for the more general case. Finally, we perform  experiments on various real-world uncertain graphs, and a crawled Twitter multilayer graph   that verify the value of the proposed primitives, and the practical value of our proposed algorithms.  

The code and the data are available at \url{https://github.com/negativedsd}.
\end{abstract}

\newpage

\section{Introduction}
\label{sec:intro}
Dense subgraph discovery (abbreviated as {\em DSD} henceforth) is a major and active topic of research in the fields of graph  algorithms and graph mining.  A wide range of real-world, data mining applications rely on DSD including correlation mining, fraud detection, electronic commerce, bioinformatics, mining Twitter data, efficient algorithm design for fast distance queries in massive networks,  and graph compression \cite{gionis2015dense}.

In this work we introduce two novel primitives for DSD. These two primitives are strongly motivated by real-world applications that we discuss in greater detail in Section~\ref{sec:motivation}. The first  question that our work addresses is related to uncertain graphs. 
Uncertain graphs appear in a wide variety of applications that we survey in Section~\ref{sec:related}. We define the uncertain graph model we use  formally in Section~\ref{sec:motivation}, but intuitively,  uncertain graphs model probabilistically real-world scenarios where each edge may exist or not in a graph (e.g., failure of a link). Problem~\ref{prob1} aims to find a {\em risk-averse} dense subgraph.  A similar formulation was suggested recently by Tsourakakis et al.  for graph matchings \cite{tsourakakis2018risk}. 

\vspace{2mm}
\begin{tcolorbox}
\begin{problem}[Risk-averse DSD] 
\label{prob1} 
Given an uncertain graph $\mathcal{G}$, how do we find a set of nodes $S$ that induces a dense subgraph in expectation, and the probability of not being dense in a realization/sample of $\mathcal{G}$ is low?
\end{problem}
\end{tcolorbox}
\vspace{2mm}

Our second problem focuses on multigraphs whose edges are associated with different types. Such graphs appear naturally in numerous applications, and are also known as multilayer multigraphs, e.g., \cite{galimberti2017core,yang2012predicting}. For example, Twitter users may interact in various ways, including {\em follow, reply, mention, retweet, like}, and {\em quote}. Similarity between two videos can be defined based on different criteria, e.g., audio, visual, and how frequently these videos are being co-watched on Youtube. Similarity between time series can be defined using a variety of measures including Euclidean distance, Fourier coefficients, dynamic time wraping, edit distance among others \cite{gunopulos2001time,serra2014empirical}. Emails between people can be classified bases on the nature of the interaction (e.g., business, family). We formulate Problem~\ref{prob2} whose goal is to detect efficiently dense subgraphs that exclude certain types of edges. Later, we will define two variations of this problem, soft- and hard-exclusion queries.

\vspace{2mm}
\begin{tcolorbox}
\begin{problem}[DSD-Exclusion-Queries]
\label{prob2} 
Given a  multigraph $G(V,E,\ell)$, where $\ell: E \rightarrow \{1,\ldots,L\}=[L]$ is the labeling function, and $L$ is the number of types of interactions, and an input set $\mathcal{I} \subseteq [L]$ of  interactions,  how do we find a set of nodes $S$ that (i) induces a dense subgraph, and (ii) does not induce any edge $e$ such that $\ell(e) \in \mathcal{I}$?
\end{problem}
\end{tcolorbox}
 \vspace{2mm}

\spara{Contributions.} Our contributions are summarized as follows. 

\squishlist
\item  We introduce two novel problems, (i) risk averse DSD,  and (ii) DSD in large-scale multilayer networks with exclusion queries.  In Section~\ref{sec:motivation} we show that these two problems are  special cases of DSD in undirected graphs with negative weights. To the best of our knowledge, this is the first work that introduces these algorithmic primitives.
\smallskip 
\item  We prove that DSD in the presence of negative weights is \NPhard in general by reducing {\sc Max-Cut} to our problem (Section~\ref{subsec:hardness}). 
\smallskip 
\item  We design a space-, and time- efficient approximation algorithm that performs best in the presence of small negative weights.   In the case of existence of large negative weights, we design a well-performing heuristic.
\smallskip 
\item We provide an experimental evaluation of our proposed methods on synthetic datasets that illustrate the effect of the parameters in our objective. This understanding allows the practitioner to choose the values of such parameters according to the desired goals of his/her application. 
\smallskip 
\item We deploy our developed primitives on the two real-world applications we introduce. We extract subgraphs from uncertain graphs with high expected induced weight and low risk. Finally, we mine Twitter data by finding dense subgraphs that exclude certain types of interactions. A non-trivial experimental contribution is the creation of an uncertain graph from the TMDB database, and  Twitter graphs from the Greek Twitter-verse.  Our algorithmic tools provide insights, and we believe that they will find  more applications in graph mining, and anomaly detection.  
\squishend

\spara{Notation.} We use the following notation. Let $deg^{+}(u)>0$ ($deg^{-}(u)$) be the positive (negative) degree of node $u$. Therefore, the total degree of $u$ is $d(u) = deg^{+}(u)-deg^{-}(u)$. 
Let $w^+(e)$ ($w^+(e)$) be the positive (negative) edge weight. Finally,   $w^+(S)$ ($w^-(S)$) is the total positive (negative) induced weight by node set $S$, and $d_S(u)=deg_S^{+}(u)-deg_S^{-}(u)$ is the total degree of node $u$ within $ S \subseteq V$.

\section{Related Work}
\label{sec:related}
\spara{Uncertain graphs} model naturally a wide variety of datasets and applications  including protein-protein interactions   \cite{asthana2004predicting,krogan2006global}, kidney exchanges \cite{roth2004kidney},  influence maximization \cite{kempe2003maximizing}, and privacy-applications \cite{boldi2012injecting}.   While a lot of research work has focused on designing graph mining algorithms for uncertain graphs 
\cite{bonchi2014core,huang2016truss,khan2015uncertain,kollios2013clustering,liu2012reliable,moustafa2014subgraph,parchas2014pursuit,potamias2010k}, there is less work on designing efficient risk-averse optimization algorithms, and even lesser with solid theoretical guarantees.

Risk-aversion has been implicitly discussed  by Lin  et al.  in their work on reliable clustering \cite{liu2012reliable}, where the authors show that interpreting probabilities as weights does not result in good clusterings.  Repetitive sampling from a large-scale uncertain graph in order to reduce the risk is inefficient. Motivated by this observation,  Parchas et al.  have proposed a heuristic to extract a good possible world in order to combine risk-aversion with efficiency  \cite{parchas2014pursuit}. However,  their work comes with no guarantees.  Jin et al. provide a risk-averse algorithm for distance queries on uncertain graphs \cite{jin2011discovering}.  He and Kempe propose robust algorithms for the influence maximization problem \cite{HeKempe}. Since then, various extensions have been proposed 
for the same problem  \cite{chen2016robust,wilder2017uncharted}. Closest to our work lies the recent work by Tsourakakis et al. who proposed efficient approximation algorithms for finding risk-averse heavy matchings in uncertain graphs and hypergraph \cite{tsourakakis2018risk}.

\spara{Dense subgraph discovery (DSD)} is a major topic of research in the fields of graph algorithms and graph mining, with many diverse applications, ranging from fraud detection to bioinformatics, see \cite{gionis2015dense} for a detailed account of such applications. Finding cliques \cite{karp}, or optimal quasi-cliques \cite{tsourakakis2013denser,tsourakakis2014fennel,tsourakakis2015streaming} is the prototypical DSD formulations but not only they are \NPhard problems,  but also hard to approximate \cite{hastad}. On the contrary, the {\em densest subgraph problem} (DSP) is solvable in polynomial time \cite{GGT89,goldberg84}. The DSP for undirected, weighted graphs $G(V,E,w), w:E \rightarrow \field{R}^+$ maximizes the degree density $\rho(S)=\frac{w(S)}{|S|}$ over all possible subgraphs $S\subseteq V$, where $w(S)=\sum_{e\in e[S]}w(e)$ is the total induced weight by subgraph.  In addition to the exact algorithm that is based on maximum flow computation, 
Charikar \cite{Char00} proved that the greedy algorithm proposed by Asashiro et al.
\cite{AITT00} produces a $\frac{1}{2}$-approximation of the densest subgraph in linear time. 
Both algorithms are efficient in terms of running times and scale to large networks. 
In the case of directed graphs, the densest subgraph problem is solved in polynomial 
time as well. Charikar \cite{Char00} provided a linear programming approach 
which requires the computation of $n^2$ linear programs
and a $\frac{1}{2}$-approximation algorithm which runs in $O(n^3+n^2m)$ time.
Khuller and Saha \cite{Khuller} improved significantly the state-of-the art by providing
an exact combinatorial algorithm and a fast $\frac{1}{2}$-approximation algorithm 
which runs in $O(n+m)$ time. Since then, many variations of the densest subgraph problem have been proposed in the literature.  Tsourakakis generalized the DSP the  the $k$-clique DSP that maximizes the average density of $k$-cliques, and also provided efficient exact and approximation algorithms \cite{tsourakakis2015kclique}, see also \cite{mitzenmacher2015scalable}. Another interesting set of variations of the DSP across a set of graphs was introduced by Semertzidis et al. \cite{semertzidis2016best}, and was analyzed further by Charikar et al. \cite{charikar2018finding}. Finally, the densest subgraph problem with exclusion queries on multilayer graphs has not been  considered before. Galimberti et al. studied core decompositions -- a concept intimately connected to DSD-- on multilayer graphs \cite{galimberti2018core}. Finally, Cadena et al.  first studied DSD with negative weights~\cite{cadena2016dense}, but their work focuses on anomaly detection, and the streaming nature of their input. 

DSD on uncertain graphs is a less well studied topic. Zou  was the first who discussed the DSP on uncertain graphs. His work  shows --as expected-- that the DSP in expectation can be solved in polynomial time \cite{zou2013polynomial}. The closest work related to our formulation is the recent work by Miyauchi and Takeda \cite{miyauchi2018robust}. While their original motivation is also DSD on uncertain graphs, the modeling assumptions, and the mathematical objective differ significantly from ours. To the best of our knowledge, there is no work on risk-averse DSD under general probabilistic assumptions as ours. 

\section{Proposed Method}
\label{sec:proposed}
\subsection{Why Negative Weights?}
\label{sec:motivation} 

\spara{Risk-averse dense subgraph discovery.} Uncertain graphs model the inherent uncertainty associated with graphs in a variety of applications, that we discussed earlier in detail, see Section~\ref{sec:related}. Here, we adopt the general model for uncertain graphs introduced by Tsourakakis et al. \cite{tsourakakis2018risk}. For completeness we present it in the following. 

\underline{Model:} Let  $\mathcal{G}([n], E, \{ f_e(\theta_e) \}_{e \in E})$ be an uncertain complete  graph on $n$ nodes,  with the complete edge set $E={[n] \choose 2}$. The weight $w(e)$ (reward) of each edge $e \in E$ is drawn according to some probability distribution $f_e$ with parameters $\vec{\theta_e}$, i.e., $w(e) \sim f_e(x;\vec{\theta_e})$.   We assume that the weight of each  edge is drawn independently from the rest;  each probability distribution is assumed to have finite mean, and finite variance.  Given this model, we define the probability/likelihood of a given graph $G$ with weights $w(e)$ on the edges as:

\begin{equation}
\label{eq:model} 
\Prob{G;\{ w(e)\}_{e \in E}} = \prod_{e \in E} f_e(w(e); \vec{\theta_e}).
\end{equation}

This model includes the standard Bernoulli model that is used extensively in the existing literature as a special case.  Specifically,  in the standard binomial uncertain graph model  an uncertain graph is modeled by the triple $\mathcal{G} = (V,E,p)$   where $p: E \rightarrow (0,1]$  is the function that assigns a probability of success to each  edge independently from the other  edges.  According to the possible-world semantics \cite{bollobas2007phase,dalvi2007efficient} that interprets  $\mathcal{G} $ as a set $\{G: (V,E_G)\}_{E_G \subseteq E}$ of $2^{|E|}$ possible deterministic graphs (worlds), each defined by a subset of $E$.  The probability of observing any possible world $G(V,E_G) \in 2^{E}$ is 
$$\Prob{G} = \prod\limits_{e \in E_G} p(e) \prod\limits_{e \in E\backslash E_G} (1-p(e)).$$  
 
A key observation to hold in mind, is that each edge $e$ in the uncertain graph is independently distributed from the rest and is associated with an expected reward $\mu_e$ (expectation)  and a risk $\sigma_e^2$ (variance). Finally, observe that without any loss of generality in our general model described by equation~\eqref{eq:model} we have assumed that the edge set is ${[n] \choose 2}$; non-edges can be modeled as edges with probability of existence zero.  

\begin{algorithm}[ht]
\caption{\label{alg:exclusion} \tt{Exclusion-Queries}$(G(V,E), \{\text{colors}\}, W>0)$ }
\begin{algorithmic}  
\FOR{$e \in E(G)$} 
\FOR{$c \in \text{colors}$} 
\IF{If $type(e)=c$}  
\STATE{$w(e) \leftarrow -W$} (else $w(e)$ remains 1)
\ENDIF
\ENDFOR
\ENDFOR
\STATE{Return  $S \subseteq V$ that achieves maximum average degree in $G(V,E,w)$.}
\end{algorithmic}
\end{algorithm}

\underline{Problem formulation.} Intuitively, our goal is to find a subgraph  $G[S]$ induced by $S \subseteq V$ such that its average expected reward $\frac{\sum\limits_{e \in E(S)} w_e}{|S|}$ is large and the  associated average risk is low  $\frac{\sum\limits_{e \in E(S)} \sigma_e^2}{|S|}$.   To achieve this purpose we model the problem as a densest subgraph discovery problem in a graph with positive (reward) and negative (risk) edge weights. Specifically, for every edge $e=(u,v) \in E(G)$ we create two edges, a positive edge with weight equal to the expected reward, i.e., $w^+(e)=\mu_e$ and a negative edge with weight equal to the opposite of the risk of the edge, i.e., $w^-(e)=\sigma_e^2$.   We wish to find a subgraph $S \subseteq V$ that has large positive average degree $\frac{w^{+}(S)}{|S|}$, and small negative average degree $\frac{w^{-}(S)}{|S|}$.  We combine the two objectives into one objective $f:2^V \rightarrow \field{R}$ that we wish to maximize: 

$$f(S) = \frac{w^{+}(S)+\lambda_1|S|}{w^{-}(S)+\lambda_2|S|}.$$

\noindent The parameters  $\lambda_1,\lambda_2 \geq 0$ are positive reals. First, observe that this dense subgraph discovery formulation is applicable to any graph with positive and negative weights. Parameters $\lambda_1,\lambda_2$ allow us to control the size of the output as follows. Let us reparameterize the two parameters as $\lambda_1 = \rho \lambda, \lambda_2=\lambda$. Then  $f(S) = \frac{w^{+}(S)+\rho \lambda |S|}{w^{-}(S)+\lambda|S|}$, so if the ratio $\rho \geq 1$, then the objective favors larger node sets, whereas when $\rho<1$ we favor smaller node sets. 

We show how to solve the problem $\max\nolimits_{S \subseteq V} f(S)$ by reducing it to standard dense subgraph discovery \cite{lawler2001combinatorial,goldberg84}. We perform binary search on $f(S)$ by answering queries of the following form:

\begin{center}
\fbox{\begin{varwidth}{\dimexpr\textwidth-2\fboxsep-2\fboxrule\relax}
\begin{quote}
Does there exist a subset of nodes $S \subseteq V$ such that $f(S) \geq q$, where $q$ is a query value? 
\end{quote}
\end{varwidth}}
\end{center}

Assuming an efficient algorithm for answering this query, and that the weights are polynomial functions of $n$, then using $O(\log n)$ queries we can find the optimal value for our objective $f:V \rightarrow \field{R}$.   By analyzing what each query corresponds to, we find:

\begin{align}
\label{eq:dsdreduction}
\frac{w^{+}(S)+\lambda_1|S|}{w^{-}(S)+\lambda_2|S|} &\geq q  \rightarrow   w^{+}(S)+\lambda_1|S| \geq q (w^{-}(S)+\lambda_2|S|)   \rightarrow \\   \nonumber 
\sum_{e \in E(S)} \underbrace{\bigg(w^{+}(e)-qw^{-}(e)\bigg)}_{\tilde{w}(e)} &\geq |S| \underbrace{(q \lambda_2 - \lambda_1)}_{q'} \rightarrow \sum_{e \in E(S)}  \frac{\tilde{w}(e)}{|S|} \geq q'.
\end{align}

The latter inequality suggests that our original problem corresponds to querying in $\tilde{G}$ --a modified version of  $G$ where the edge weight of any edge $e$ becomes $w^{+}(e)-qw^{-}(e)$-- 
whether there exists a subgraph $S$  with density greater than $q'$, where $q'=q\lambda_2-\lambda_1$. However, this does not imply that our problem is poly-time solvable. The densest subgraph problem is poly-time solvable using a maximum flow formulation  
when the weights are positive rationals 
 \cite{goldberg84}. As we will prove in the next section,  the densest subgraph problem when there exist negative weights is \NPhard in general. However, our analysis above leads to a straight-forward corollary that is worth stating.  Intuitively, when for {\em each edge} $e$ the ratio $\frac{w^+(e)}{w^-(e)}$ is large enough, then our problem is solvable in polynomial time.

\begin{corollary} 
\normalfont
Assume that $w^{+}(e) \geq q_{max} w^{-}(e)$ for all $e \in E^{+} \cup E^{-}$, where $q_{max}$ is the maximum possible query value. Then, the densest subgraph problem is solvable in polynomial time. 
\end{corollary}

\begin{proof}
If $w^{+}(e) \geq q_{max} w^{-}(e)$ for each $e \in E$, we obtain $\tilde{w}(e) \geq 0$ for each $e \in E$ in  inequality~\eqref{eq:dsdreduction} is equivalent to solving the densest subgraph problem in an undirected graph with non-negative weights, see \cite{goldberg84,tsourakakis2015kclique}.     
\end{proof}

Observe that a trivial upper bound of $q_{\max}$ can be obtained by setting $w^+(S) = \sum_{e \in E(G)} w^+(e), w^-(S)=0$, and since $\lambda_1|S| \leq \lambda_1 n, \lambda_2|S| \geq \lambda_2$ for all $S \neq \emptyset$, we see that $q_{\max} \leq \frac{\sum_{e \in E(G)} w^+(e)+\lambda_1 n}{\lambda_2}$. For polynomially bounded weights, this is a polynomial function of $n$, hence the number of  binary search iterations is logarithmic.

\underline{Controlling the risk in practice.} There exist real-world scenarios where the practitioner wants to control the trade-off between reward and risk, see \cite{tsourakakis2018risk}. An effective way to change the risk tolerance is as follows  by multiplying the negative induced weight $w^-(S)$ by  $B  \in (0,+\infty)$.  Namely, our objective $f:2^{V} \rightarrow \field{R}$ is  
$f(S) = \frac{w^{+}(S)+\lambda_1|S|}{Bw^{-}(S)+\lambda_2|S|}.$
An interesting open problem is to develop a formal  (bi-criteria) approximation for risk averse DSD along the lines of  \cite{ravi1996constrained,tsourakakis2018risk}.

\spara{Soft and hard exclusion dense subgraph queries.}  Given the Twitter network, where user accounts may interact in more than one ways (e.g., {\it follow, retweet, mention, quote, reply}),  can we find a dense subgraph that does not contain any {\em follow} but contains many {\em  reply} interactions?   We ask this question in a more general form.  

\begin{tcolorbox}
\begin{problem}
Given a large-scale multilayer network, how do we find a dense subgraph that {\em excludes} certain types of edges?
\end{problem}
\end{tcolorbox}

\noindent We consider two types of such queries. The {\em soft} and {\em hard} queries. In the former case we want to find subgraphs with perhaps few  edges of certain types, in the latter case we want to exclude fully such edges.   An algorithmic primitive  that can answer efficiently these queries can be used to understand the structure of large-scale multilayer networks,  and find anomalies and interesting patterns.  As a result,  subgraphs that do not induce any edge of any excluded type will have positive weight, whereas subgraphs that induce even one edge of a forbidden type will have $-\infty$ weight.   In principle, we set the edge weight of an excluded type to $-W$ where $W>0$ is an input parameter.   The pseudo-code in Algorithm shows this approach.   
Again, dense subgraph discovery with negative weights plays the key role in developing such a graph primitive. In practice, a practitioner may range $\kappa$ from small to large values.

\subsection{Hardness} 
\label{subsec:hardness}

We prove that solving the densest subgraph problem on graphs with negative weights is \NPhard. We formally define our problem {\sc Neg-DSD}.

\begin{tcolorbox}
\begin{problem}[Neg-DSD]
\label{negdsd}
Given a graph $G$ with loops and possibly negative weights, find the subset $A$ of $V$ that maximizes $\frac{w(A)}{|A|}$. 
\end{problem}
\end{tcolorbox}

\noindent We prove that {\sc Neg-DSD} is \NPhard. Our reduction is based on the  the proposed strategy by Peter Shor for showing that the max-cut problem on graphs with possibly negative edges is \NPhard \cite{csexchange}. This is stated as the Theorem~\ref{thm:hardness}.

\begin{theorem}
\label{thm:hardness}
 \normalfont
{\sc Neg-DSD} is \NPhard. 
\end{theorem}

For convenience, we define the decision version of the maximum cut problem \cite{csexchange}.  

\begin{problem*}[Max-Cut] 
Given a graph $G(V,E)$ and a constant $c$, find a partition $(A, B)$ of $V$ such that  $cut(A, B) > c$. 
\end{problem*}  

Our proof strategy is inspired by Peter Shor's proof that max-cut with negative weight edges is \NPhard \cite{csexchange}. We provide a detailed proof sketch of Theorem~\ref{thm:hardness}.

\begin{proof} 
 \normalfont
First, we define the {\sc Positive-Cut} problem, and show that it is \NPhard  by reducing the {\sc Max-Cut} problem to it. 

\begin{problem*}[Positive-Cut] 
Given a graph $G$ with possibly negative weights, find a partition $(A, B)$ of $V$ such that  $cut(A, B) > 0$. 
\end{problem*} 

We choose two nodes  $u,v$ that lie on opposite sides of an optimal max cut $(A^*,B^*)$.  Despite the fact we do not know the max cut, we can perform this step in polynomial time by repeating the following procedure for all possible pairs of nodes; if we cannot find a positive cut for any of the pairs, then the answer to the {\sc Max-Cut} is negative. We construct a graph $G'$ by adding a very large negative weight equal to $-d$ from $u$ and $v$ to all other vertices, and an edge of weight $(n-2)d-c$ between $u,v$. All cuts that place $u,v$ on the same side will be negative in $G'$ provided $d$ is  sufficiently large. All other cuts will be positive if and only if the corresponding cut in $G$ is greater than $c$. Therefore, {\sc Positive-Cut} is \NPhard.  

Finally we prove that {\sc Neg-DSD} is \NPhard using a reduction from {\sc Positive-Cut}.  We construct a graph $G'$ by negating every weight in $G$ putting a loop on every vertex so that its weighted degree is zero. Hence the sum of the degrees of any set $A$  in $G'$ is equal to $0=\sum_{v \in S} 0 = 2w(A)+cut(A,\bar{A})$. Observe that a cut $(A, B)$ has positive weight in $G$ if and only if $A$ has positive average degree.  This completes the proof. 
\end{proof}

\subsection{Algorithms and Heuristics} 

A popular algorithm for the densest subgraph problem is Charikar's algorithm \cite{Char00}.  We study the performance of this algorithm in the presence of negative weights. The pseudocode is given as 
Algorithm~\ref{alg:peeling}. The algorithm iteratively removes from the graph the node of the smallest degree $d(v)=deg^+(v)-deg^-(v)$, and among the sequence of $n$ produced graphs, outputs the one that achieves the highest degree density. Our main theoretical result for the performance of Algorithm~\ref{alg:peeling} is stated as Theorem~\ref{thm:guarantee}.

\begin{algorithm}[ht]
\caption{\label{alg:peeling} \tt{Peeling}$(G)$ }
\begin{algorithmic}
\STATE{$n\leftarrow |V|, H_n \leftarrow G$}
\FOR{$i\leftarrow n$ to $2$} 
\STATE{Let $v$ be the vertex of $G_i$ of minimum degree, i.e.,  $d(v)=deg^{+}(v)-deg^{-}(v)$} (break ties arbitrarily)
\STATE{$H_{i-1} \leftarrow H_i \backslash{v}$}
\ENDFOR
\STATE{Return  $H_j$ that achieves maximum average degree among $H_i$s, $i=1,\ldots,n$.}
\end{algorithmic}
\end{algorithm}

\begin{theorem} 
\label{thm:guarantee}
 \normalfont
Let $G(V,E,w)$, $w:E \rightarrow \field{R}$ be an undirected weighted graph with possibly negative weights. If the negative degree $deg^{-}(u)$ of any node $u$  is upper bounded by $\Delta$, then Algorithm~\ref{alg:peeling} outputs a set whose density is at least $\frac{\rho^*}{2}-\frac{\Delta}{2}$. 
\end{theorem}

\begin{proof} 
 \normalfont
Let $S^*$ be the optimal densest subgraph in $G$ with average density $\frac{w(S^*)}{|S^*|}=\rho^*$. By the optimality of $S^*$ we obtain that $d_{S^*}(v) \geq  \rho^*$, and then trivially $deg^{+}(v)\geq \rho^*$. Consider the execution of  algorithm~\ref{alg:peeling}, and let $u \in S^*$ be the first vertex from  $S^*$ removed during the peeling. Let $S$ be the set of nodes at that iteration, including $u$. By the peeling process, we have $d_S(v)\geq d_S(u)$ for all $v \in S$. 
Furthermore,  

$$d_S(u) = deg_S^{+}(u)-deg_S^{-}(u)\geq \deg_S^{+}(u) - \Delta,$$ 

\noindent since by our assumption $deg_S^{-}(u)\leq deg^{-}(u) \leq \Delta$.   This implies that 

\begin{align*}
2w(S) &= \sum\limits_{v \in S} d_S(v) \geq \sum\limits_{v \in S} deg_S^{+}(v) - |S| \Delta \geq |S| (\rho^*-\Delta) \rightarrow 
\frac{w(S)}{|S|}  \geq \frac{\rho^*}{2}-\frac{\Delta}{2}.
\end{align*}

This yields that the output of Algorithm~\ref{alg:peeling} outputs a subgraph $H$ with degree density at least $ \frac{\rho^*}{2}-\frac{\Delta}{2}$. 
\end{proof} 

When the additive error term in the approximation is small compared to the term $\frac{\rho^*}{2}$, then the peeling algorithm performs effectively. In practice, Algorithm~\ref{alg:peeling} performs well on large-scale graphs where the negative weights are small. In the presence of  large negative degrees, the approximation guarantees become less meaningful, or even meaningless. 

\begin{claim*} 
 \normalfont
In the presence of large negative weights, Algorithm~\ref{alg:peeling} may perform arbitrarily bad. 
\end{claim*} 

This is illustrated in  Figure~\ref{fig:counter}($\alpha$) that provides a bad graph instance with $n+4$ nodes for our proposed algorithm.  Let $W = \frac{n-4}{3}$. Then, $3W-n < -3$.  The degrees of  the $n+4$ nodes are as follows: 

$$\underbrace{3W-n}_{\text{one node}}<\underbrace{-3}_{n-2 \text{~nodes}}<\underbrace{-2}_{\text{two nodes}}< 0 < \underbrace{2\epsilon+W}_{\text{three nodes}}.$$

\begin{figure*}[htp]
\centering
\centering
\begin{tabular}{@{}c@{}@{\ }c@{}@{\ }c@{}}
\includegraphics[width=0.43\textwidth]{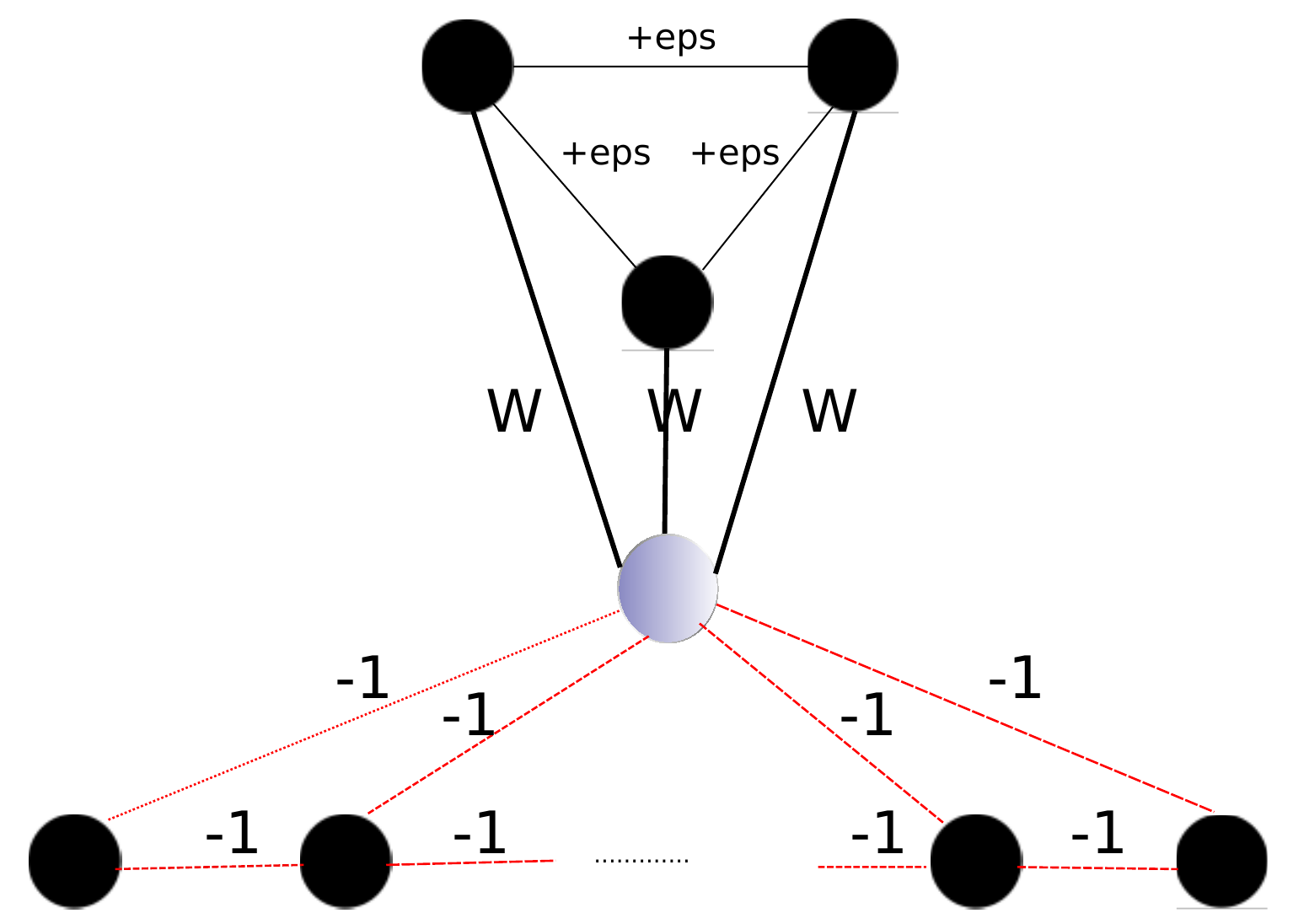}  & & \includegraphics[width=0.23\textwidth]{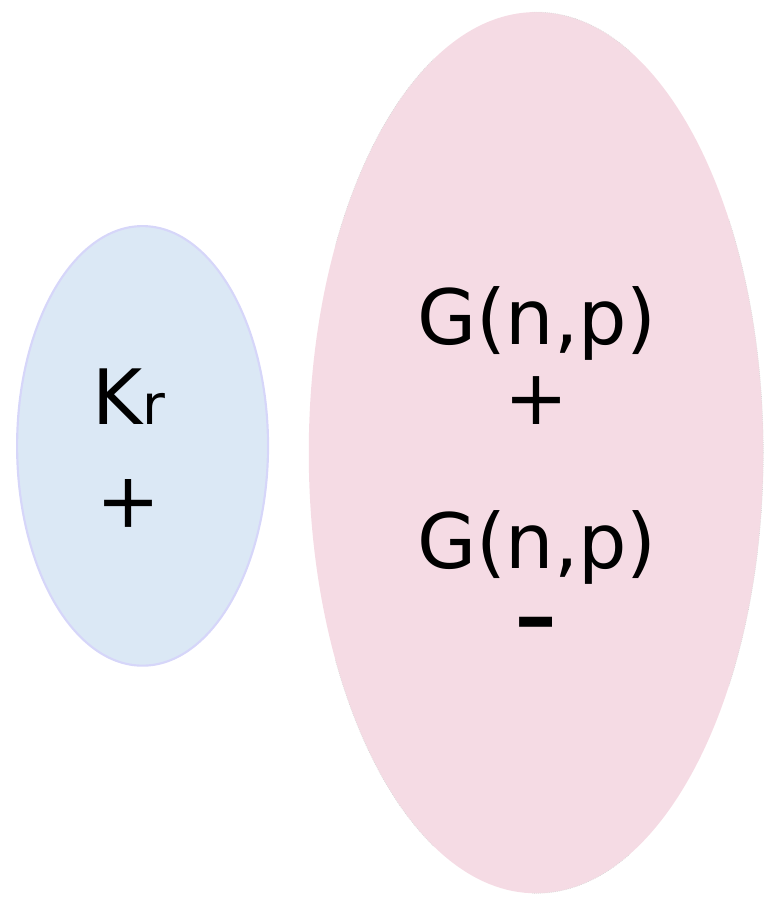} \\ 
($\alpha$) & & ($\beta$) \\ 
\end{tabular}
\caption{\label{fig:counter} Bad peeling instances. For details, see Section~\ref{sec:proposed}.}
\end{figure*} 

\noindent Therefore, the center node is removed first, and the peeling algorithm will output as the densest subgraph the triangle of density $\epsilon$. The optimal densest subgraph has $\frac{3W+3\epsilon}{4}$. By allowing  $\epsilon$ to be arbitrarily small, we observe that the approximation ratio becomes arbitrarily bad. 
To tackle such scenarios, i.e., where nodes from the densest subgraph are peeled earlier than when they should,  we propose an effective heuristic which is outlined in Algorithm~\ref{alg:heuristic}.  The algorithm again peels the nodes but scores every node $u$ according to $Cdeg^+{(u)}-deg^-{(u)}$, where $C>0$ is a parameter that is part of the input.

\spara{Remark about $C$ in Algorithm~\ref{alg:heuristic}.} While Figure~\ref{fig:counter}($\alpha$) suggests the use of $C \geq 1$, it could be the case that $C$ has to be set to a value less than 1 to obtain good results. We provide an example where using $C<1$ can help in providing a better peeling permutation of the nodes. Consider a graph whose weights are either $+1$ or $-1$,  that consists of two connected components. The first component is a positive clique on $r$ nodes. The second component is the union of two random binomial graphs   $G(n,p)$ where $p=\frac{1}{2}$. This is illustrated in Figure~\ref{fig:counter}($\beta$).  The degree of any node $u$ in the first component is $deg(u)=deg^+(u)-deg^-(u)=(r-1)-0$. The expected degree of any node in the second component is 0. Furthermore, the average degree of any subset of nodes in the 2nd component is 0 in expectation. However, using concentration bounds  (details omitted) one can  show that it is likely that there will exist a node $u$ in the second component with positive degree  $\kappa \sqrt n$ and negative degree  $\kappa' \sqrt n$  with $\kappa>\kappa'$, and therefore positive total degree. Only the use of a $C<1$ will improve  the peeling ordering; for example one can immediately see that in the extreme case where $C=0$ the nodes of the second component will be removed first. 

\spara{Rule-of-thumb.} In practice, given that each run of the algorithm takes linear time, we can afford to run  the algorithm for a bunch of $C$ values and return the densest subgraph among the outputs produced by each run, instead of using one value for $C$. This strategy is applied in Section~\ref{sec:exp}.

\begin{algorithm}[ht]
\caption{\label{alg:heuristic} \tt{Heuristic-Peeling}$(G,C)$ }
\begin{algorithmic}
\REQUIRE{$C \in (0,+\infty)$}
\STATE{$n\leftarrow |V|, H_n \leftarrow G$}
\FOR{$i\leftarrow n$ to $2$} 
\STATE{Let $v$ be the vertex of $G_i$ of minimum degree, i.e.,  $d(v)=Cdeg^{+}(v)-deg^{-}(v)$} (break ties arbitrarily)
\STATE{$H_{i-1} \leftarrow H_i \backslash{v}$}
\ENDFOR
\STATE{Return  $H_j$ that achieves maximum average degree among $H_i$s, $i=1,\ldots,n$.}
\end{algorithmic}
\end{algorithm}

\spara{Shifting the negative weights.} Finally, for the sake of completeness, we mention that the perhaps natural idea of shifting all the weights by the most negative weight in the graph, in order to obtain non-negative weights, and apply the exact polynomial time algorithm on the weight-shifted graph may perform arbitrarily bad. To see why, consider a graph on $n+5$ nodes that consists of three components, a triangle with positive weights equal to 1, an edge with a large negative weight $-\Delta<0$, and a large clique on $n$ nodes, whose each edge weight is equal to $-\epsilon<0$. In this graph, the densest subgraph is the positive triangle. However, shifting the weights by $+\Delta$, the degree density of the triangle becomes $1+\Delta$, and of the clique $\frac{ {\Delta-\epsilon \choose 2}}{n}$. For large enough $\Delta$, assuming $\epsilon$ is negligible, the densest subgraph is the clique whose true degree density is negative.  Also experimentally, this heuristic performs extremely poorly.

\section{Experimental results}
\label{sec:exp}

\begin{figure*}[!htp]
\centering
\begin{tabular}{@{}c@{}@{\ }c@{}@{\ }c@{}}
\includegraphics[width=0.33\textwidth]{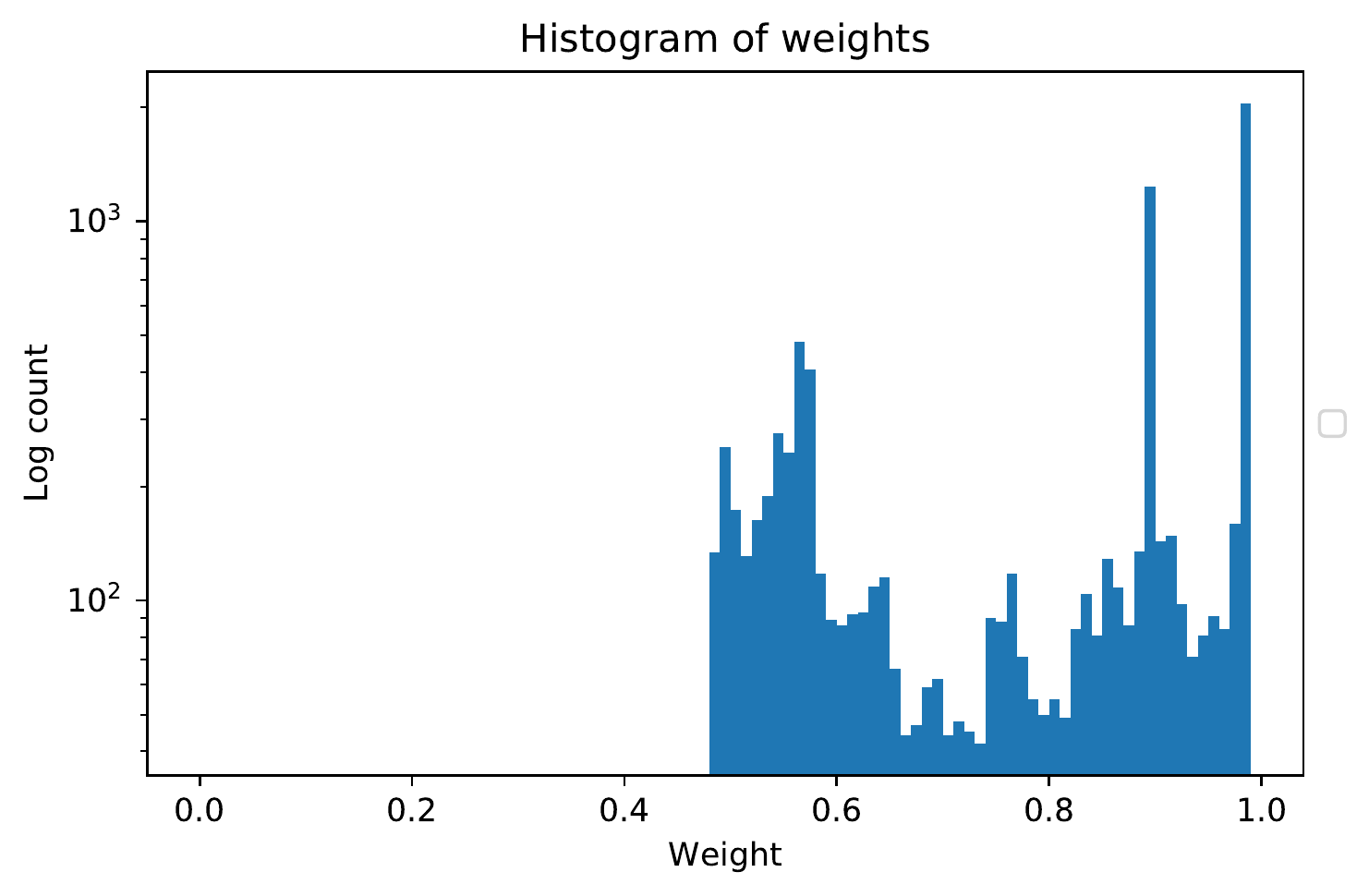} & \includegraphics[width=0.33\textwidth]{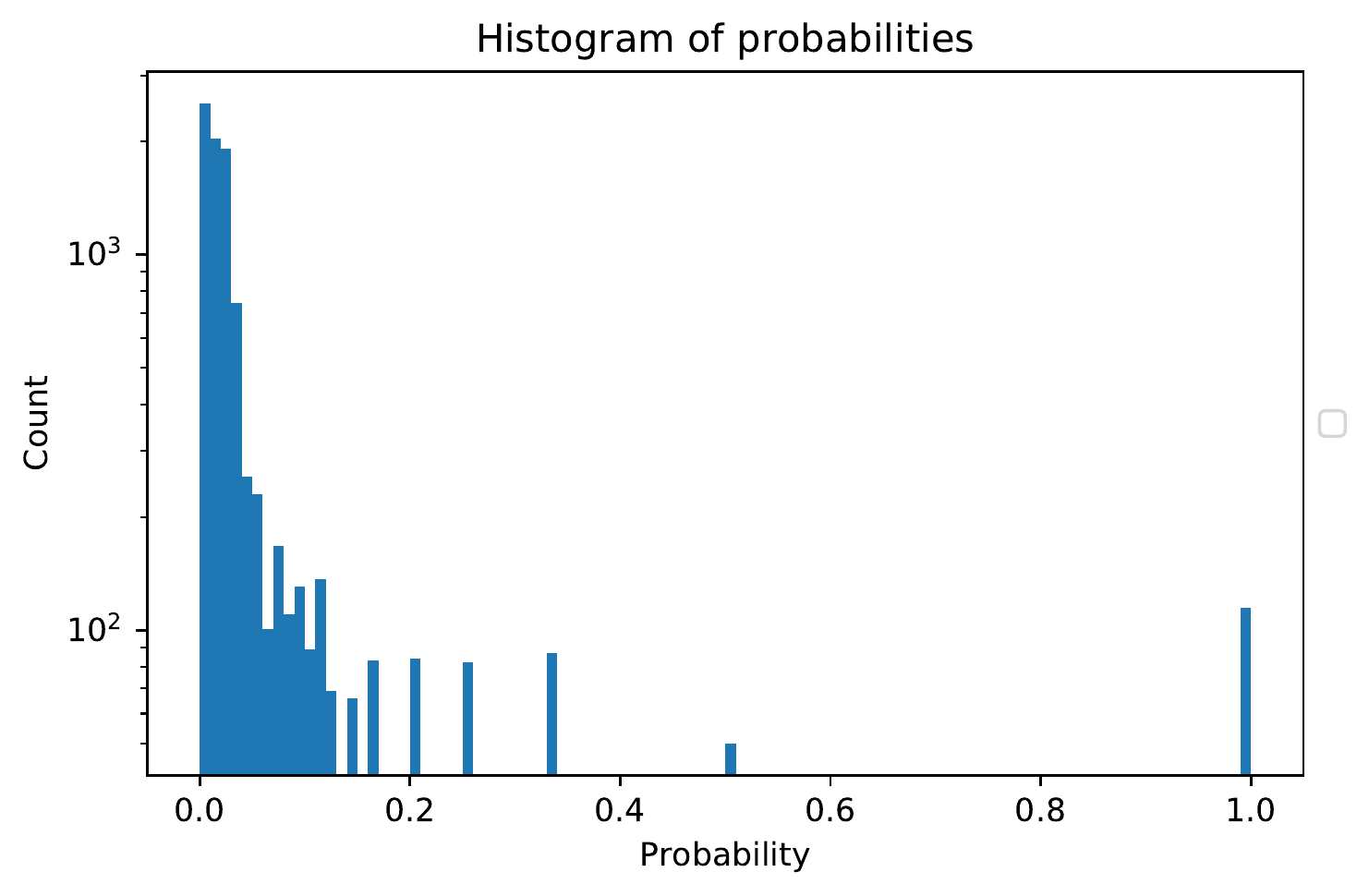}    & \includegraphics[width=0.33\textwidth]{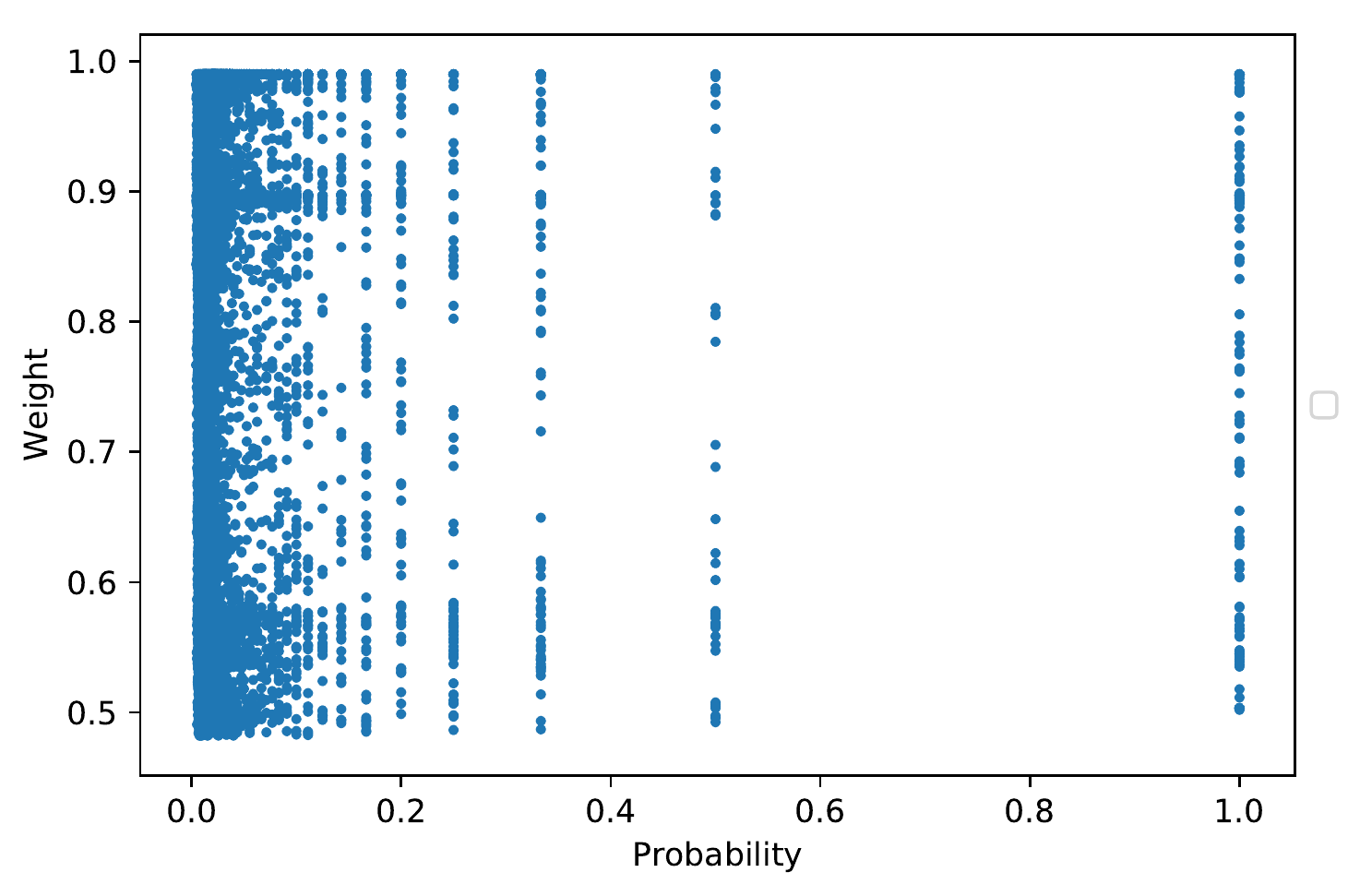}       \\
($\alpha$) & ($\beta$) & ($\gamma$) \\
\includegraphics[width=0.33\textwidth]{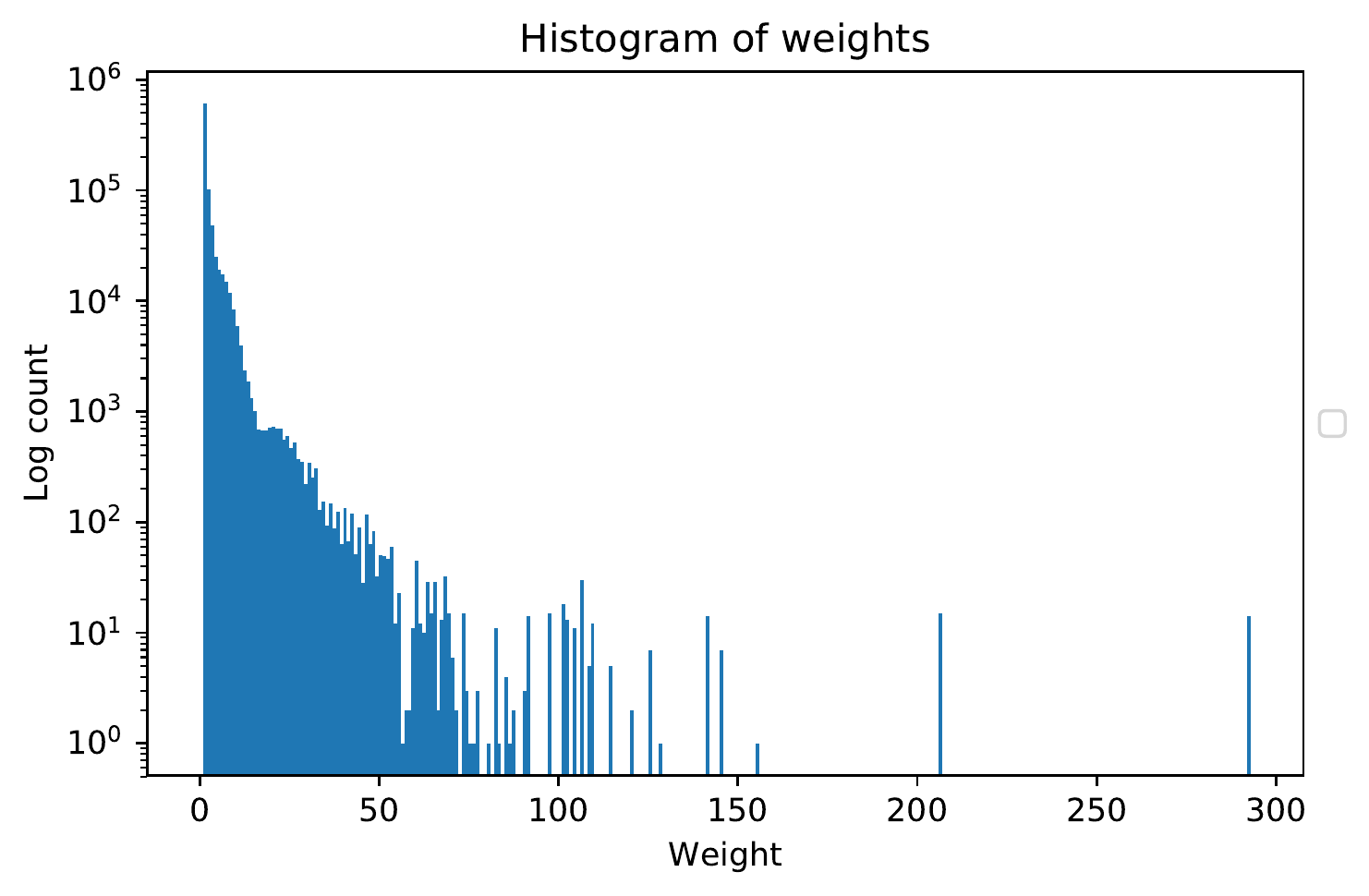} & \includegraphics[width=0.33\textwidth]{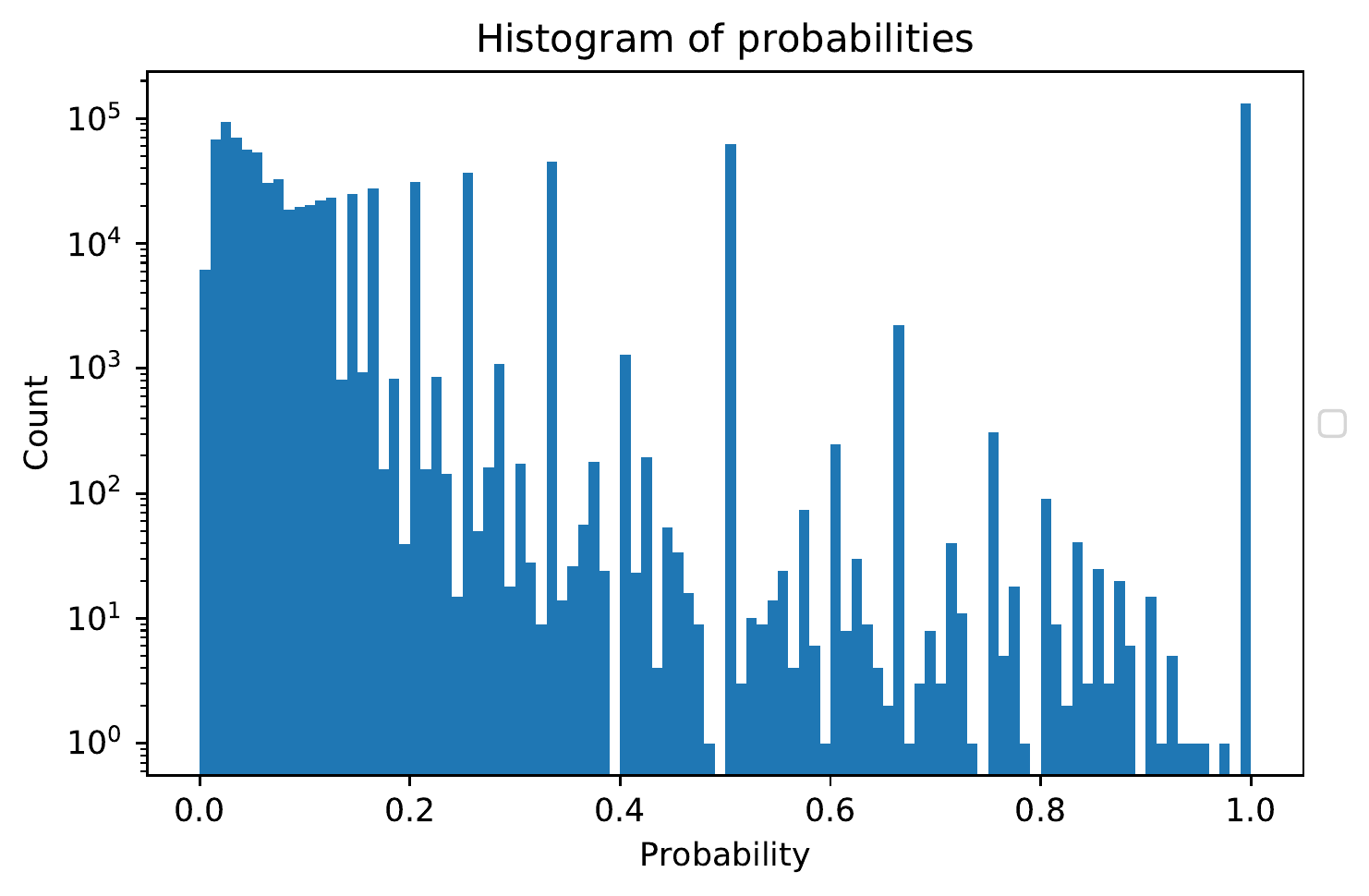}    & \includegraphics[width=0.33\textwidth]{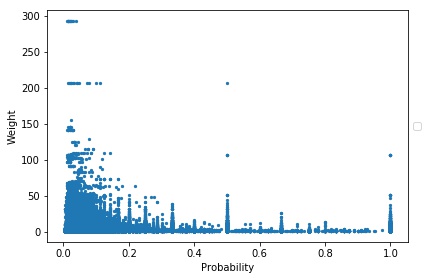}       \\
($\delta$) & ($\epsilon$) & ($\sigma\tau$) \\ 
\end{tabular}
\caption{\label{fig:uncertaingraphs} Uncertain graph statistics. First and second rows correspond to {\em Collins} and {\em TMDB} datasets respectively.  ($\alpha$),($\delta$) Log histogram of weights. ($\beta$),($\epsilon$) Log histogram of edge probabilities. ($\gamma$),($\sigma\tau$) Scatterplot of weights vs. edge probabilities.}
\end{figure*}

\begin{table*}[!ht]
\begin{center}
\begin{tabular}{|l|c|c|} \hline
   Name            & $n$          & $m$   \\       \hline
 \textcolor{green}{$\blacksquare$}  Biogrid   &  5\,640 & 59\,748  \\ 
 \textcolor{green}{$\blacksquare$} Collins          & 1\,622& 9\,074   \\ 
 \textcolor{green}{$\blacksquare$} Gavin         &1\,855 & 7\,669   \\ 
 \textcolor{green}{$\blacksquare$} Krogan core        &2\,708 & 7\,123   \\ 
 \textcolor{green}{$\blacksquare$}  Krogan extended        & 3\,672 & 14\,317   \\  
  \textcolor{cyan}{$\odot$} TMDB           & 160\,784 & 883\,842    \\ \hline 
    \textcolor{cyan}{$\odot$}  Twitter (Feb. 1)  & 621\,617& (902\,834, 387\,597, 222\,253, 30\,018, 63\,062) \\  
  \textcolor{cyan}{$\odot$}  Twitter (Feb. 2)  &706\,104 & (1\,002\,265, 388\,669, 218\,901, 29\,621, 64\,282)  \\  
  \textcolor{cyan}{$\odot$}  Twitter (Feb. 3)  & 651\,109& (1\,010\,002, 373\,889, 218\,717, 27\,805, 59\,503)   \\  
  \textcolor{cyan}{$\odot$} Twitter (Feb. 4)  & 528\,594& (865\,019, 435\,536, 269\,750, 32\,584, 71\,802)  
\\  
  \textcolor{cyan}{$\odot$} Twitter (Feb. 5)  & 631\,697& (999\,961, 396\,223, 233\,464, 30\,937, 66\,968) \\  
  \textcolor{cyan}{$\odot$} Twitter (Feb. 6)  & 732\,852&(941\,353, 407\,834, 239\,486, 31\,853, 67\,374)   \\  
  \textcolor{cyan}{$\odot$} Twitter (Feb. 7)  & 742\,566&(1\,129\,011, 406\,852, 236\,121, 30\,815, 68\,093)  \\  
 \hline
\end{tabular}
\end{center}
\caption{\label{tab:datasets} Datasets used in our experiments. The number of vertices $n$ and edges $m$ 
is recorded for each graph. The datasets annotated by   \textcolor{green}{$\odot$} have been created by us, and are publicly available. The five-dimensional vector containing the number of edges for each day of Twitter correspond to {\em follow, retweet, mention, quote, reply} respectively. For details, see Section~\ref{sec:setup}.}
\end{table*}

\subsection{Experimental setup} 
\label{sec:setup}

\spara{Datasets.} The datasets we have used in our experiments are shown in Table~\ref{tab:datasets}. We use five uncertain graphs, {\it Biogrid, Collins, Gavin, Krogan core, Krogan extended} that have been used in prior biological studies (e.g., \cite{collins2007toward,gavin2006proteome,krogan2006global}), and are available at \cite{paccanarolab}, and one uncertain graph  that we created from the TMDB movie database as follows, and is available at \cite{babis1}. The set of nodes corresponds to actors, and the probability of the edge is equal to the probability that these two actors co-star in a movie. Specifically, for actors $u,v$, the probability $p(u,v)$ is equal to the Jaccard coefficient $J(M_u,M_v) = \frac{ | M_u \cap M_v| }{ | M_u \cup M_v| }$, where $M_u,M_v$ are the sets of movies that $u,v$ have co-starred respectively. We choose weights to represent a function of the popularity of the movies, i.e., a score assigned to each movie by TMDB\footnote{In TMDB the highest score is 10, and the lowest is 1.}. Intuitively, these scores reflect the reward of a potential collaboration between two actors. While there are many ways to set the weight of an edge for two actors (e.g. average popularity), we focus on the most popular movies they have co-starred in. The main rationale behind this choice is that the majority of actors play in movies whose majority popularity  is 1, i.e., the lowest possible. For a pair of actors $\{u,v\}$, let $s_0 \geq \ldots \geq s_{k-1}$ where $k=\min( |M_u \cap M_v|, 5)$ be the popularity scores of movies they have co-starred in. We set $w(u,v) = \sum_{j=0}^{k-1}  \frac{s_j}{2^j}$, i.e., a discounted sum of popularities, focusing more on the most popular movies the  two actors have co-starred in.

\begin{figure*}[!htp]
\centering
\centering
\begin{tabular}{@{}c@{}@{\ }c@{}}
\includegraphics[width=0.43\textwidth]{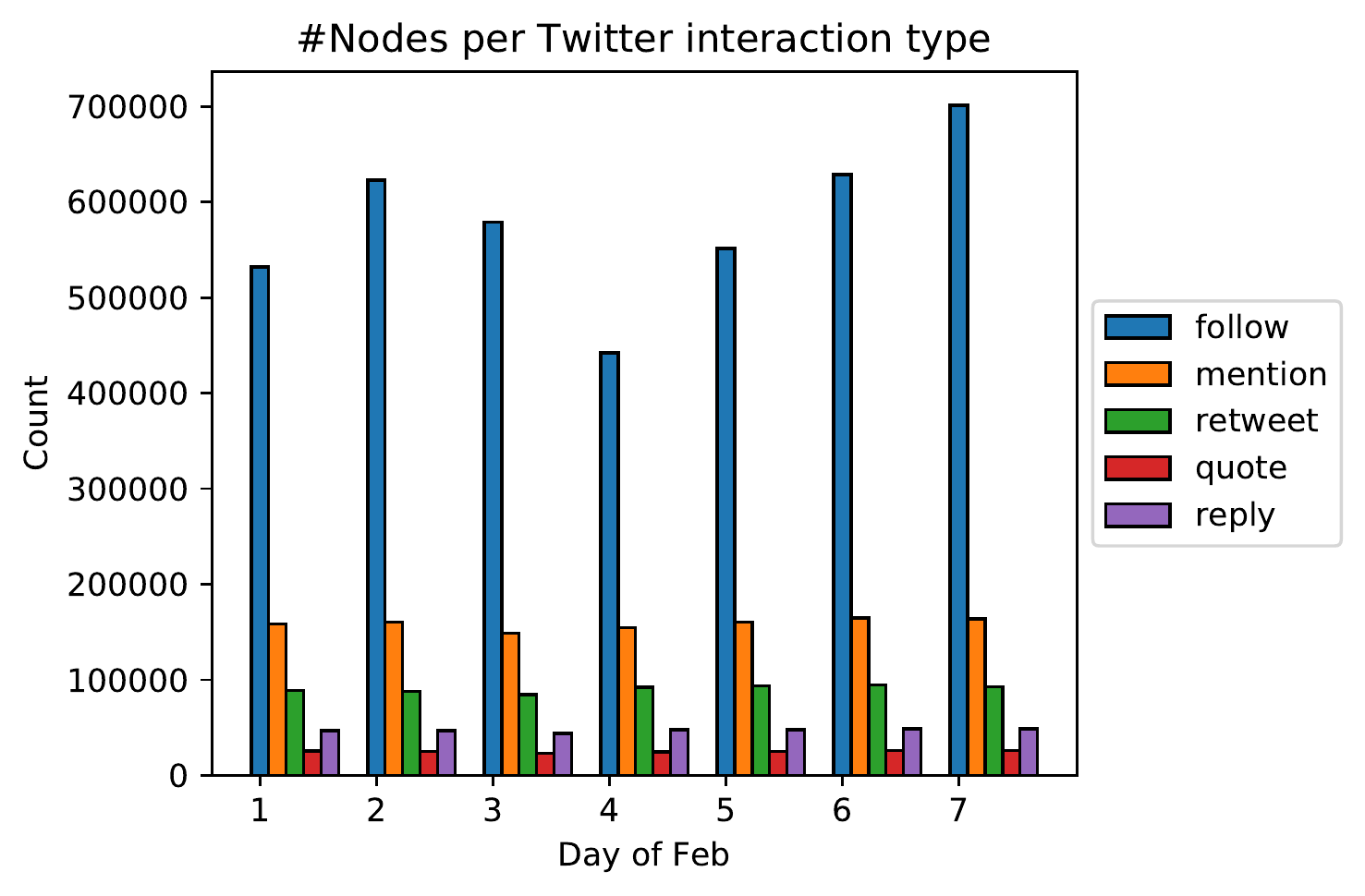}  & \includegraphics[width=0.43\textwidth]{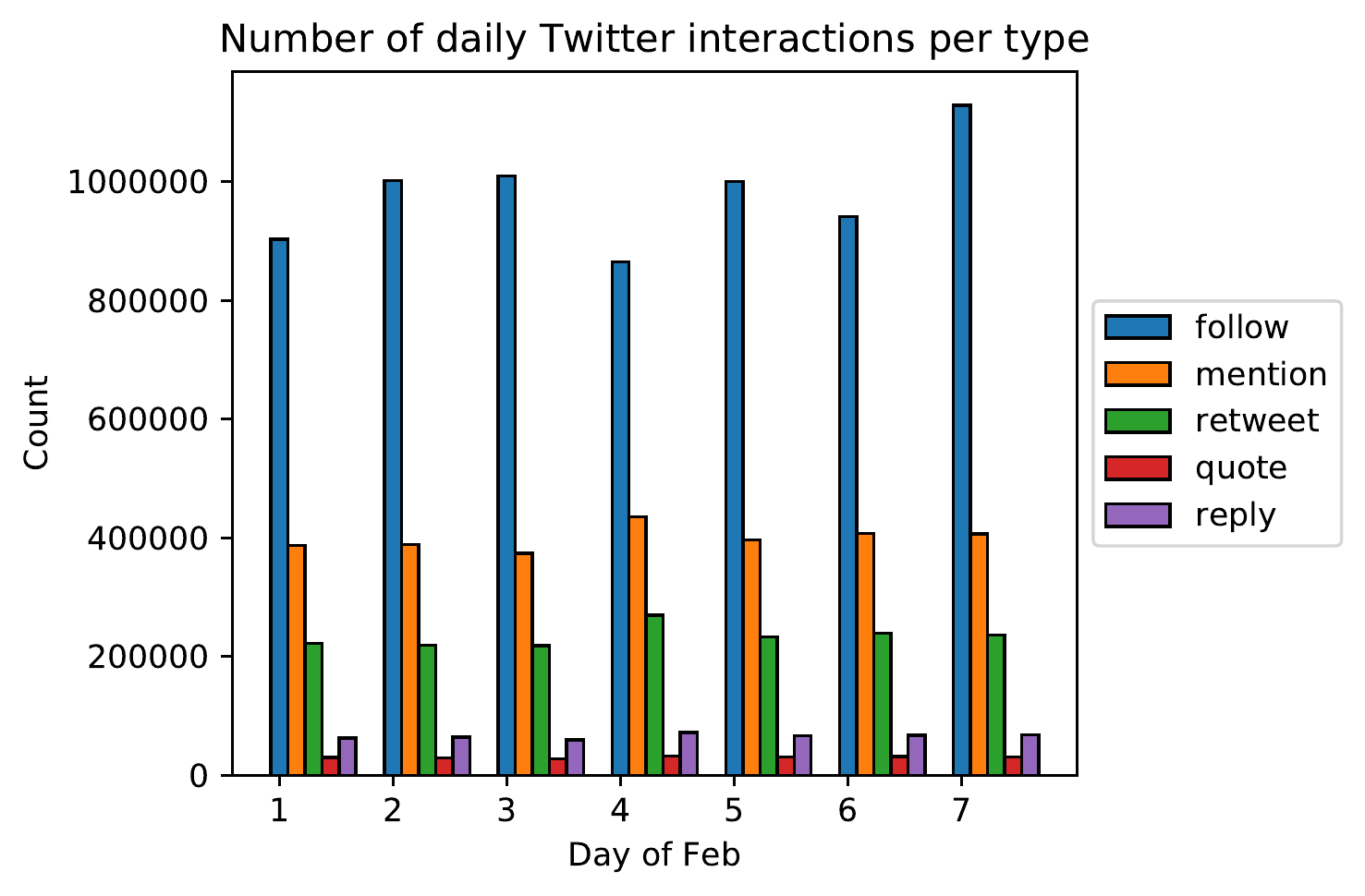} \\ 
($\alpha$) & ($\beta$)  \\ 
\end{tabular}
\caption{\label{fig:twitterstats} ($\alpha$) Count of  Twitter accounts per day, and ($\beta$) count of Twitter interactions for the first week of February 2018,  involved in five types of interactions.}
\end{figure*}

Figure~\ref{fig:uncertaingraphs} provides a detailed view of basic properties of two of the uncertain graphs we use in our experiments. The first and second row correspond to the {\em Collins} and {\em TMDB} datasets respectively. The first and second columns show the histograms of the weights and edge probabilities in log-scale, and the third column provides a scatter-plot of edge weights versus edge probabilities. The same results for the rest of the uncertain graphs appears in Figure~\ref{fig:uncertaingraphs_appendix} in the Appendix~\ref{sec:appendix}. Finally, we used an open-source twitter API crawler to monitor twitter traffic
between February 1st and February 14th, 2018~\cite{pratikakis2018twawler}.  We provide detailed information about each daily graph. Here, the number of edges is a five dimensional vector, whose coordinates correspond to the number of follows, mentions, retweets, quotes, and replies. Figure~\ref{fig:twitterstats}  shows these counts. Specifically, Figure~\ref{fig:twitterstats}($\alpha$) shows the number of Twitter accounts (nodes) involved in five types of Twitter interactions, {\em follow, retweet, mention, quote, and reply} for the first seven days of February 2018. The total number of nodes involved in all interactions is shown in Table~\ref{tab:datasets}. Similarly, Figure~\ref{fig:twitterstats}($\beta$) shows the number of Twitter interactions per type. The {\em follow} interactions are the majority for each day, and the {\em mention} interaction comes second for each day too. The datasets we use are overall small, and medium sized, therefore our proposed algorithm for a fixed $C$ value, requires few seconds or few minutes for the largest graphs. 

\spara{Machine specs and code.} The experiments were performed on a single machine, with an Intel Xeon CPU
at 2.83 GHz, 6144KB cache size, and 50GB of main memory.  The code is written in Python, and is available at \url{https://github.com/negativedsd}.

\subsection{Risk-averse DSD}  
\label{subsec:riskaverse}

\begin{table*}
\centering
  \begin{tabular}{|c|c|c|c|}
    \hline
    $B$ &   Average exp. reward & average risk  & $|S^*|$   \\  \hline
  0.25	& 0.18 	& 0.09 & 	6 \\ 
  1        &	0.17 & 	0.08 &	10 \\ 
   2	 &    0.13 &	0.06 & 31 \\  \hline
  \end{tabular}
\caption{\label{tab:gavin} Exploring the effect of risk tolerance parameter $B$ on the {\em gavin} dataset. For details, see Section~\ref{subsec:riskaverse}. }
\end{table*}

We perform two risk averse DSD experiments. First, for various fixed pairs of ($\lambda_1, \lambda_2)$ values, we range the parameter $B$ (reminder: $B$ is the multiplicative factor of $w^{-}(S)$, see \underline{Controlling the risk in practice}, Section~\ref{sec:motivation})  to control the trade-off between expected average reward and average risk.  A typical   outcome of our algorithm on the set of uncertain graphs we have tested it on for $\lambda_1=\lambda_2=1$, and $C=1$ is summarized in Table~\ref{tab:gavin}. As $B$ increases, we tolerate less risk, and the  average expected reward drops. This shows the trade-off between expected reward and risk.  

\begin{figure*}[htp]
\centering
\begin{tabular}{@{}c@{}@{\ }c@{}@{\ }c@{}}
\includegraphics[width=0.3\textwidth]{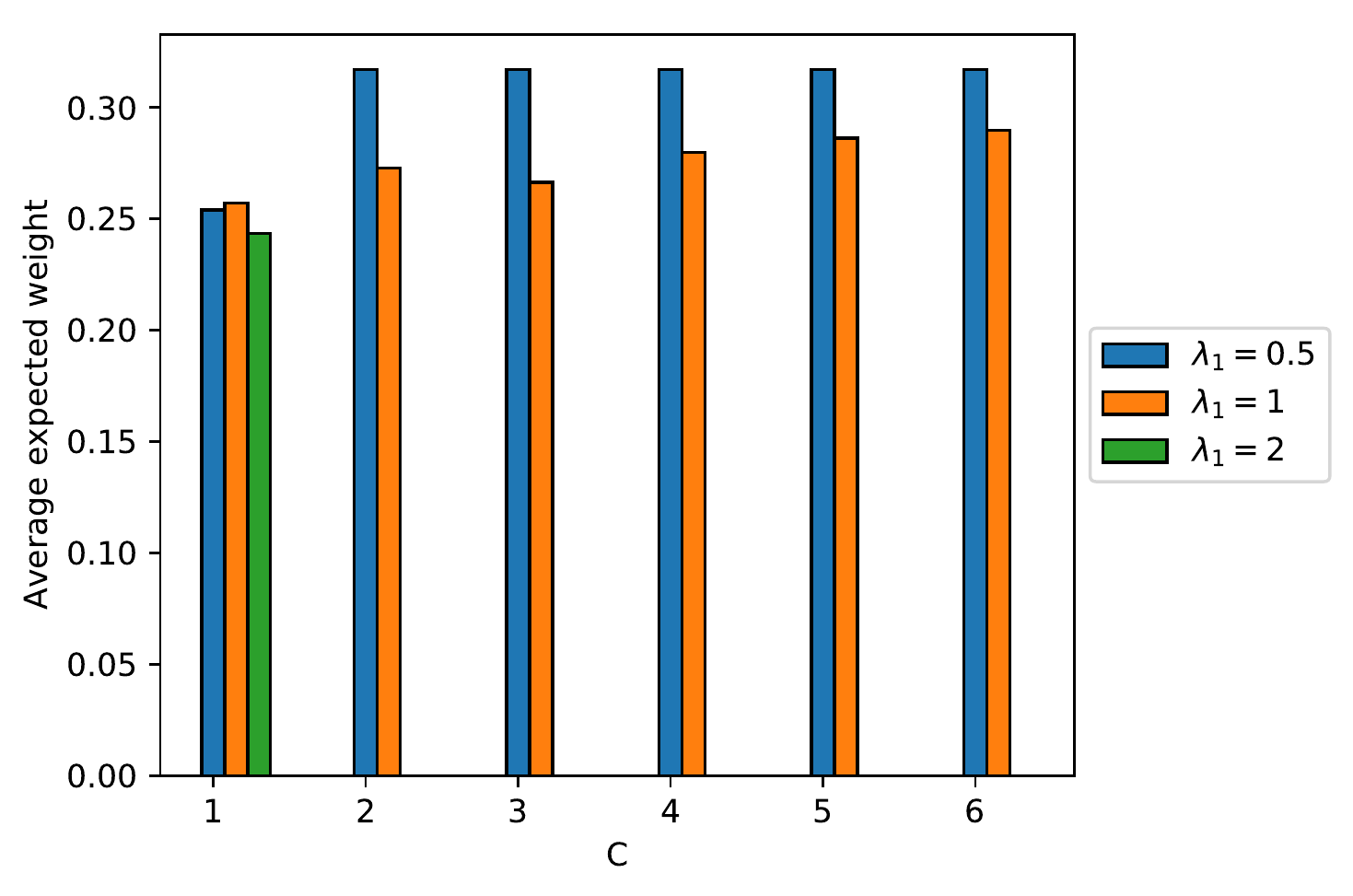} & \includegraphics[width=0.3\textwidth]{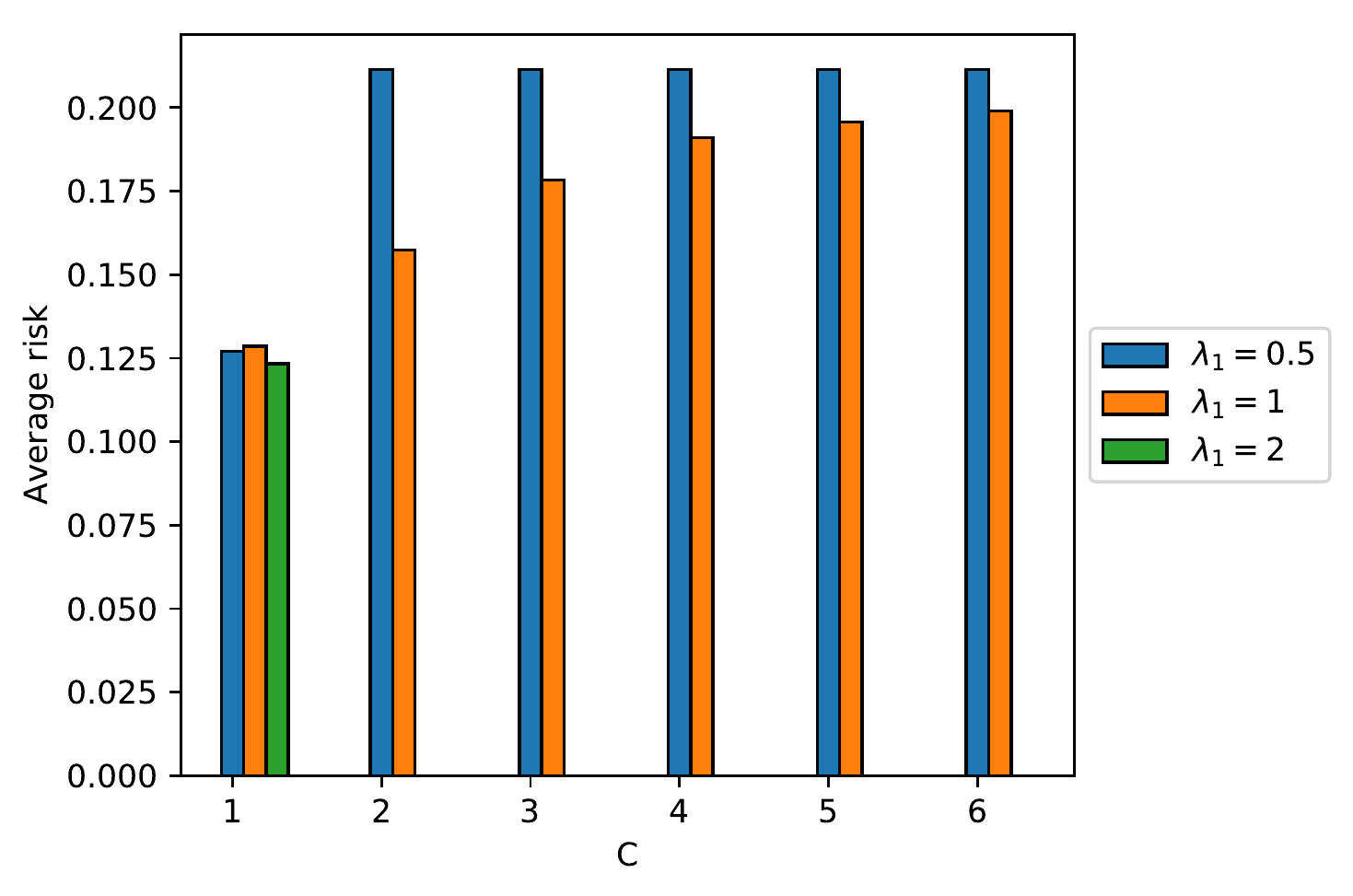} & \includegraphics[width=0.3\textwidth]{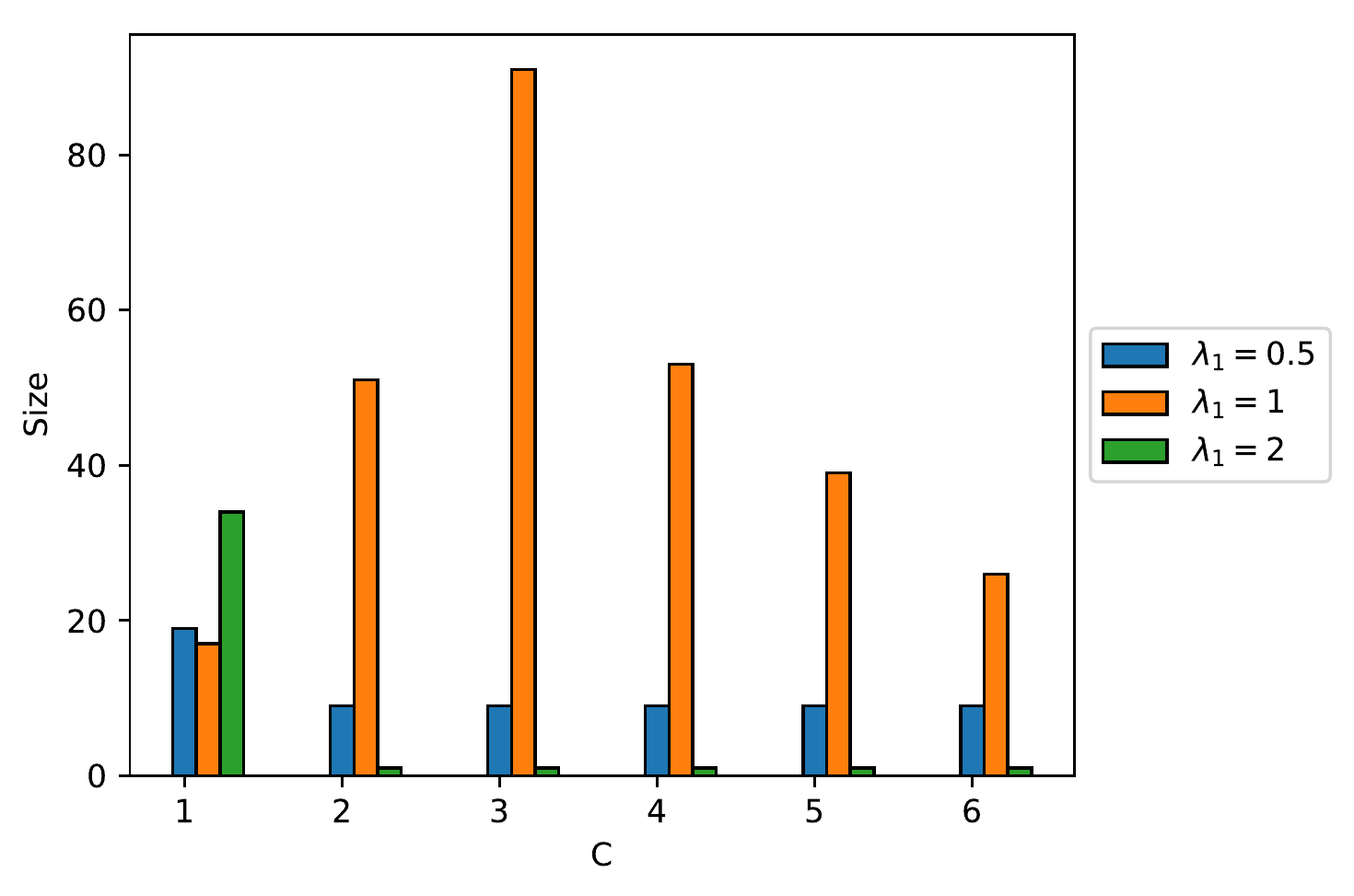}  \\
($\alpha$) & ($\beta$) & ($\gamma$) \\
\includegraphics[width=0.33\textwidth]{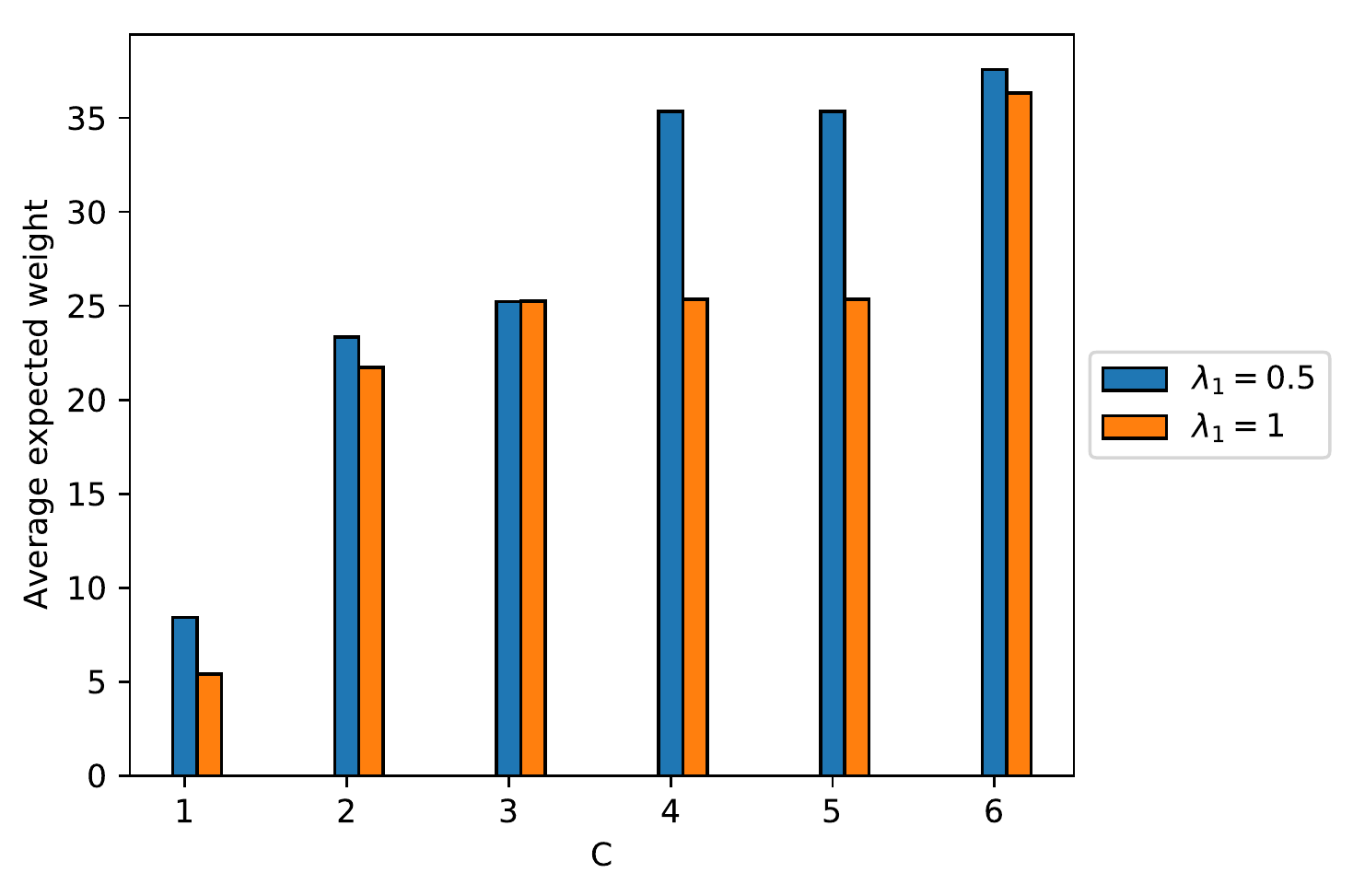} & \includegraphics[width=0.33\textwidth]{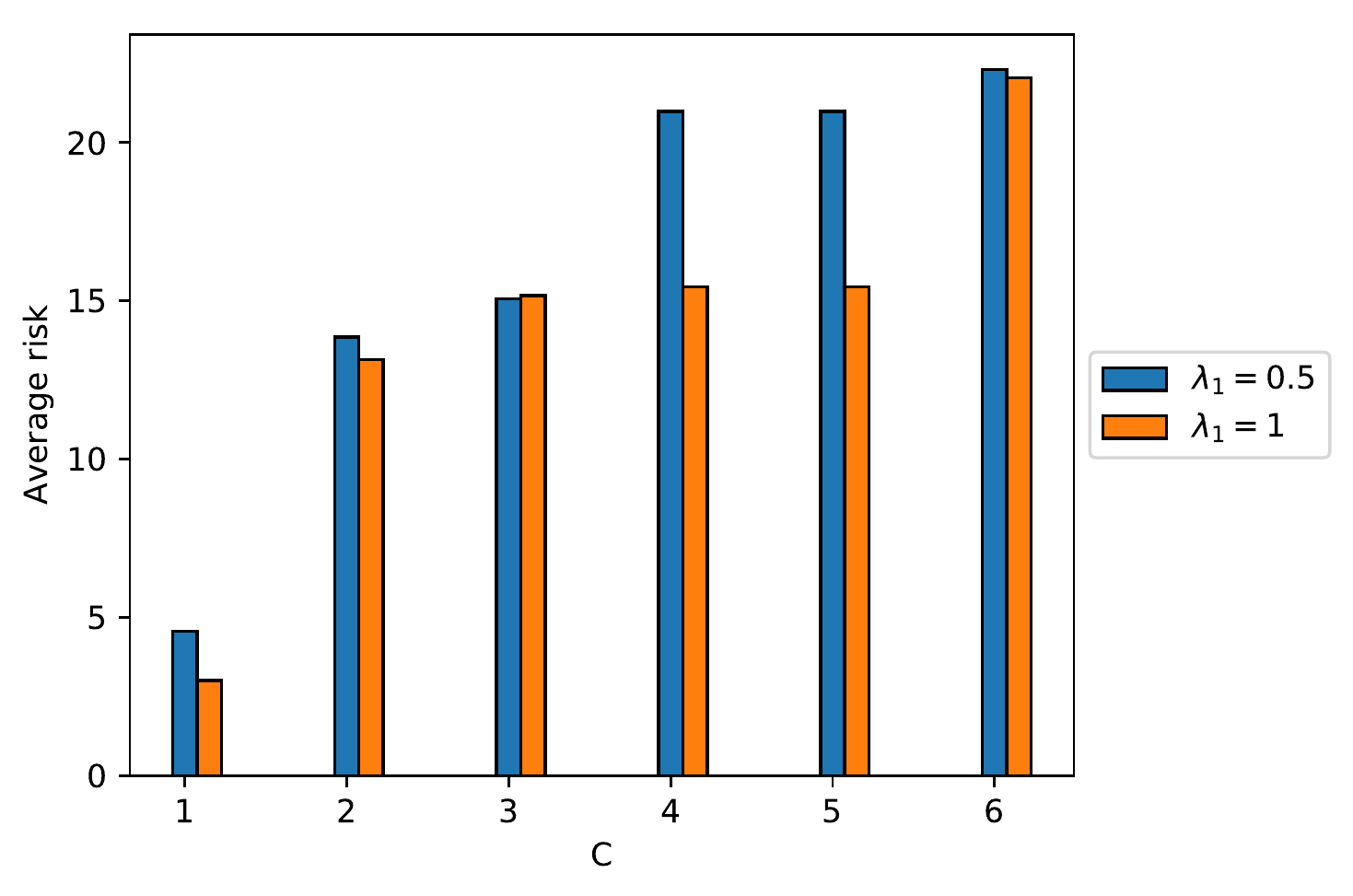}    & \includegraphics[width=0.33\textwidth]{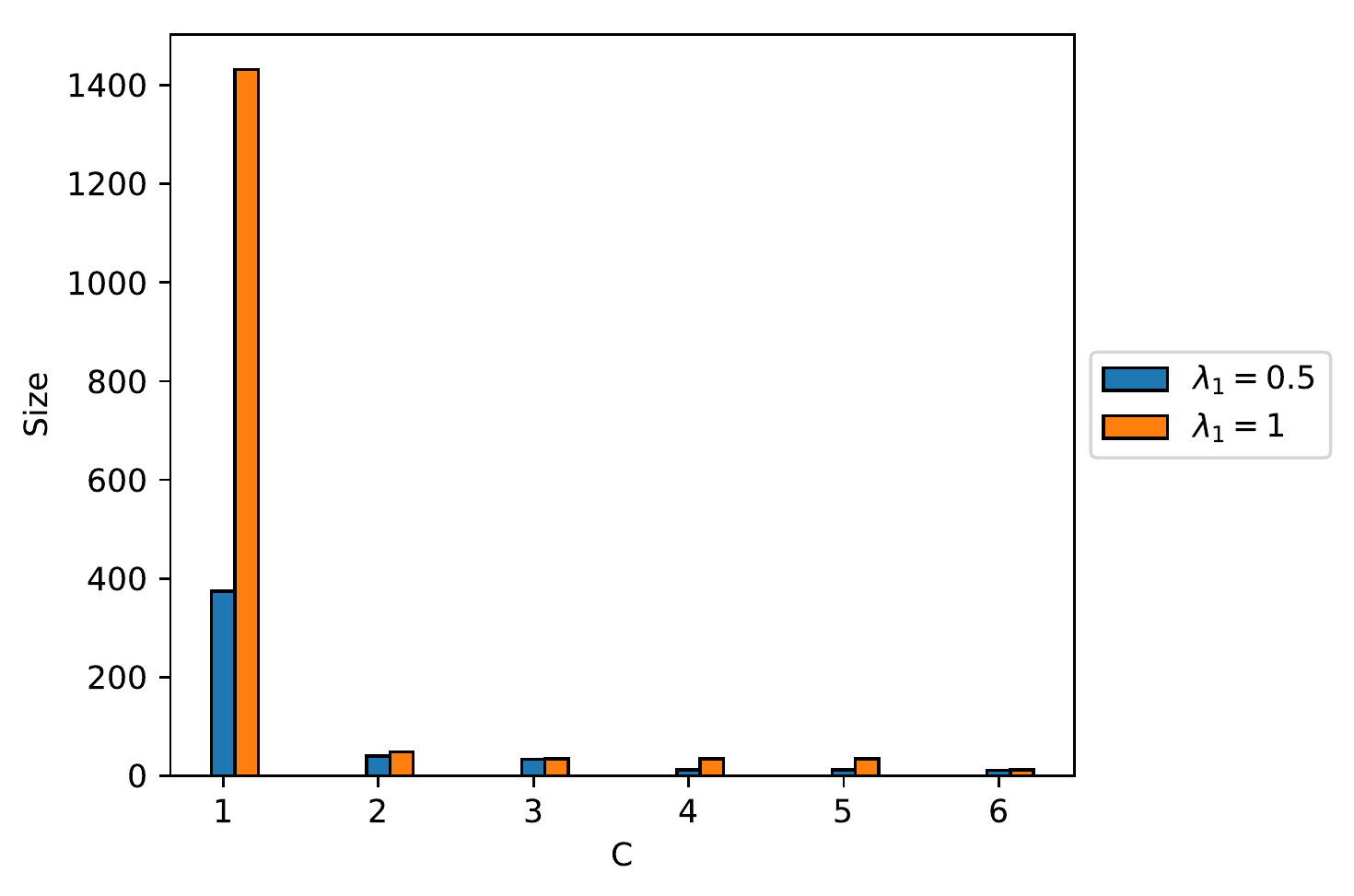}       \\
($\delta$) & ($\epsilon$) & ($\sigma\tau$) \\ 
\end{tabular}
\caption{\label{fig:riskaverse}  Risk averse DSD results for {\sc Collins} ($\alpha$)  average expected weight, ($\beta$) average risk, ($\gamma$) output size, and for {\sc TMDB}  ($\delta$)  average expected weight, ($\epsilon$) average risk, ($\sigma\tau$) output size. For details, see Section~\ref{subsec:riskaverse}. }
\end{figure*}

In our second experiment we test the effect of rest of the  parameters. We fix $B=1$, and then we perform the following procedure. For each dataset,  we fix a pair of  $(\lambda_1,\lambda_2)$ values and run our proposed algorithm using 7 values of $C$. The $C$ value $0.5$ always resulted in trivial results that would skew a lot the plots so it is omitted. Specifically, for $C=0.5$  for all three pairs of $\lambda$ values we use, we obtain (almost) the whole graph as output of the peeling process.   The three pairs of $\lambda$ values we use are  $(\lambda_1,\lambda_2) \in \{(0.5,1),(1,1),(2,1) \}$.  Our results are shown in Figure~\ref{fig:riskaverse} for the {\sc Collins} and the {\sc TMDB} graphs respectively. We remark that for the TMDB graph, the last pair of $\lambda$ values $(2,1)$ results in obtaining the whole graph as the optimal solution, so we omit it from the plots in Figures~\ref{fig:riskaverse}($\delta$),($\epsilon$), and ($\sigma\tau$), see also Figure~\ref{fig:riskaversefull} ($\iota\gamma$), ($\iota\delta$), and ($\iota\epsilon$) for the complete results. Changing $C$ value in principle does not affect risk aversion (e.g., Figure~\ref{fig:riskaverse}($\beta$)), but in some cases due to the different peeling orderings that different $C$ values  yield the output may be associated with different risks (e.g., Figure~\ref{fig:riskaverse})($\delta$)). We also observe that as we increase $\lambda_1$ the size of the output increases. This agrees with the insights we provide in Section~\ref{sec:proposed}; namely, we reward larger sets of nodes. The results for the rest of the datasets are included in the Appendix~\ref{sec:appendix}.

\begin{figure*}[htp]
\centering
\begin{tabular}{@{}c@{}@{\ }c@{}@{\ }c@{}}
\includegraphics[width=0.33\textwidth]{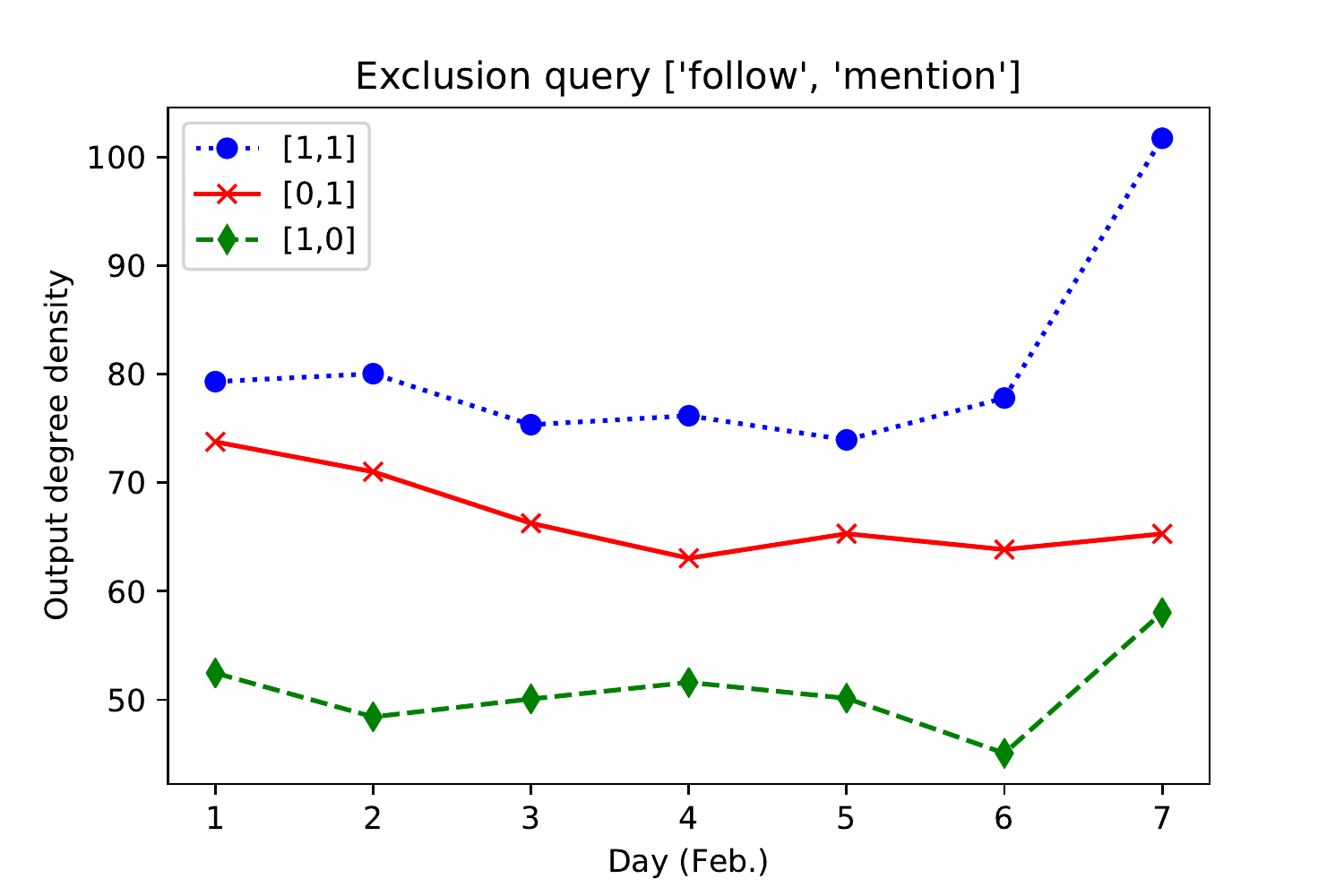} & \includegraphics[width=0.33\textwidth]{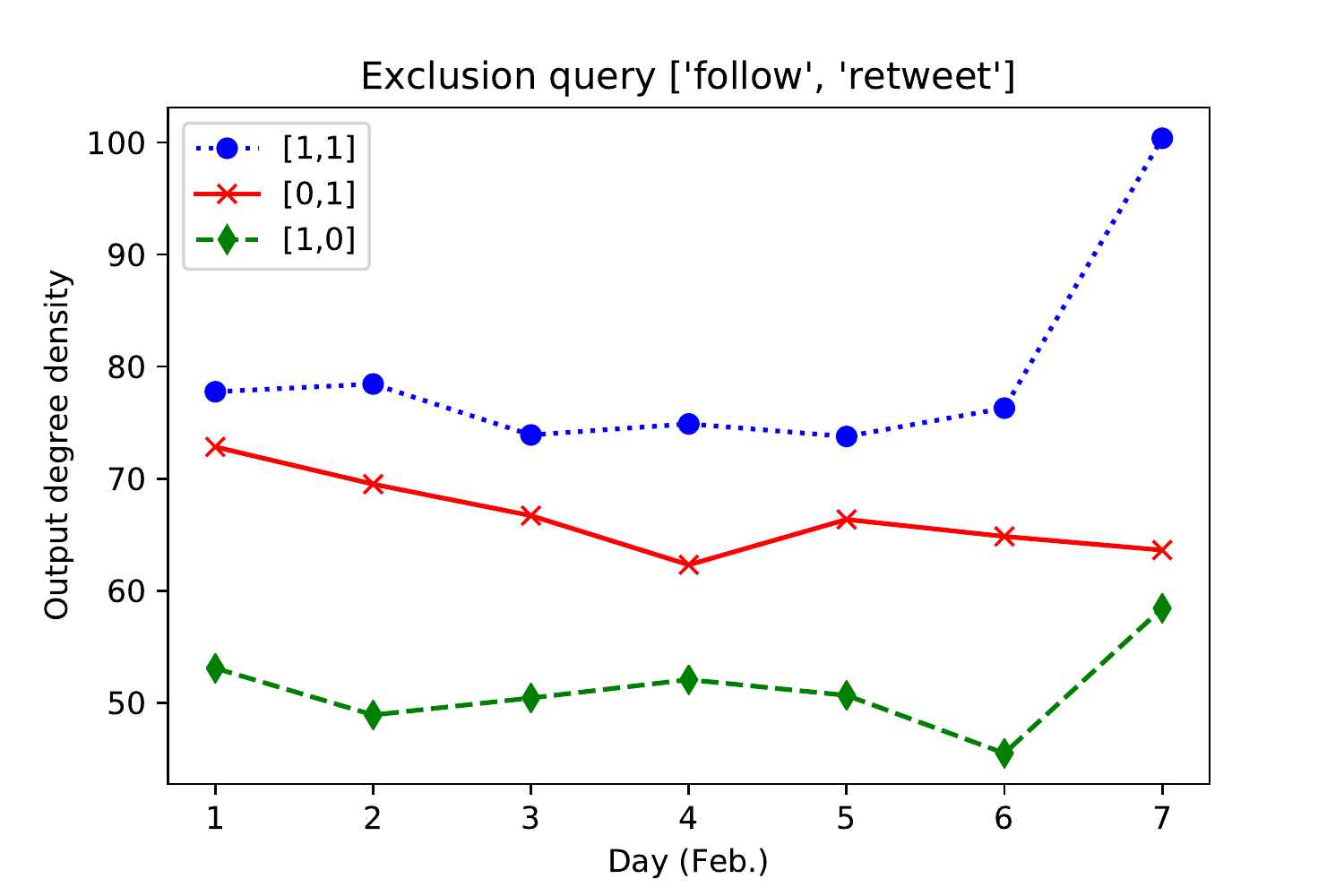}    & \includegraphics[width=0.33\textwidth]{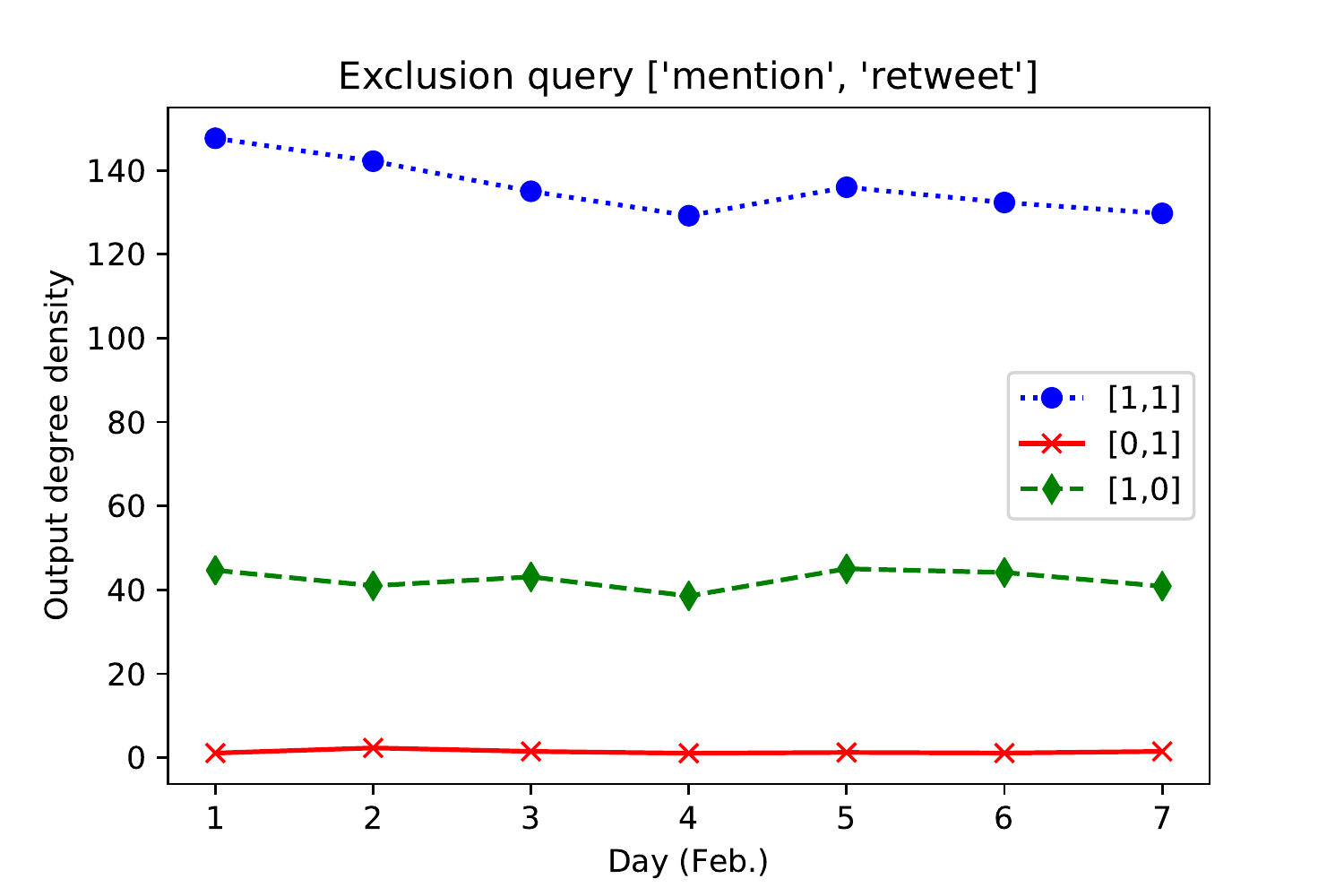}       \\
($\alpha$) & ($\beta$) & ($\gamma$) \\
\includegraphics[width=0.33\textwidth]{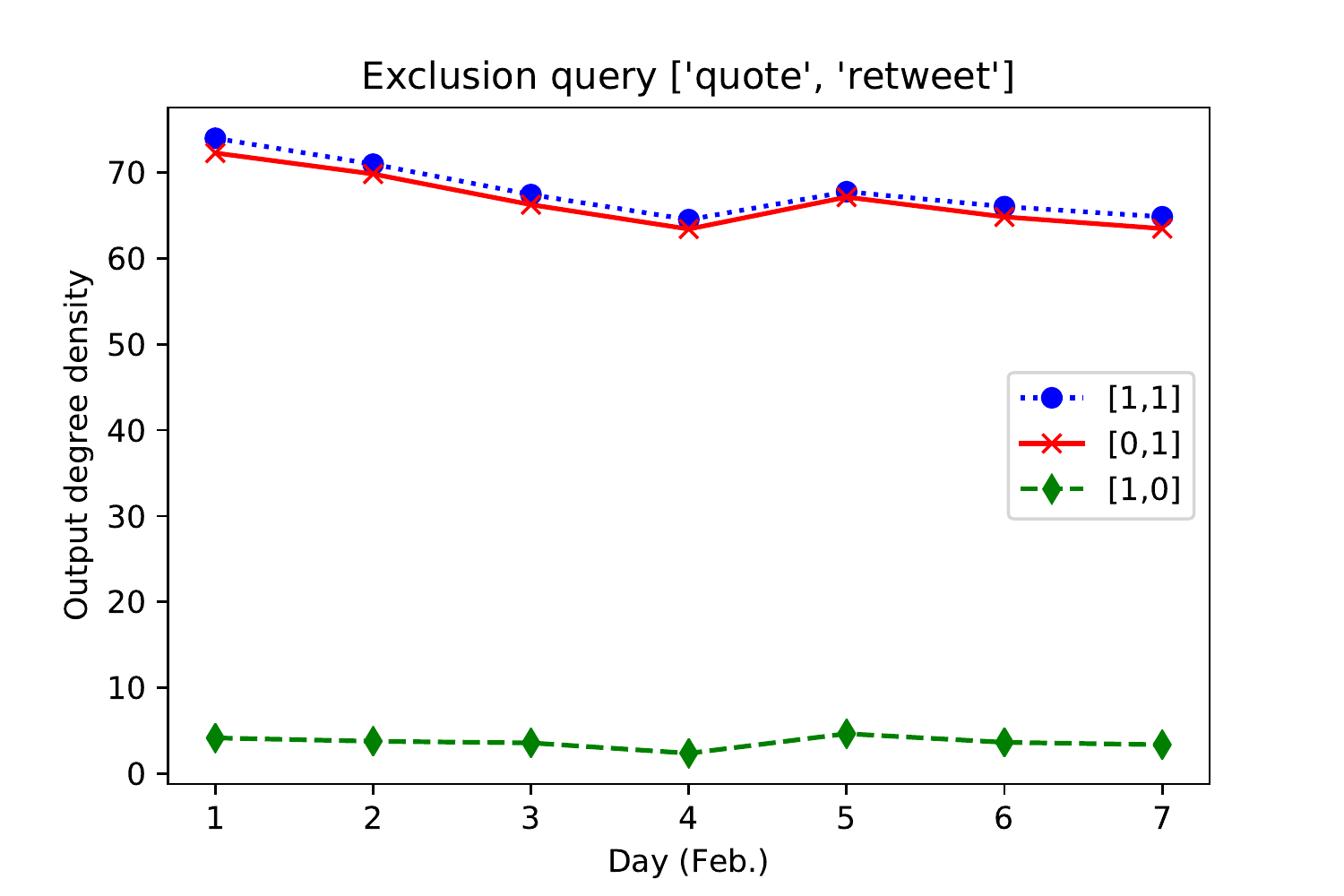} & \includegraphics[width=0.33\textwidth]{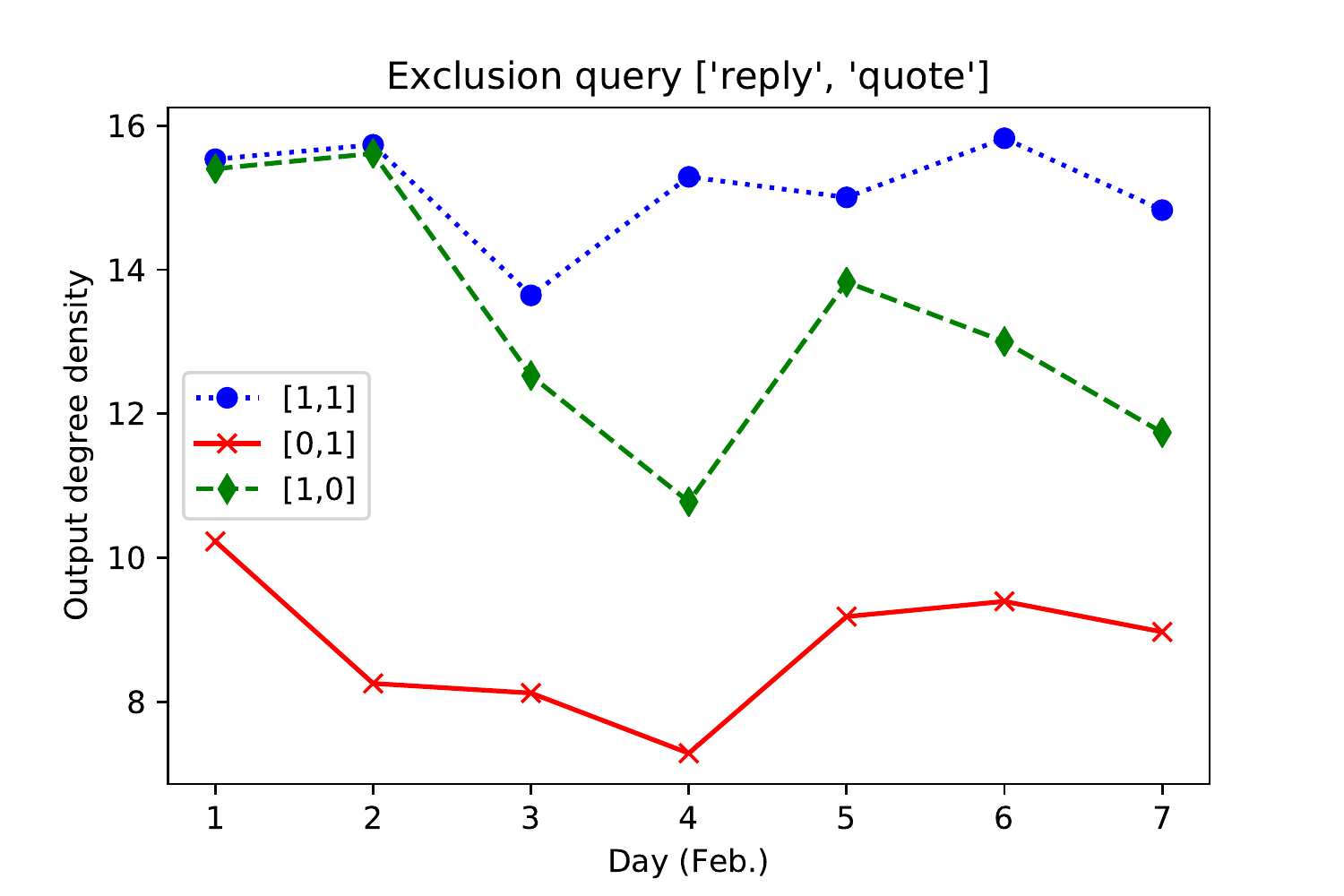}    & \includegraphics[width=0.33\textwidth]{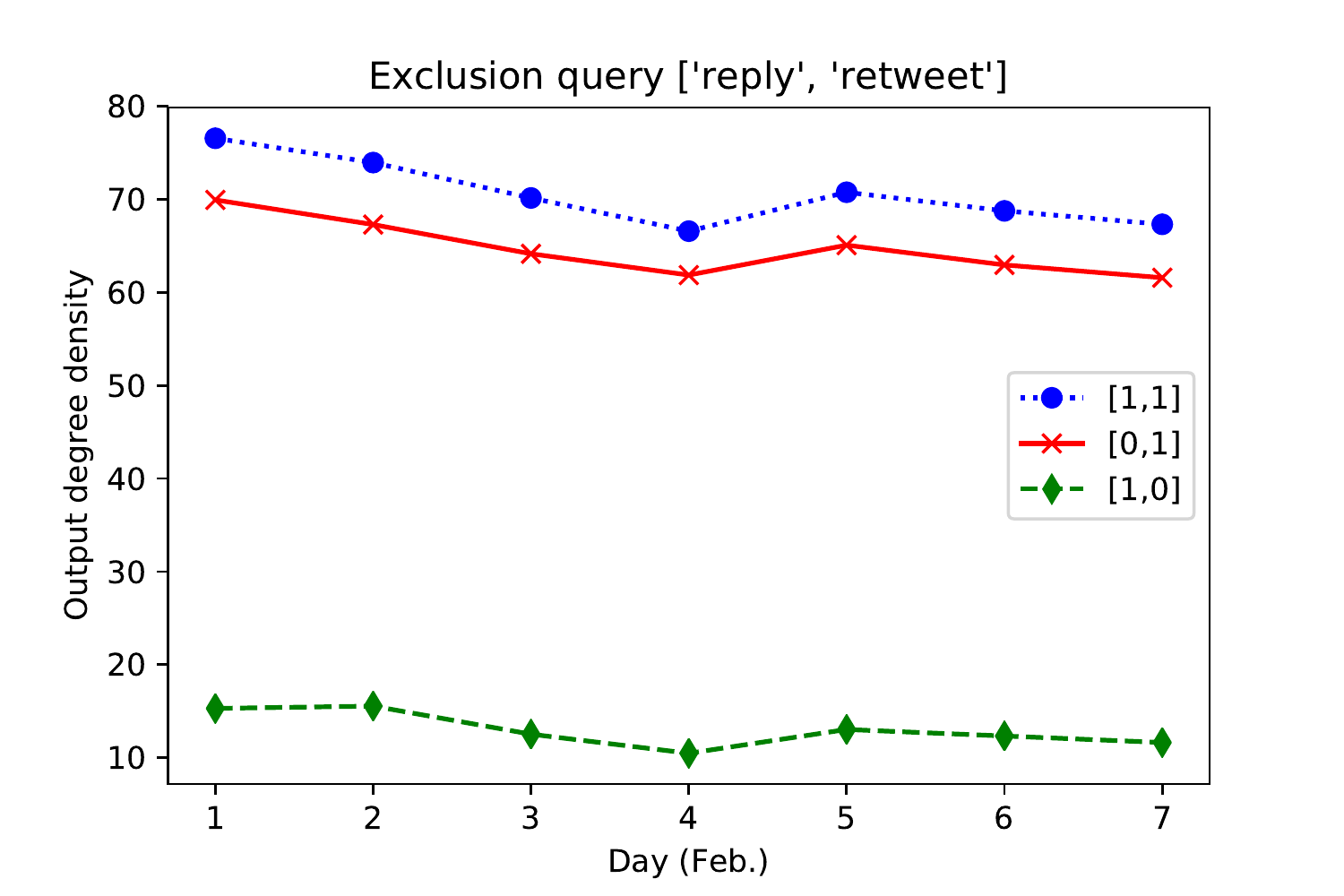}       \\
($\delta$) & ($\epsilon$) & ($\sigma\tau$) \\ 
\end{tabular}
\caption{\label{fig:peelingdegree} Degree density for three exclusion queries per each pair of interaction types  over the period of the first week of February 2018. ($\alpha$) Follow and mention. ($\beta$) Follow and retweet. ($\gamma$) Mention and retweet. ($\delta$) Quote and retweet. ($\epsilon$) Reply and quote. ($\sigma\tau$) Reply and retweet. }
\end{figure*}

\subsection{Mining Twitter using DSD-Exclusion queries} 
\label{subsec:twitterqueries} 

We test our DSD exclusion query primitive on the Twitter daily data. 
We present results that we obtain for different pairs of graphs induced by different types of interactions, for $C=1$. For each such pair,  we all possible non-trivial exclusion queries: 

\squishlist
\item Every type of interaction is allowed  (query denoted as $[1,1]$).
\item One of the two interaction types is excluded (queries denoted as $[1,0]$, and $[0,1]$).
\squishend

Figures~\ref{fig:peelingdegree},~\ref{fig:peelingsize} show for each pair of interactions the degree density, and the size (i.e., number of nodes) of the output. Interestingly, observe that in Figure~\ref{fig:peelingsize}($\gamma$) the exclusion query $[0,1]$ that excludes mentions and allows retweets results in density close to 0. This is because  the Twitter API considers every retweet as a mention. By excluding mentions, we exclude all retweets! The density is not zero, due to some small noise in the crawled mentions, i.e., there exist a few retweets that have not been included in the mentions.

We have performed more exclusion queries that involve more types of interactions. For instance, by looking into {\em reply, quote, retweet} interactions, we find the following results for two queries on February 1st, 2018. 

\squishlist 
\item When we allow all types  we find a subset of 351 nodes,  whose retweet density is 72.6, reply density 3.86, and quote density 1.08. We observe this difference since the retweet layer of interactions is much denser than the other two. 
\item When we exclude the retweets, but allow quotes and replies, we find a set of 30 nodes whose reply degree density is 15.46, and quote degree density 0.066.
\squishend

\spara{Effect of $C$, and $W$.} As we discussed earlier, ranging $W$, from small values to $+\infty$ quantifies how much we care  about excluding the undesired edge types.  Table~\ref{tab:onerun} shows what we observe typically on all experiments we have performed. Specifically, we perform an exclusion query $[1,0]$ on the {\em retweet, reply} interactions. We denote by $S^*$ the output of Algorithm~\ref{alg:heuristic}.  By inspecting the last column $\rho_{\text{reply}}(S^*)$ of the table, we observe  that even when we set the weight of each reply interaction to -1 (soft query), our algorithm outputs a set $S^*$ with very few replies, for all $C \in \{\frac{1}{10},1,10\}$ values we use.  When $W$ is set to the very large value $200\,000$ (hard query),  $\rho_{\text{reply}}(S^*)$ becomes 0 but we also observe a drop in the degree density of the retweets. For instance for $C=1$,  $\rho_{\text{retweet}}(S^*)$
drops from 72.70 to 30.38. 

\begin{table*}
\centering
  \begin{tabular}{|c|c|c|c|c|}
    \hline
    $C$ &  $W$  &  $|S^*|$ & $\rho_{\text{retweet}}(S^*)$ & $\rho_{\text{reply}}(S^*)$   \\  \hline
    \multirow{3}{*}{$0.1$} & 1 & 296 & 63.44 & -0.75 \\
        & 5 &  99 & 45.67 & -0.01 \\
        & 200\,000 & 200 & 30.37 & 0 \\
     \hline
    \multirow{3}{*}{$1$} & 1 & 346  & 72.70 & -2.75 \\
        & 5 &  319 & 68.70 & -1.29 \\
        & 200\,000 & 200 & 30.38 & 0 \\
     \hline
  \multirow{3}{*}{$10$} & 1 & 351  & 73.10 & -3.31 \\
        & 5 &  351 & 73.10 & -3.31 \\
        & 200\,000 & 200  & 30.37 & 0 \\
     \hline
  \end{tabular}
\caption{\label{tab:onerun} Exploring the effect of the negative weight $-W$ on the excluded edge types for various $C$ values. For details, see Section~\ref{subsec:twitterqueries}. }
\end{table*}

\begin{figure*}[htp]
\centering
\begin{tabular}{@{}c@{}@{\ }c@{}@{\ }c@{}}
\includegraphics[width=0.33\textwidth]{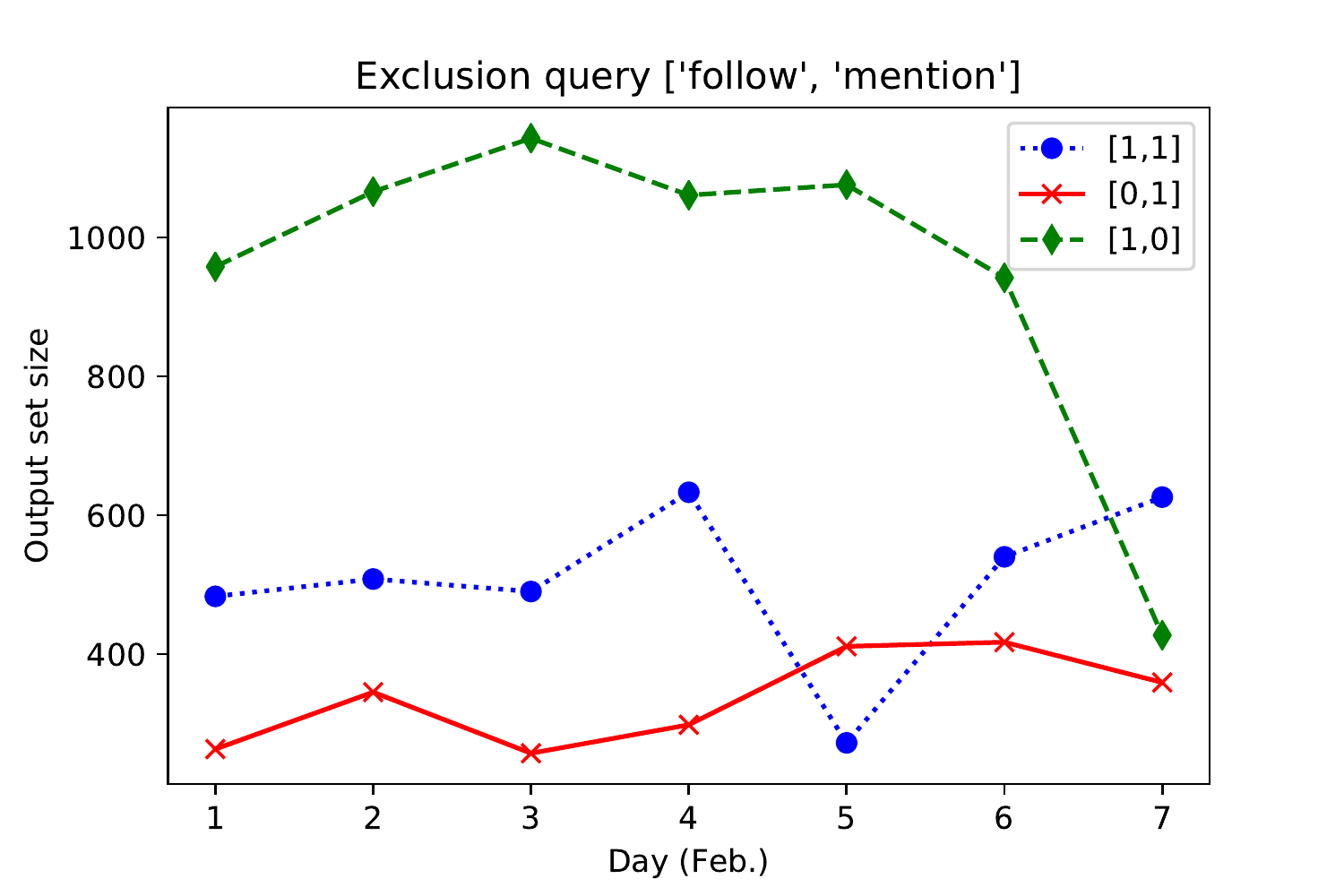} & \includegraphics[width=0.33\textwidth]{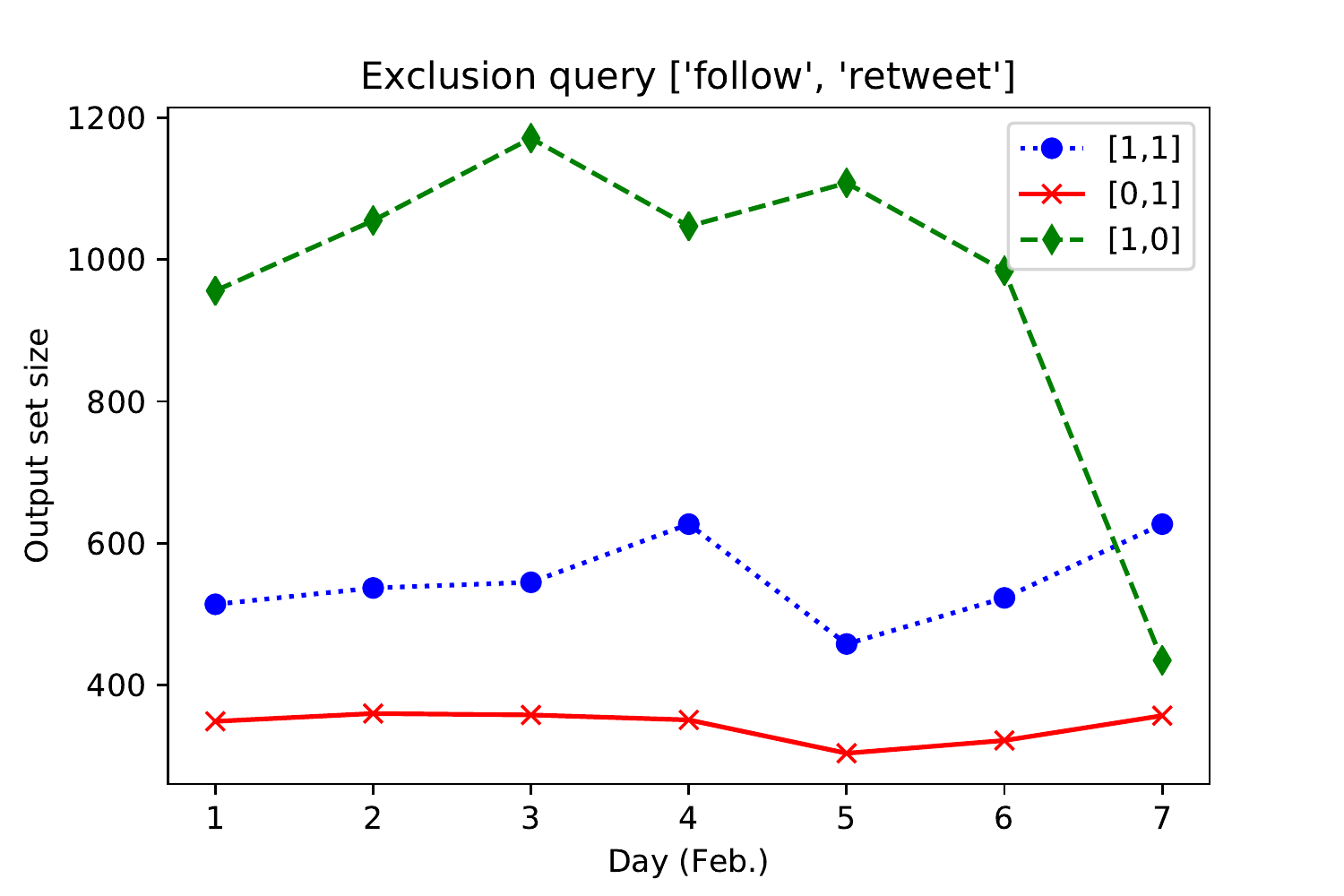}    & \includegraphics[width=0.33\textwidth]{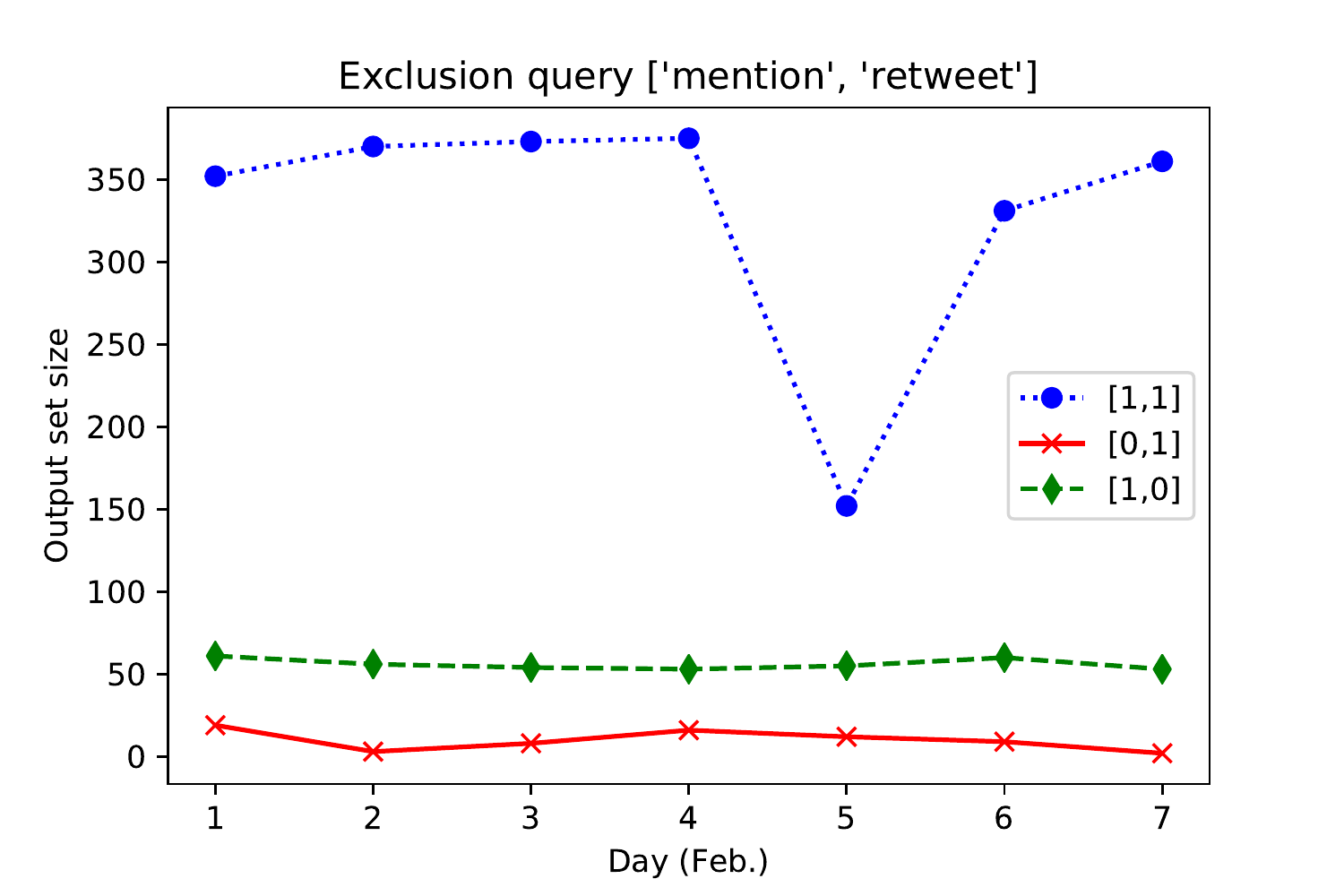}       \\
($\alpha$) & ($\beta$) & ($\gamma$) \\
\includegraphics[width=0.33\textwidth]{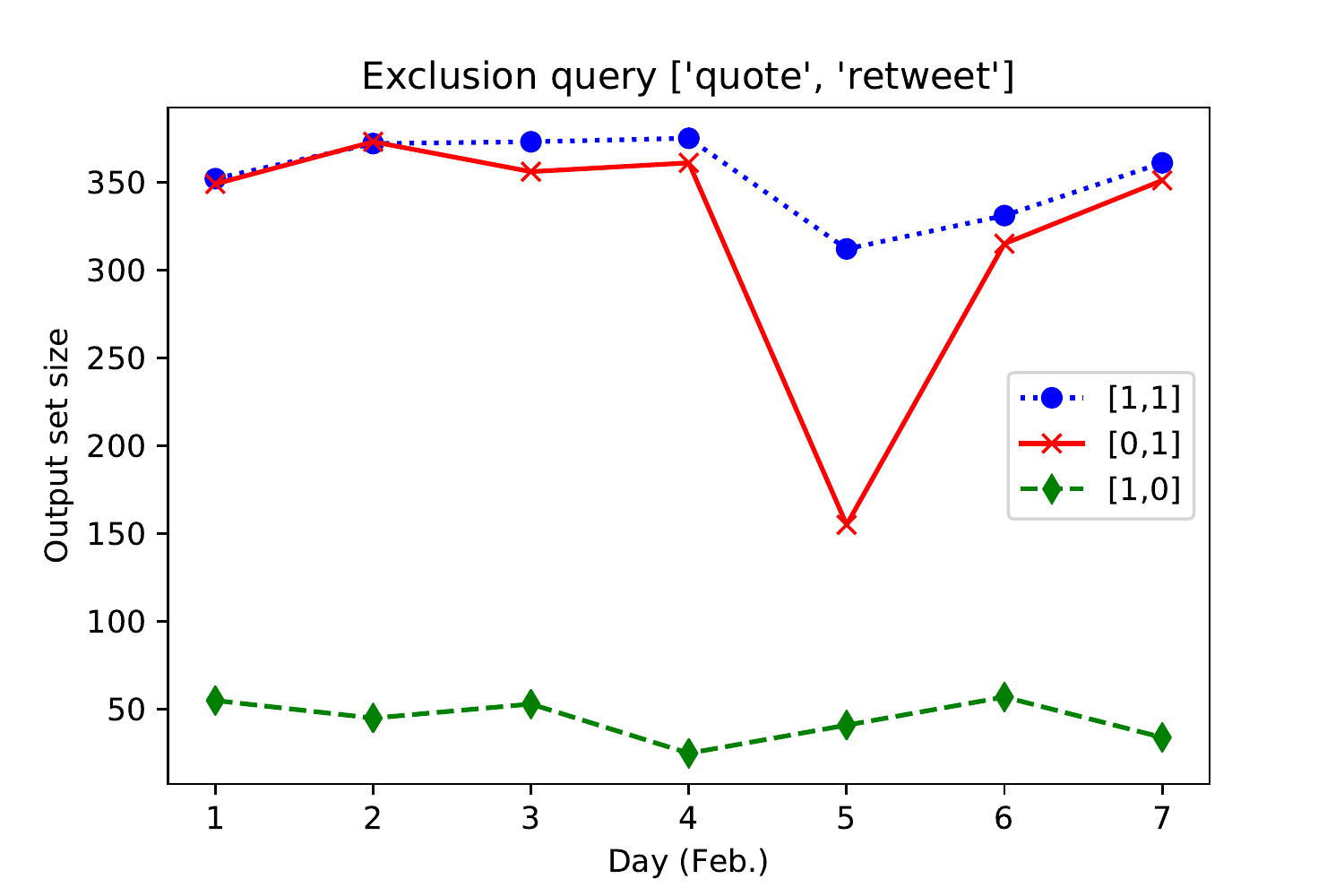} & \includegraphics[width=0.33\textwidth]{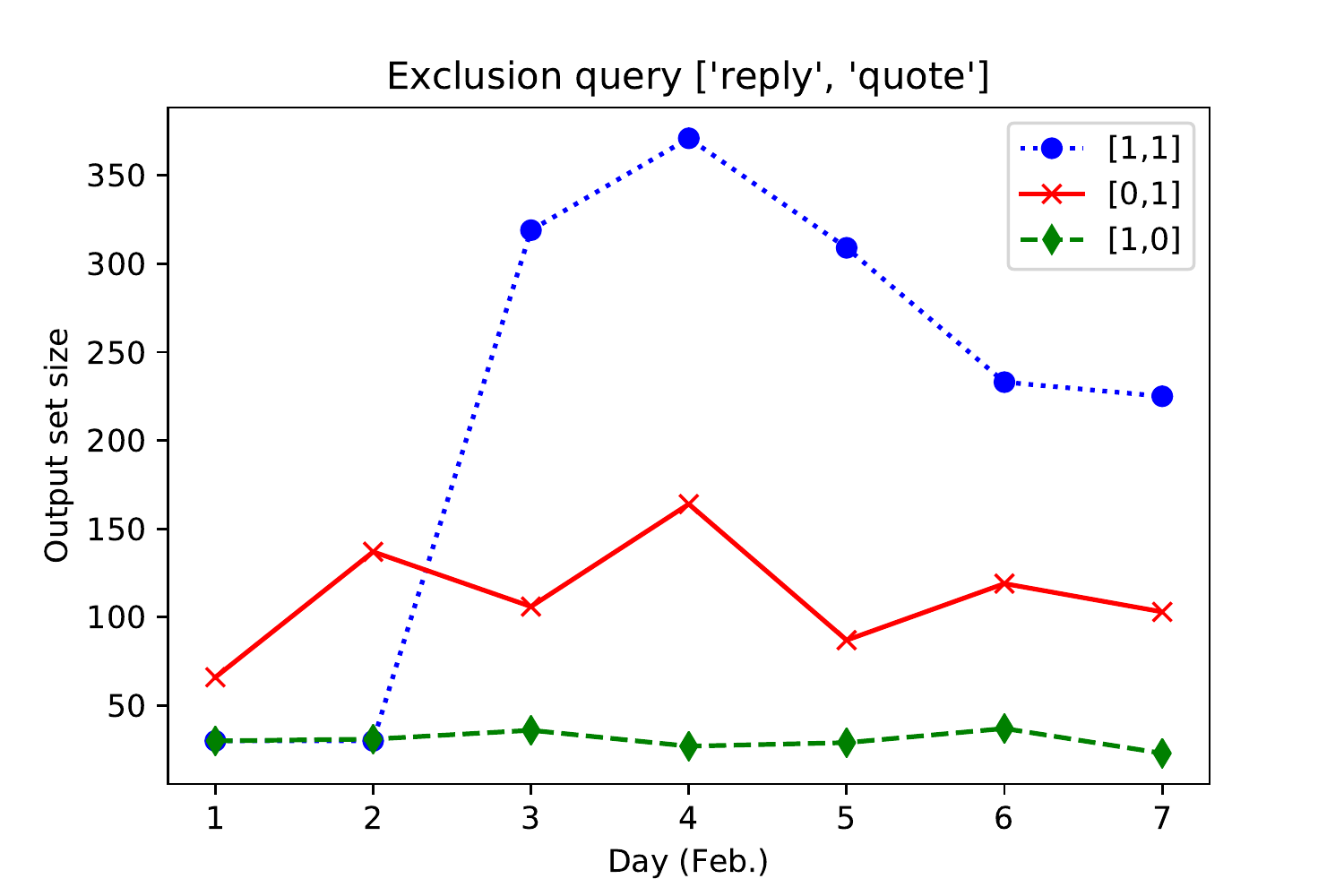}    & \includegraphics[width=0.33\textwidth]{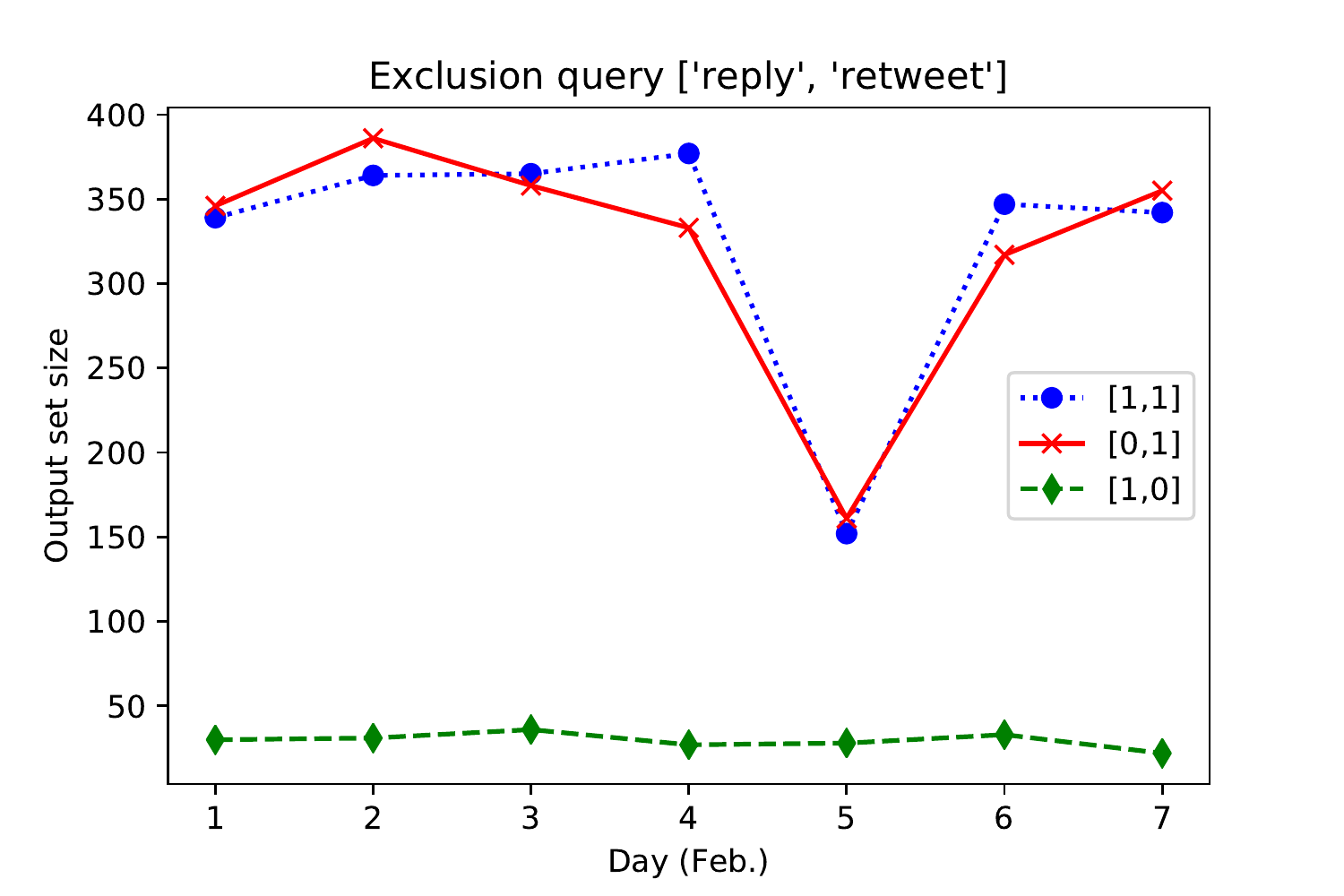}       \\
($\delta$) & ($\epsilon$) & ($\sigma\tau$) \\ 
\end{tabular}
\caption{\label{fig:peelingsize} Output sizes for three exclusion queries per each pair of interaction types over the period of the first week of February 2018. ($\alpha$) Follow and mention. ($\beta$) Follow and retweet. ($\gamma$) Mention and retweet. ($\delta$) Quote and retweet. ($\epsilon$) Reply and quote. ($\sigma\tau$) Reply and retweet.  }
\end{figure*}

\section{Conclusion}
\label{sec:concl}
\spara{Summary.} In this paper we study  dense subgraph discovery problem on graphs with negative weights in greater depth than prior work \cite{cadena2016dense}. We show that the problem in \NPhard, and then we propose algorithms that are based on peeling, and are both space-, and time-efficient. Furthermore, we provide two important graph mining primitives, that are both formalized as \NegDSD problems. The first primitive is applicable to uncertain graphs, and  extracts subgraphs that in expectation induce large weight, and are risk-averse.  The second primitive enables for efficient mining of multilayer graphs; specifically, it extracts dense subgraphs that exclude certain types of undesired edges, that are passed as input to the algorithm.  Given the ubiquitousness of uncertain and multilayer graphs, and the importance of  DSD \cite{gionis2015dense}, we believe that our primitives will be applied on various applications. Finally, we test our proposed methods on various real-world datasets, and verify experimentally their usefulness and efficiency.

\spara{Open Problems.}  While we have performed some work on understanding the computational complexity of the problem, understanding at a greater depth the complexity (especially under reasonable assumptions on the negative weights) of the problem remains largely open. For example, we provided sufficient conditions under which the problem is poly-time solvable. What are necessary and sufficient conditions for poly-time solvability of the DSP on graphs with negative edge weights? Furthermore, using our primitives on more applications is an interesting direction (e.g., mining time-series; for example, can we find  clusters of time-series that are correlated under one similarity measure but not correlated under similarity a second similarity measure?). Also, developing an approximation or bi-criteria approximation algorithms for risk averse DSD that aims to maximize the expected reward subject to bounds on the risk is an interesting open problem.  Finally, designing efficient risk-averse  graph mining algorithms is a broad  interesting direction.

%\bibliographystyle{abbrv}
%\bibliography{ref}

\begin{thebibliography}{10}

\bibitem{csexchange}
Max-cut with negative weight edges by Peter Shor
  \url{https://cstheory.stackexchange.com/questions/2312/max-cut-with-negative-weight-edges}.

\bibitem{AITT00}
Y.~Asahiro, K.~Iwama, H.~Tamaki, and T.~Tokuyama.
\newblock Greedily finding a dense subgraph.
\newblock {\em J. Algorithms}, 34(2), 2000.

\bibitem{asthana2004predicting}
S.~Asthana, O.~D. King, F.~D. Gibbons, and F.~P. Roth.
\newblock Predicting protein complex membership using probabilistic network
  reliability.
\newblock {\em Genome research}, 14(6):1170--1175, 2004.

\bibitem{boldi2012injecting}
P.~Boldi, F.~Bonchi, A.~Gionis, and T.~Tassa.
\newblock Injecting uncertainty in graphs for identity obfuscation.
\newblock {\em Proceedings of the VLDB Endowment}, 5(11):1376--1387, 2012.

\bibitem{bollobas2007phase}
B.~Bollob{\'a}s, S.~Janson, and O.~Riordan.
\newblock The phase transition in inhomogeneous random graphs.
\newblock {\em Random Structures \& Algorithms}, 31(1):3--122, 2007.

\bibitem{bonchi2014core}
F.~Bonchi, F.~Gullo, A.~Kaltenbrunner, and Y.~Volkovich.
\newblock Core decomposition of uncertain graphs.
\newblock In {\em Proc. of the 20th ACM SIGKDD conference}, pages 1316--1325.
  ACM, 2014.

\bibitem{cadena2016dense}
J.~Cadena, A.~K. Vullikanti, and C.~C. Aggarwal.
\newblock On dense subgraphs in signed network streams.
\newblock In {\em 2016 IEEE 16th International Conference on Data Mining
  (ICDM)}, pages 51--60. IEEE, 2016.

\bibitem{Char00}
M.~Charikar.
\newblock Greedy approximation algorithms for finding dense components in a
  graph.
\newblock In {\em APPROX}, 2000.

\bibitem{charikar2018finding}
M.~Charikar, Y.~Naamad, and J.~Wu.
\newblock On finding dense common subgraphs.
\newblock {\em arXiv preprint arXiv:1802.06361}, 2018.

\bibitem{babis1}
T.~Chen and C.~E. Tsourakakis.
\newblock Tmdb uncertain graph.
\newblock
  \url{https://drive.google.com/open?id=1C69MndtfSoUflPkeBa0mbiC9FZD6xbpN}.

\bibitem{chen2016robust}
W.~Chen, T.~Lin, Z.~Tan, M.~Zhao, and X.~Zhou.
\newblock Robust influence maximization.
\newblock In {\em Proceedings of the 22nd ACM SIGKDD International Conference
  on Knowledge Discovery and Data Mining}, pages 795--804. ACM, 2016.

\bibitem{collins2007toward}
S.~R. Collins, P.~Kemmeren, X.-C. Zhao, J.~F. Greenblatt, F.~Spencer, F.~C.
  Holstege, J.~S. Weissman, and N.~J. Krogan.
\newblock Toward a comprehensive atlas of the physical interactome of
  saccharomyces cerevisiae.
\newblock {\em Molecular \& Cellular Proteomics}, 6(3):439--450, 2007.

\bibitem{dalvi2007efficient}
N.~N. Dalvi and D.~Suciu.
\newblock Efficient query evaluation on probabilistic databases.
\newblock {\em {VLDB} J.}, 16(4):523--544, 2007.

\bibitem{galimberti2017core}
E.~Galimberti, F.~Bonchi, and F.~Gullo.
\newblock Core decomposition and densest subgraph in multilayer networks.
\newblock In {\em Proceedings of the 2017 ACM on Conference on Information and
  Knowledge Management}, pages 1807--1816. ACM, 2017.

\bibitem{galimberti2018core}
E.~Galimberti, F.~Bonchi, F.~Gullo, and T.~Lanciano.
\newblock Core decomposition in multilayer networks: Theory, algorithms, and
  applications.
\newblock {\em arXiv preprint arXiv:1812.08712}, 2018.

\bibitem{GGT89}
G.~Gallo, M.~D. Grigoriadis, and R.~E. Tarjan.
\newblock A fast parametric maximum flow algorithm and applications.
\newblock {\em Journal of Computing}, 18(1), 1989.

\bibitem{gavin2006proteome}
A.-C. Gavin, P.~Aloy, P.~Grandi, R.~Krause, M.~Boesche, M.~Marzioch, C.~Rau,
  L.~J. Jensen, S.~Bastuck, B.~D{\"u}mpelfeld, et~al.
\newblock Proteome survey reveals modularity of the yeast cell machinery.
\newblock {\em Nature}, 440(7084):631, 2006.

\bibitem{gionis2015dense}
A.~Gionis and C.~E. Tsourakakis.
\newblock Dense subgraph discovery: Kdd 2015 tutorial.
\newblock In {\em Proceedings of the 21th ACM SIGKDD International Conference
  on Knowledge Discovery and Data Mining}, pages 2313--2314. ACM, 2015.

\bibitem{goldberg84}
A.~V. Goldberg.
\newblock Finding a maximum density subgraph.
\newblock Technical report, University of California at Berkeley, 1984.

\bibitem{gunopulos2001time}
D.~Gunopulos and G.~Das.
\newblock Time series similarity measures and time series indexing.
\newblock In {\em ACM Sigmod Record}, 2001.

\bibitem{hastad}
J.~Hastad.
\newblock Clique is hard to approximate within $n^{1-\epsilon}$.
\newblock {\em Acta Mathematica}, 182(1), 1999.

\bibitem{HeKempe}
X.~He and D.~Kempe.
\newblock Robust influence maximization.
\newblock In {\em Proceedings of the 22Nd ACM SIGKDD International Conference
  on Knowledge Discovery and Data Mining}, KDD '16, pages 885--894, New York,
  NY, USA, 2016. ACM.

\bibitem{huang2016truss}
X.~Huang, W.~Lu, and L.~V. Lakshmanan.
\newblock Truss decomposition of probabilistic graphs: Semantics and
  algorithms.
\newblock In {\em Proceedings of SIGMOD 2016}, pages 77--90, 2016.

\bibitem{jin2011discovering}
R.~Jin, L.~Liu, and C.~C. Aggarwal.
\newblock Discovering highly reliable subgraphs in uncertain graphs.
\newblock In {\em Proceedings of KDD 2011}, pages 992--1000, 2011.

\bibitem{karp}
R.~M. Karp.
\newblock Reducibility among combinatorial problems.
\newblock In R.~Miller and J.~Thatcher, editors, {\em Complexity of Computer
  Computations}. 1972.

\bibitem{kempe2003maximizing}
D.~Kempe, J.~Kleinberg, and {\'E}.~Tardos.
\newblock Maximizing the spread of influence through a social network.
\newblock In {\em Proceedings of KDD 2003}, pages 137--146. ACM, 2003.

\bibitem{khan2015uncertain}
A.~Khan and L.~Chen.
\newblock On uncertain graphs modeling and queries.
\newblock {\em Proceedings of the VLDB Endowment}, 8(12):2042--2043, 2015.

\bibitem{Khuller}
S.~Khuller and B.~Saha.
\newblock On finding dense subgraphs.
\newblock In {\em 36th International Colloquium on Automata, Languages and
  Programming (ICALP)}, 2009.

\bibitem{kollios2013clustering}
G.~Kollios, M.~Potamias, and E.~Terzi.
\newblock Clustering large probabilistic graphs.
\newblock {\em IEEE Transactions on Knowledge and Data Engineering},
  25(2):325--336, 2013.

\bibitem{krogan2006global}
N.~J. Krogan, G.~Cagney, H.~Yu, G.~Zhong, X.~Guo, A.~Ignatchenko, J.~Li, S.~Pu,
  N.~Datta, A.~P. Tikuisis, et~al.
\newblock Global landscape of protein complexes in the yeast saccharomyces
  cerevisiae.
\newblock {\em Nature}, 440(7084):637, 2006.

\bibitem{paccanarolab}
P.~lab.
\newblock
  \url{http://www.paccanarolab.org/static_content/clusterone/cl1_datasets.zip}.

\bibitem{lawler2001combinatorial}
E.~L. Lawler.
\newblock {\em Combinatorial optimization: networks and matroids}.
\newblock Courier Corporation, 2001.

\bibitem{liu2012reliable}
L.~Liu, R.~Jin, C.~Aggarwal, and Y.~Shen.
\newblock Reliable clustering on uncertain graphs.
\newblock In {\em Proceedings of ICDM 2012}, pages 459--468. IEEE, 2012.

\bibitem{mitzenmacher2015scalable}
M.~Mitzenmacher, J.~Pachocki, R.~Peng, C.~Tsourakakis, and S.~C. Xu.
\newblock Scalable large near-clique detection in large-scale networks via
  sampling.
\newblock In {\em Proceedings of the 21th ACM SIGKDD International Conference
  on Knowledge Discovery and Data Mining}, pages 815--824. ACM, 2015.

\bibitem{miyauchi2018robust}
A.~Miyauchi and A.~Takeda.
\newblock Robust densest subgraph discovery.
\newblock In {\em 2018 IEEE International Conference on Data Mining (ICDM)},
  pages 1188--1193. IEEE, 2018.

\bibitem{moustafa2014subgraph}
W.~E. Moustafa, A.~Kimmig, A.~Deshpande, and L.~Getoor.
\newblock Subgraph pattern matching over uncertain graphs with identity linkage
  uncertainty.
\newblock In {\em Proceedings of ICDE 2014}, pages 904--915. IEEE, 2014.

\bibitem{parchas2014pursuit}
P.~Parchas, F.~Gullo, D.~Papadias, and F.~Bonchi.
\newblock The pursuit of a good possible world: extracting representative
  instances of uncertain graphs.
\newblock In {\em Proceedings SIGMOD 2014}, pages 967--978, 2014.

\bibitem{potamias2010k}
M.~Potamias, F.~Bonchi, A.~Gionis, and G.~Kollios.
\newblock K-nearest neighbors in uncertain graphs.
\newblock {\em Proceedings of the VLDB Endowment}, 3(1-2):997--1008, 2010.

\bibitem{pratikakis2018twawler}
P.~Pratikakis.
\newblock twawler: A lightweight twitter crawler.
\newblock {\em arXiv preprint arXiv:1804.07748}, 2018.

\bibitem{ravi1996constrained}
R.~Ravi and M.~X. Goemans.
\newblock The constrained minimum spanning tree problem.
\newblock In {\em Scandinavian Workshop on Algorithm Theory}, pages 66--75.
  Springer, 1996.

\bibitem{roth2004kidney}
A.~E. Roth, T.~S{\"o}nmez, and M.~U. {\"U}nver.
\newblock Kidney exchange.
\newblock {\em The Quarterly Journal of Economics}, 119(2):457--488, 2004.

\bibitem{semertzidis2016best}
K.~Semertzidis, E.~Pitoura, E.~Terzi, and P.~Tsaparas.
\newblock Best friends forever (bff): Finding lasting dense subgraphs.
\newblock {\em arXiv preprint arXiv:1612.05440}, 2016.

\bibitem{serra2014empirical}
J.~Serra and J.~L. Arcos.
\newblock An empirical evaluation of similarity measures for time series
  classification.
\newblock {\em Knowledge-Based Systems}, 67:305--314, 2014.

\bibitem{tsourakakis2015streaming}
C.~Tsourakakis.
\newblock Streaming graph partitioning in the planted partition model.
\newblock In {\em Proceedings of the 2015 ACM on Conference on Online Social
  Networks}, pages 27--35. ACM, 2015.

\bibitem{tsourakakis2013denser}
C.~Tsourakakis, F.~Bonchi, A.~Gionis, F.~Gullo, and M.~Tsiarli.
\newblock Denser than the densest subgraph: extracting optimal quasi-cliques
  with quality guarantees.
\newblock In {\em Proceedings of the 19th ACM SIGKDD international conference
  on Knowledge discovery and data mining}, pages 104--112. ACM, 2013.

\bibitem{tsourakakis2014fennel}
C.~Tsourakakis, C.~Gkantsidis, B.~Radunovic, and M.~Vojnovic.
\newblock Fennel: Streaming graph partitioning for massive scale graphs.
\newblock In {\em Proceedings of the 7th ACM international conference on Web
  search and data mining}, pages 333--342. ACM, 2014.

\bibitem{tsourakakis2015kclique}
C.~E. Tsourakakis.
\newblock The k-clique densest subgraph problem.
\newblock {\em 24th International World Wide Web Conference (WWW)}, 2015.

\bibitem{tsourakakis2018risk}
C.~E. Tsourakakis, S.~Sekar, J.~Lam, and L.~Yang.
\newblock Risk-averse matchings over uncertain graph databases.
\newblock {\em arXiv preprint arXiv:1801.03190}, 2018.

\bibitem{wilder2017uncharted}
B.~Wilder, A.~Yadav, N.~Immorlica, E.~Rice, and M.~Tambe.
\newblock Uncharted but not uninfluenced: Influence maximization with an
  uncertain network.
\newblock In {\em Proceedings of the 16th Conference on Autonomous Agents and
  MultiAgent Systems}, pages 1305--1313. International Foundation for
  Autonomous Agents and Multiagent Systems, 2017.

\bibitem{yang2012predicting}
Y.~Yang, N.~Chawla, Y.~Sun, and J.~Hani.
\newblock Predicting links in multi-relational and heterogeneous networks.
\newblock In {\em 2012 IEEE 12th international conference on data mining},
  pages 755--764. IEEE, 2012.

\bibitem{zou2013polynomial}
Z.~Zou.
\newblock Polynomial-time algorithm for finding densest subgraphs in uncertain
  graphs.
\newblock In {\em Proceedings of MLG Workshop}, 2013.

\end{thebibliography}

%\newpage

\section{Appendix}
\label{sec:appendix} 
\spara{Uncertain graphs.} Figure~\ref{fig:uncertaingraphs} shows basic statistics for the datasets we use in our experiments. Each row corresponds to a dataset, the first and second columns correspond to the log-histograms of weights and edge probabilities respectively. The last column shows the scatter plot of weights and edge probabilities.

\spara{Risk-averse DSD.} The results of our proposed algorithm on risk averse DSD  for the experiment described in Section~\ref{sec:exp} are shown in Figure~\ref{fig:riskaversefull}. In this experiment $B=1$, and we range $C$ for three different $(\lambda_1,\lambda_2)$ pairs, i.e., $(0.5,1),(1,1),(2,1)$  Each row corresponds to a dataset, the first, second, and third columns to the average expected weight, average risk, and size of the output. The rows correspond to Collins, Gavin, Krogan core, Krogan extended, and TMDB respectively.  

\begin{figure*}[htp]
\centering
\begin{tabular}{@{}c@{}@{\ }c@{}@{\ }c@{}}
\includegraphics[width=0.33\textwidth]{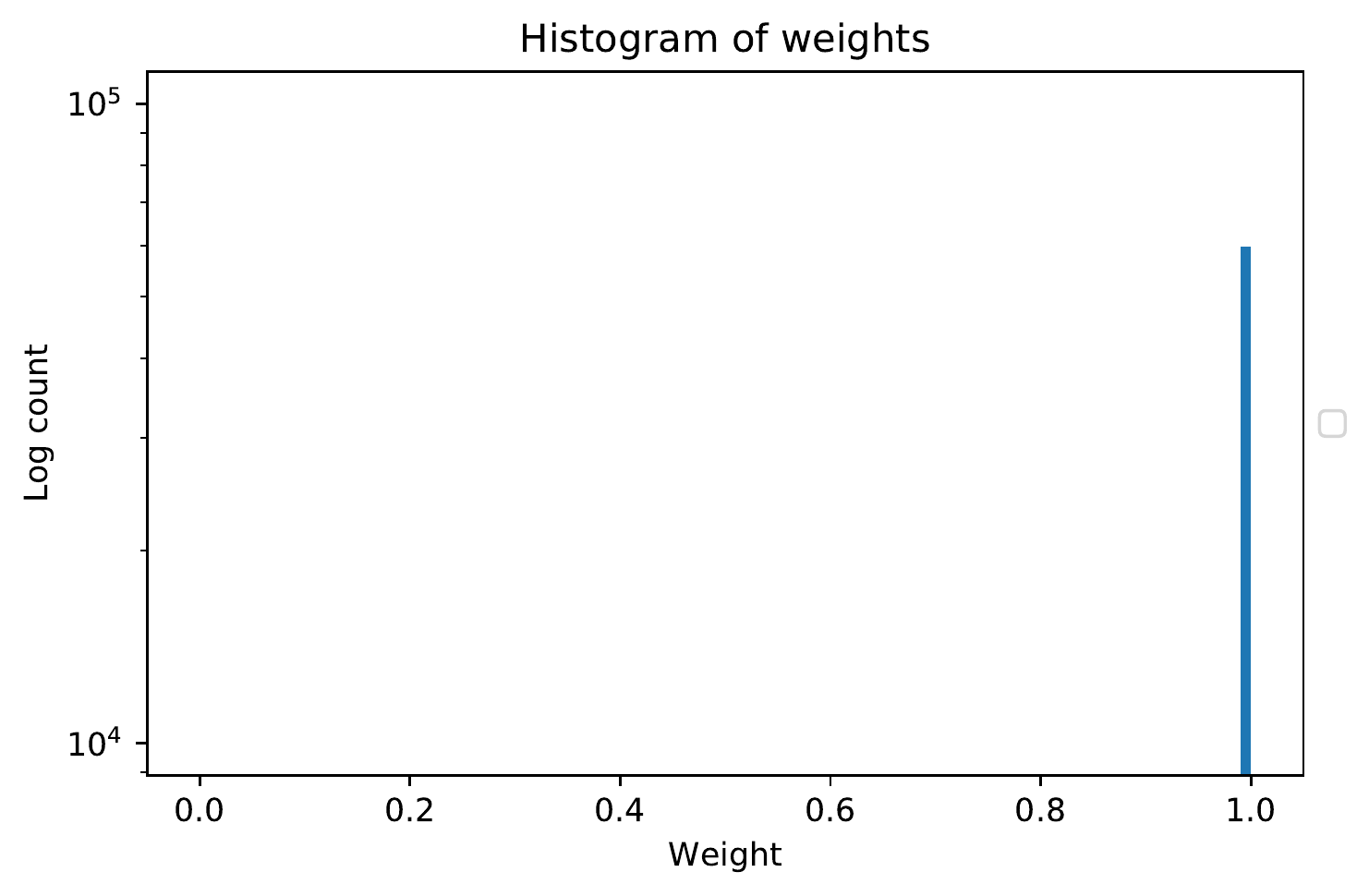} & \includegraphics[width=0.33\textwidth]{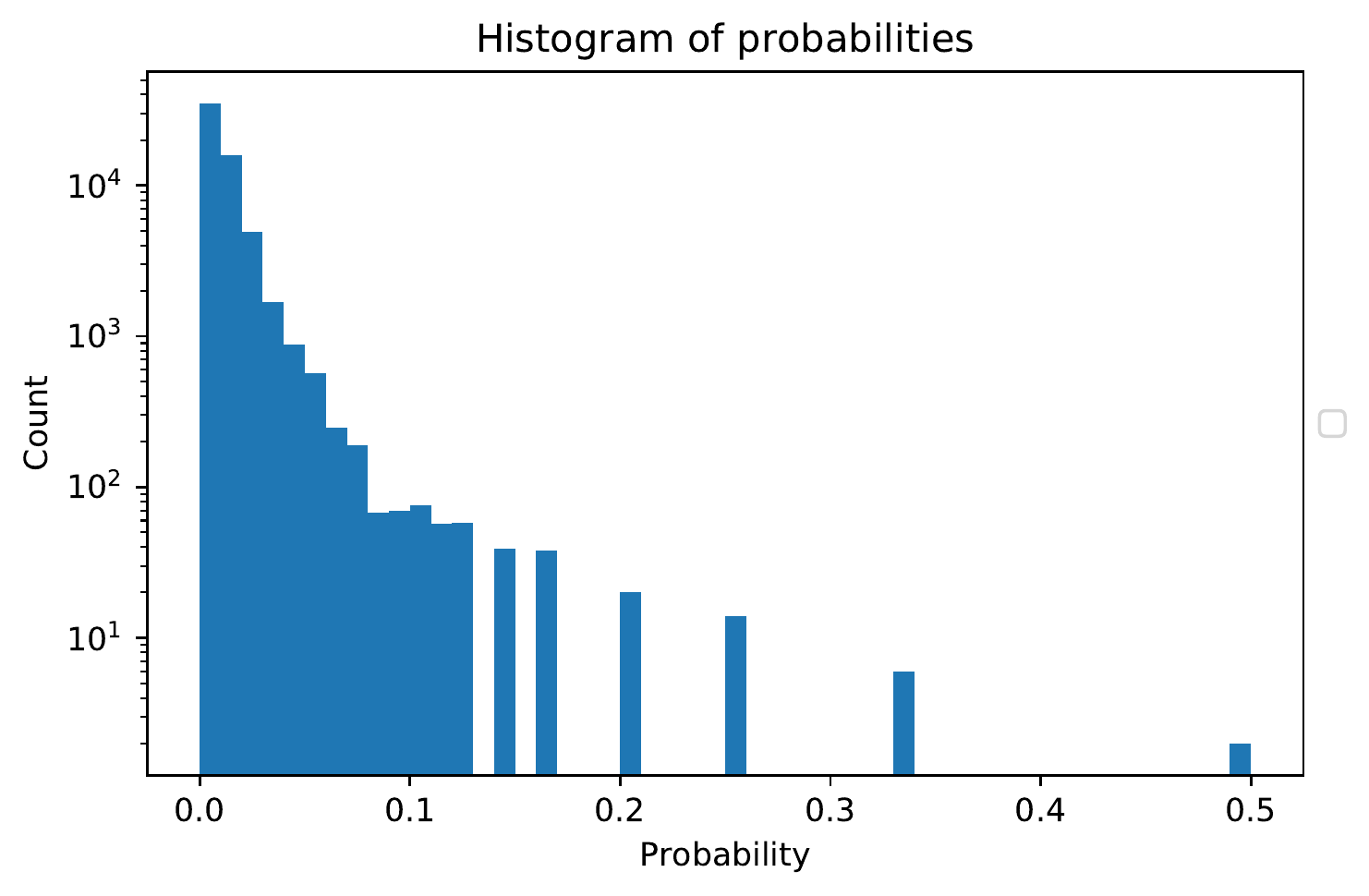}    & \includegraphics[width=0.33\textwidth]{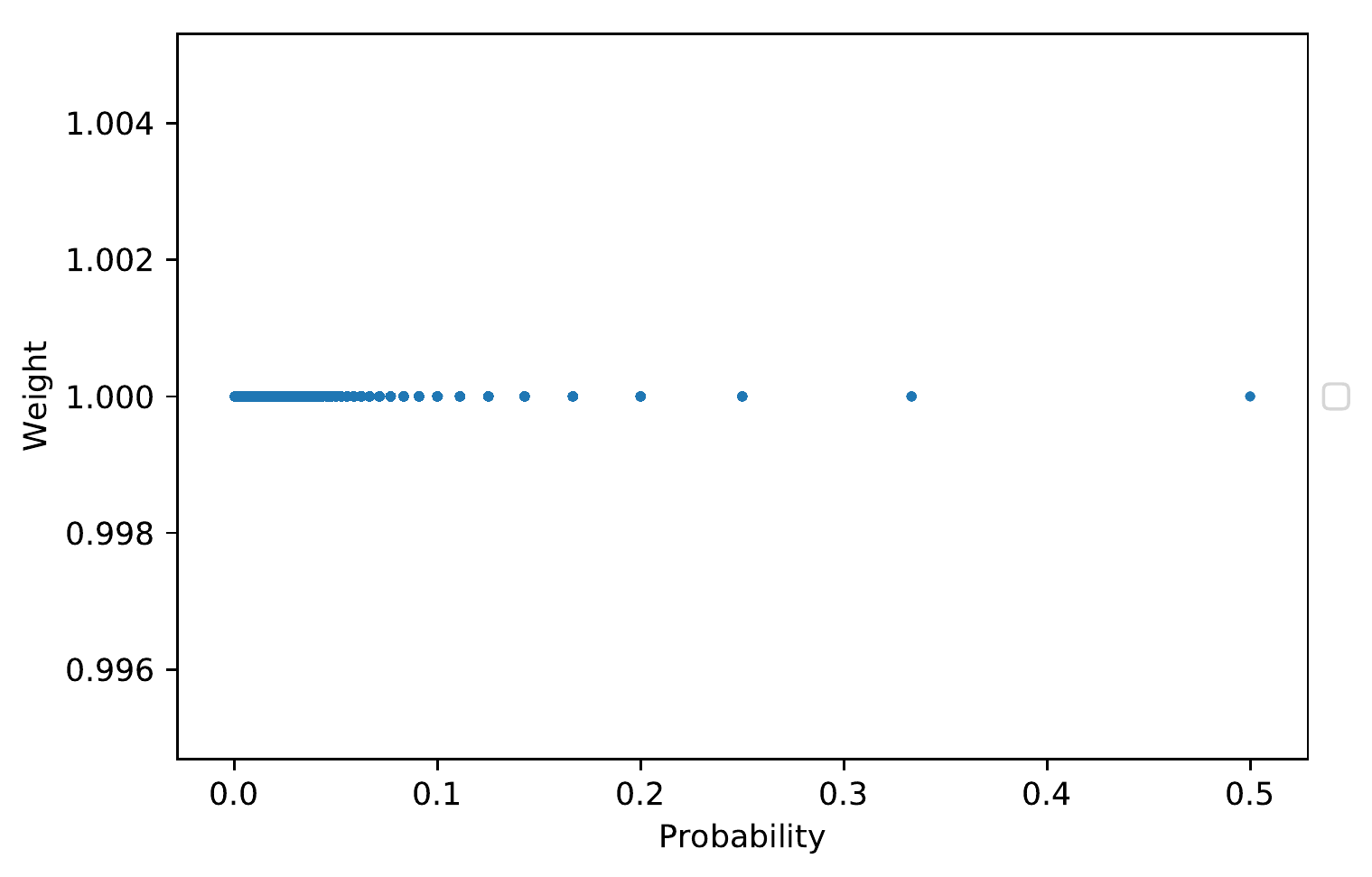}       \\
($\alpha$) & ($\beta$) & ($\gamma$) \\
\includegraphics[width=0.33\textwidth]{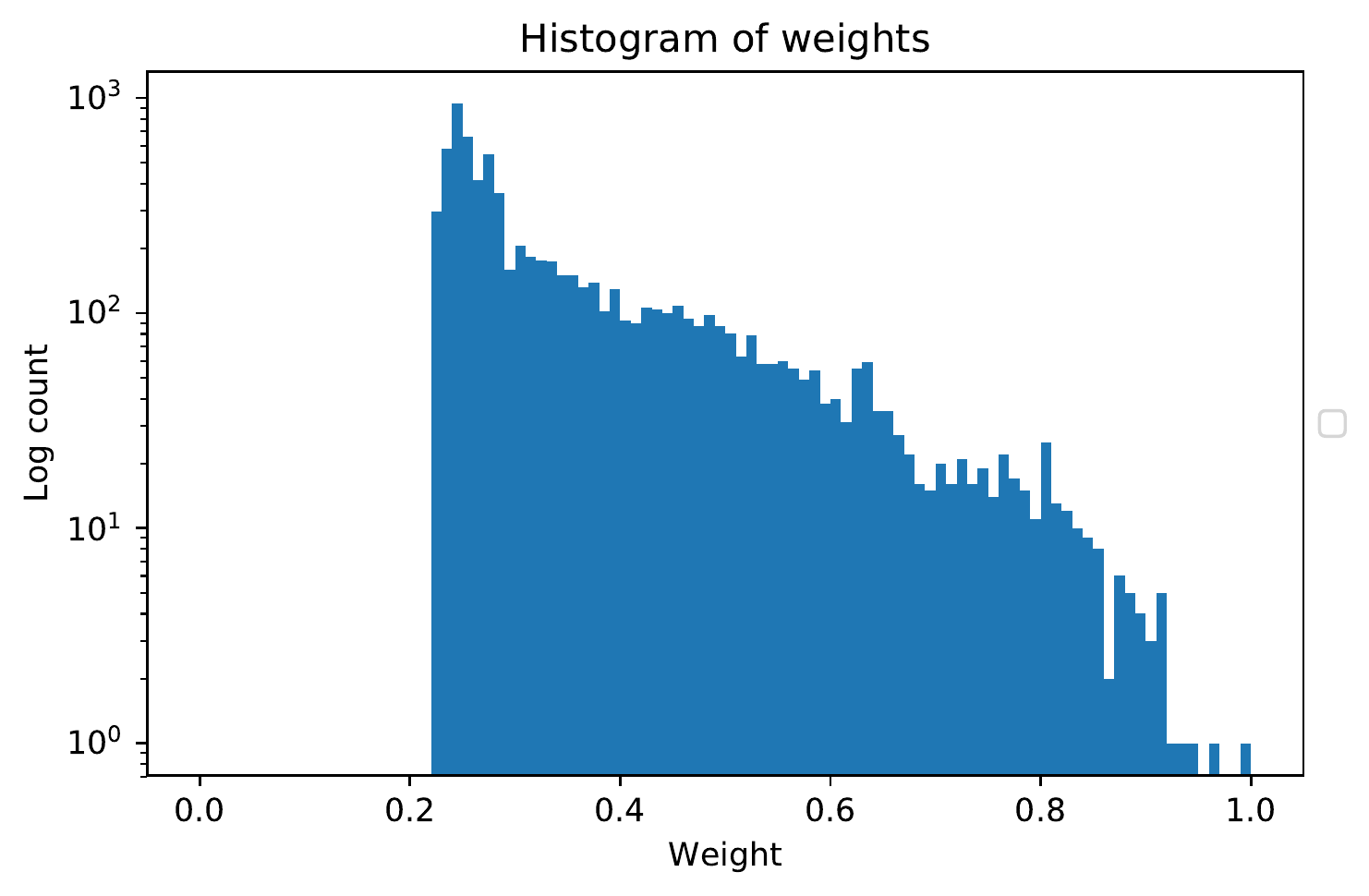} & \includegraphics[width=0.33\textwidth]{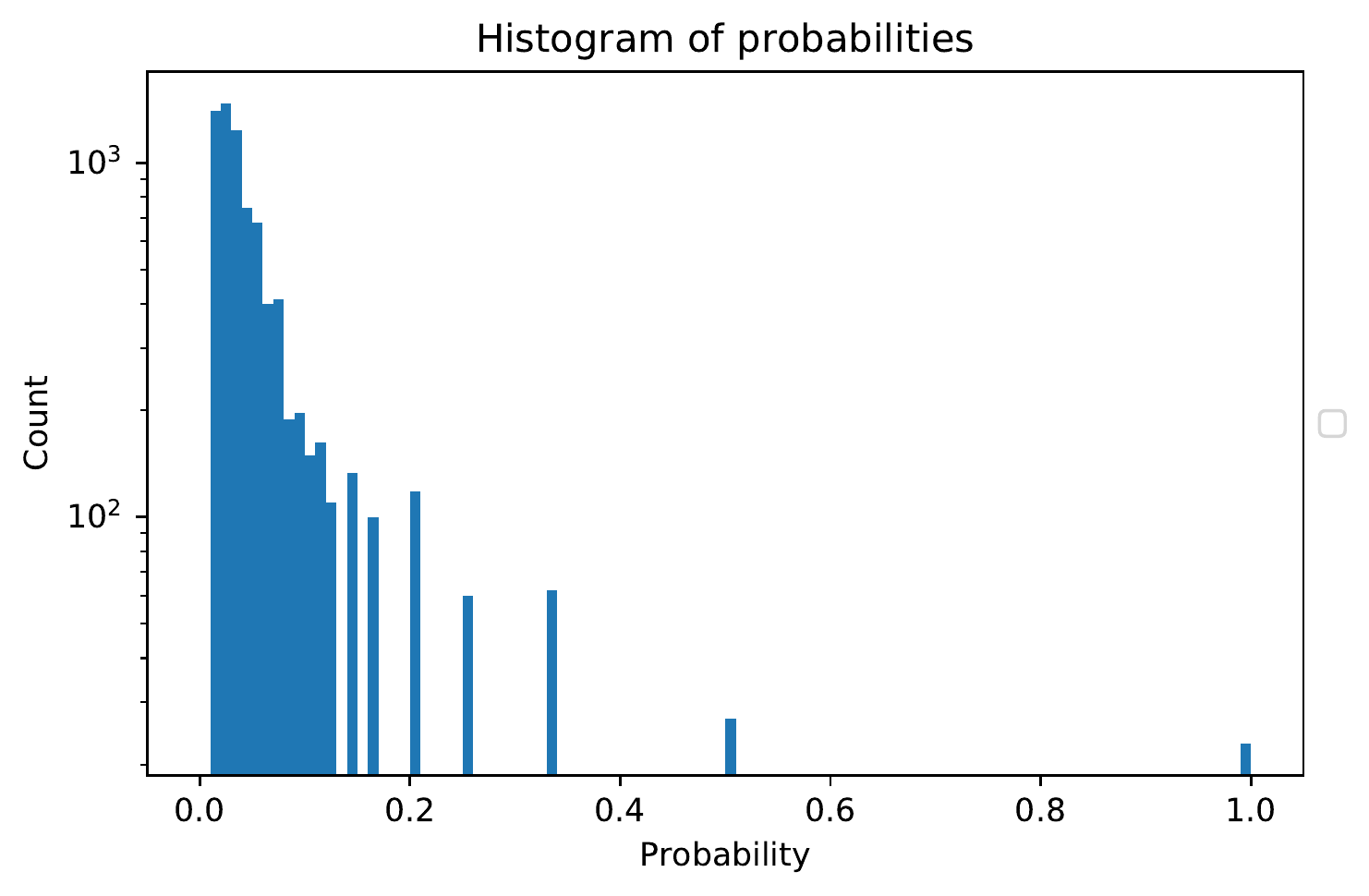}    & \includegraphics[width=0.33\textwidth]{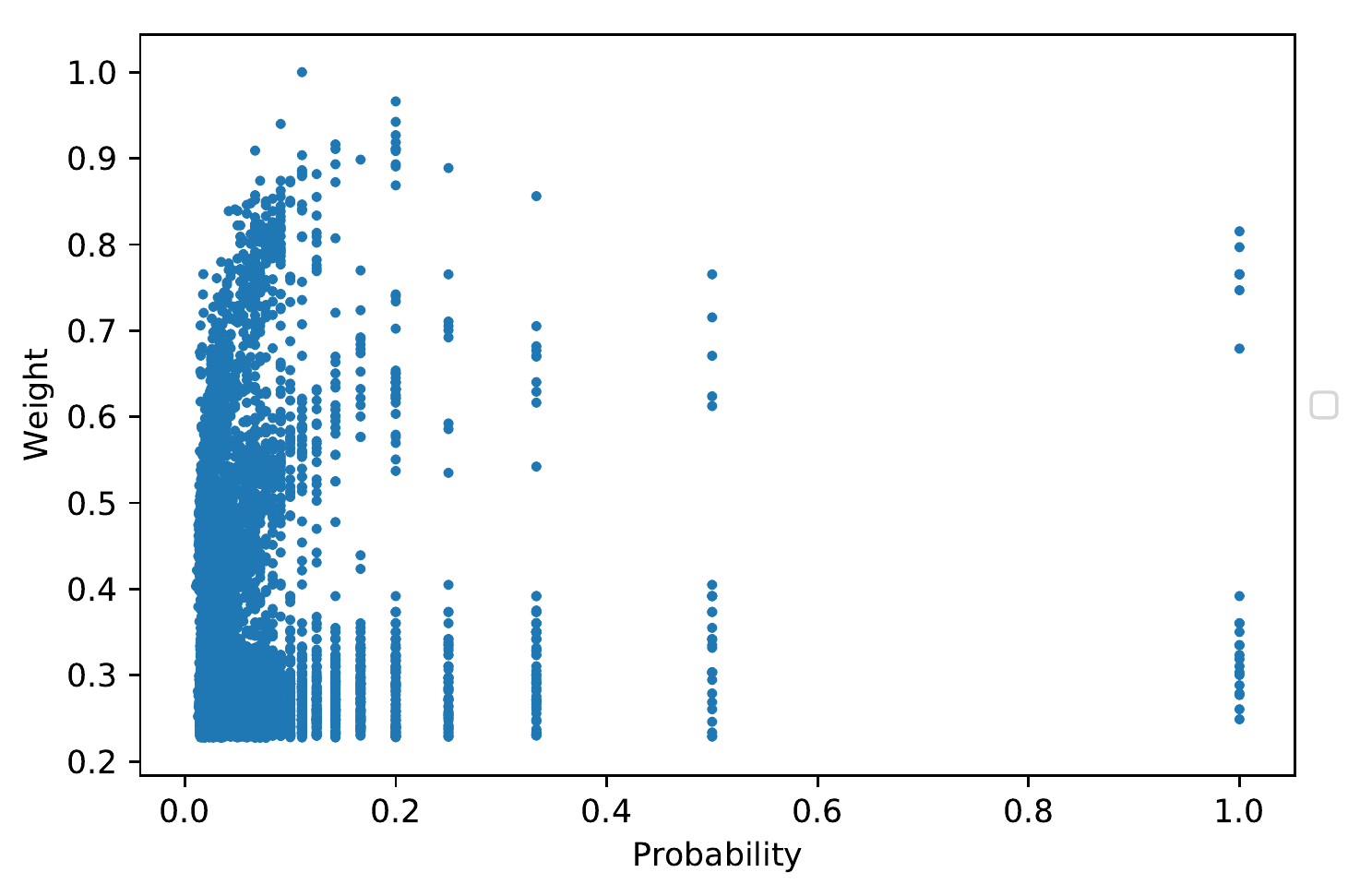}       \\
($\delta$) & ($\epsilon$) & ($\sigma\tau$) \\ 
 \includegraphics[width=0.33\textwidth]{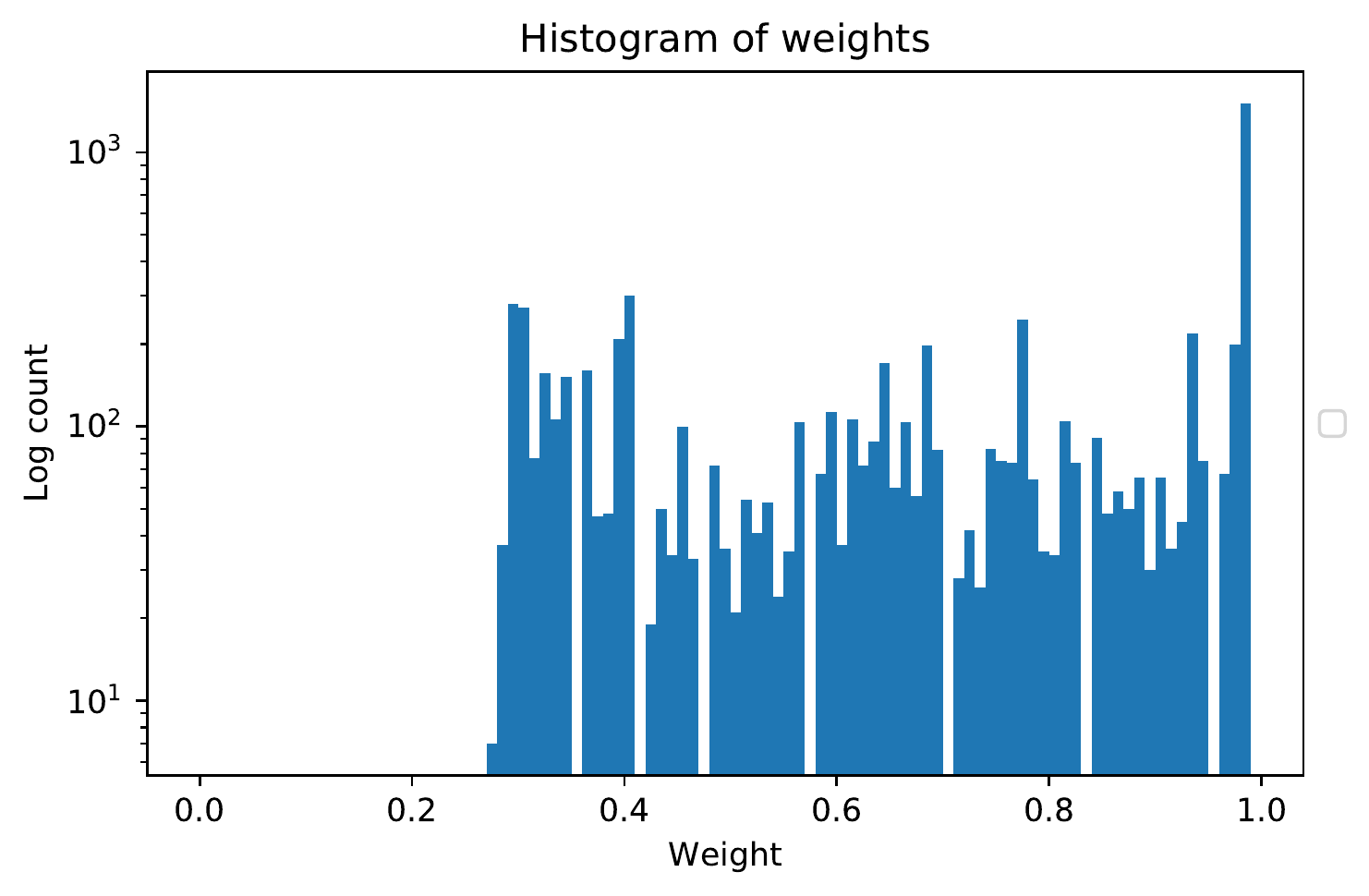} & \includegraphics[width=0.33\textwidth]{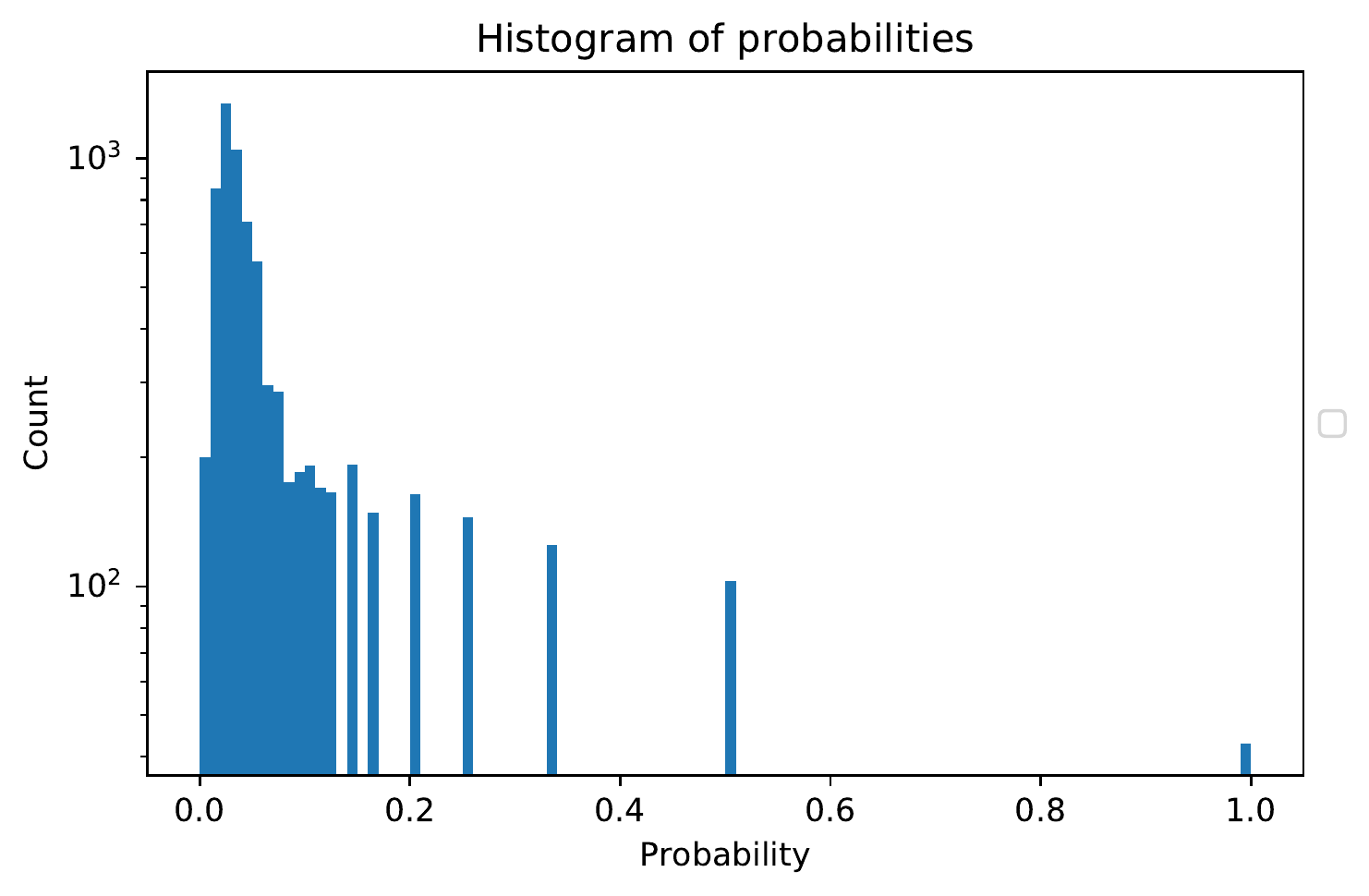}    & \includegraphics[width=0.33\textwidth]{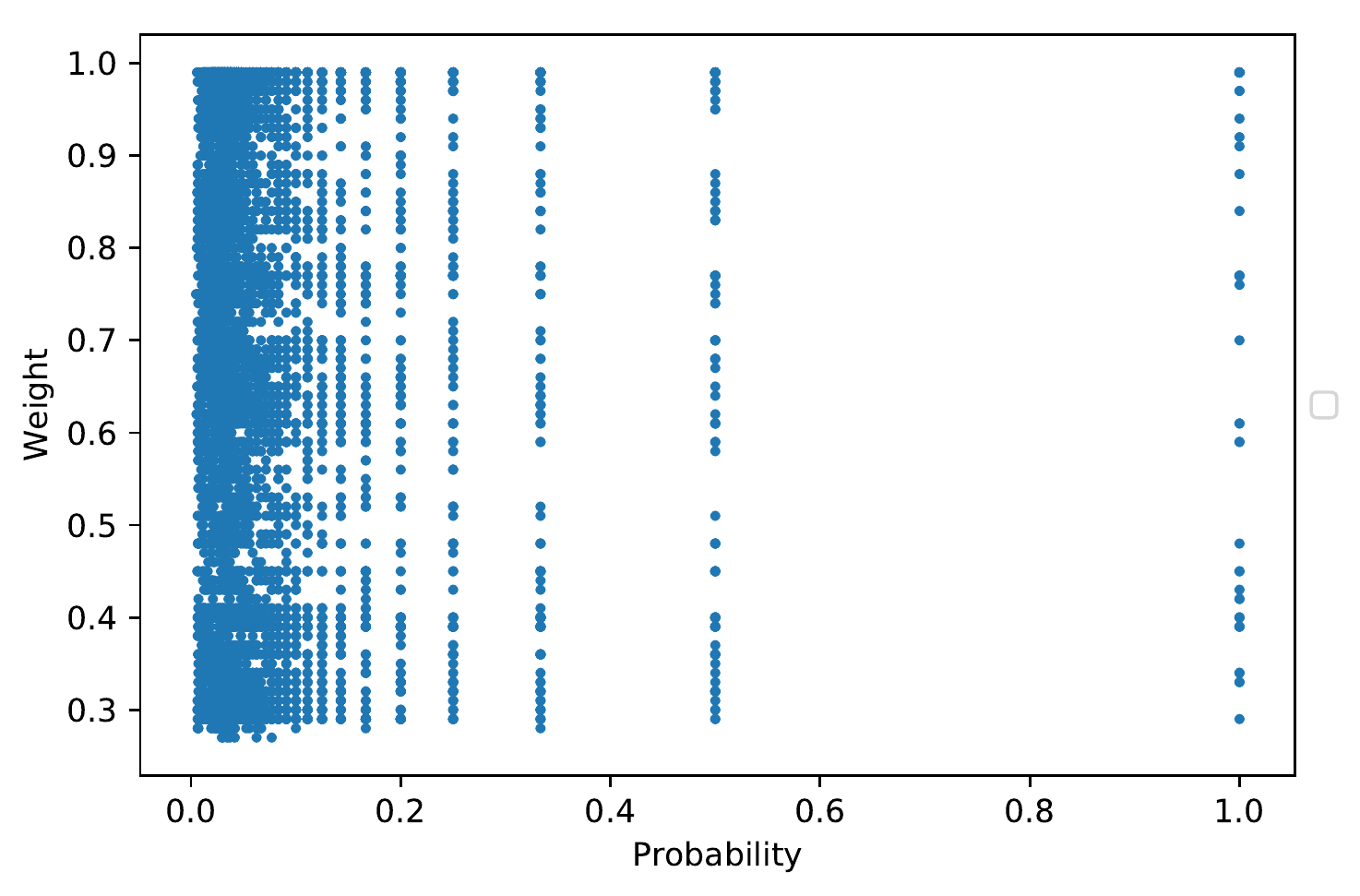}       \\
($\zeta$) & ($\eta$) & ($\theta$) \\ 
\includegraphics[width=0.33\textwidth]{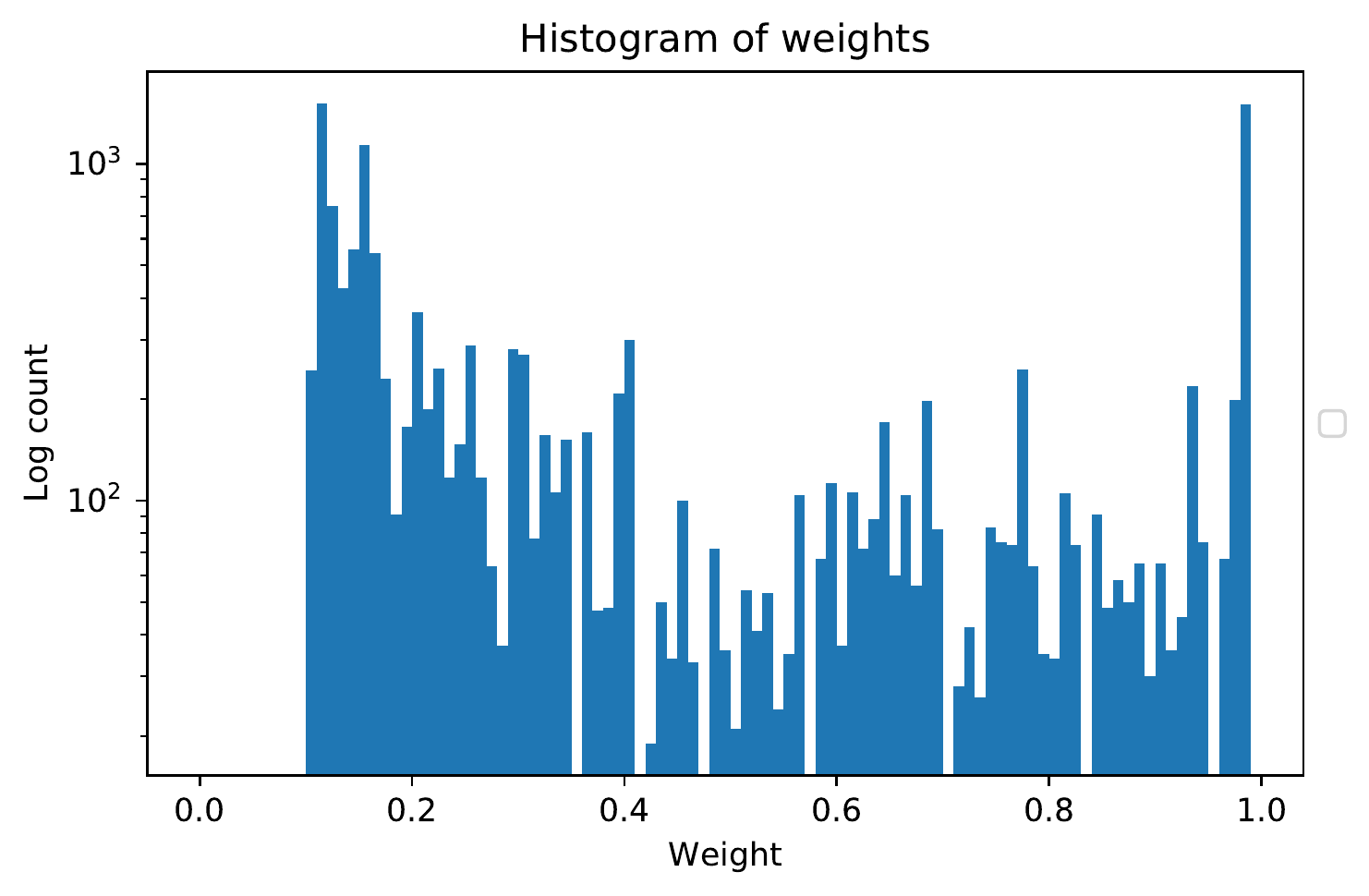} & \includegraphics[width=0.33\textwidth]{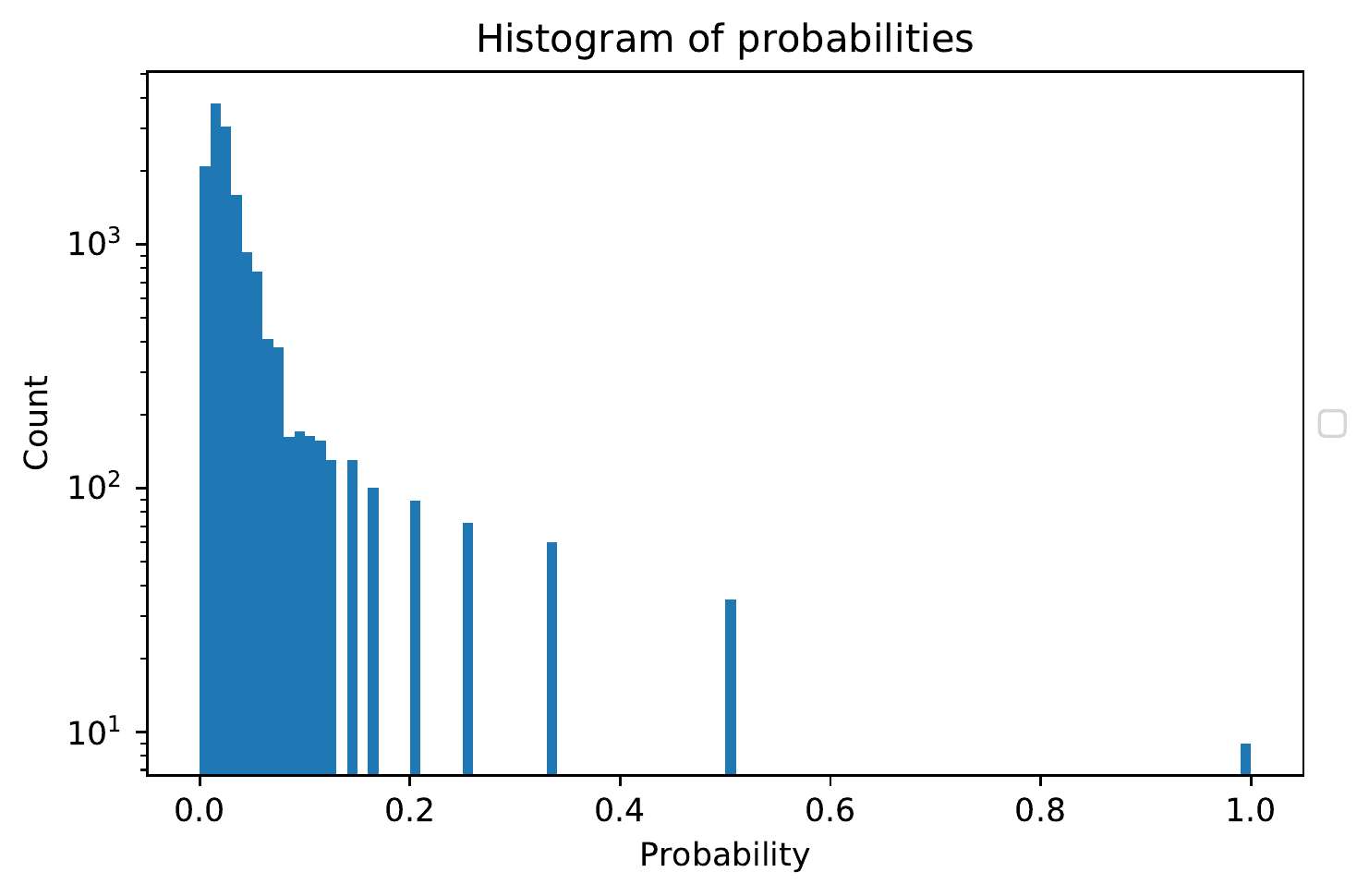}    & \includegraphics[width=0.33\textwidth]{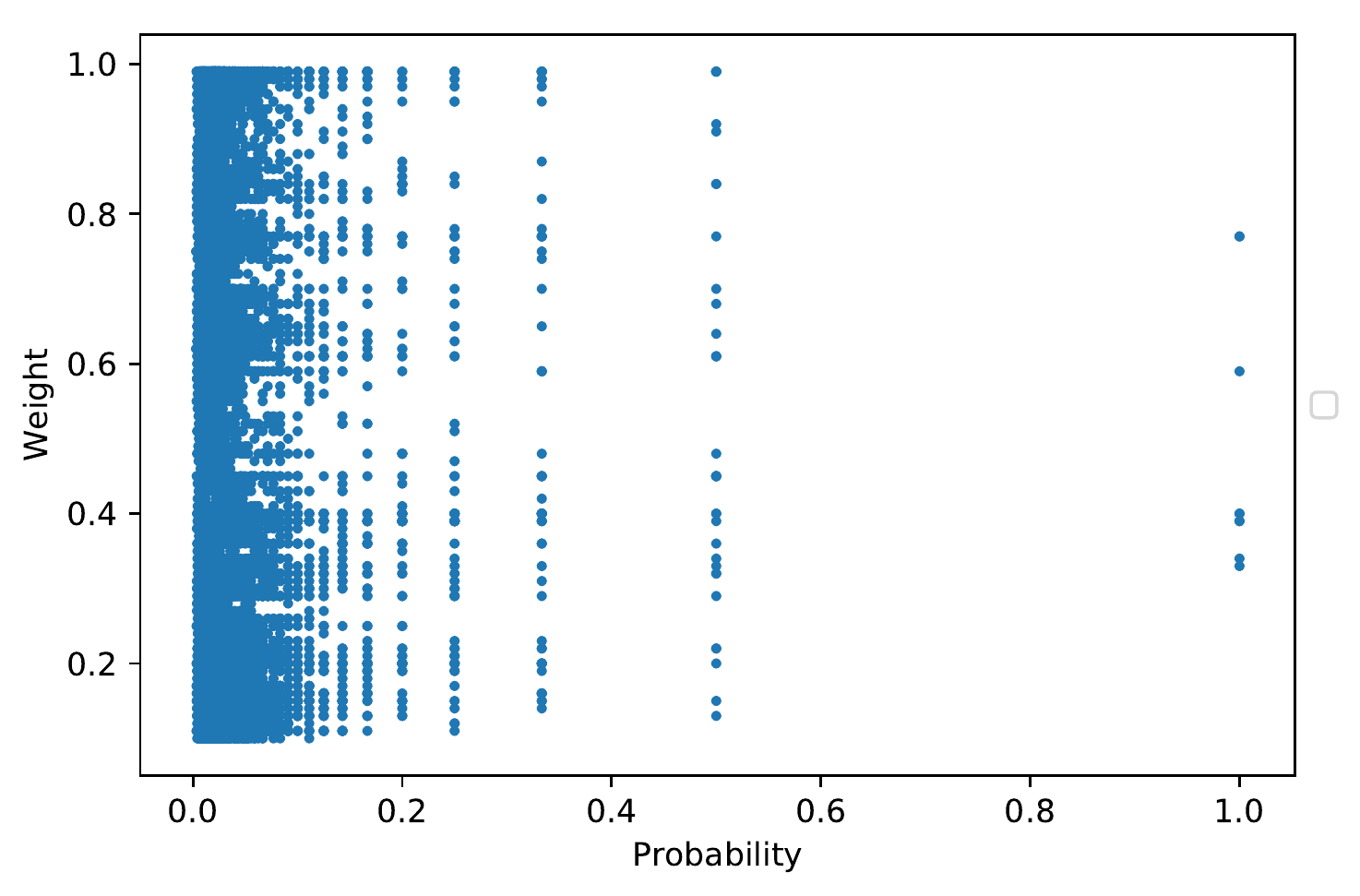}       \\
($\iota$) & ($\iota\alpha$) & ($\iota\beta$) \\   
\end{tabular}
\caption{\label{fig:uncertaingraphs_appendix} Uncertain graphs' statistics. Each row corresponds to {\em biogrid, gavin, krogan, krogan extended} respectively. First and second column show histograms of weights and edge probabilities respectively. The third column shows the scatterplot among the latter quantities.}
\end{figure*}

\begin{figure*}[htp]
\centering
\begin{tabular}{@{}c@{}@{\ }c@{}@{\ }c@{}}
\includegraphics[width=0.33\textwidth]{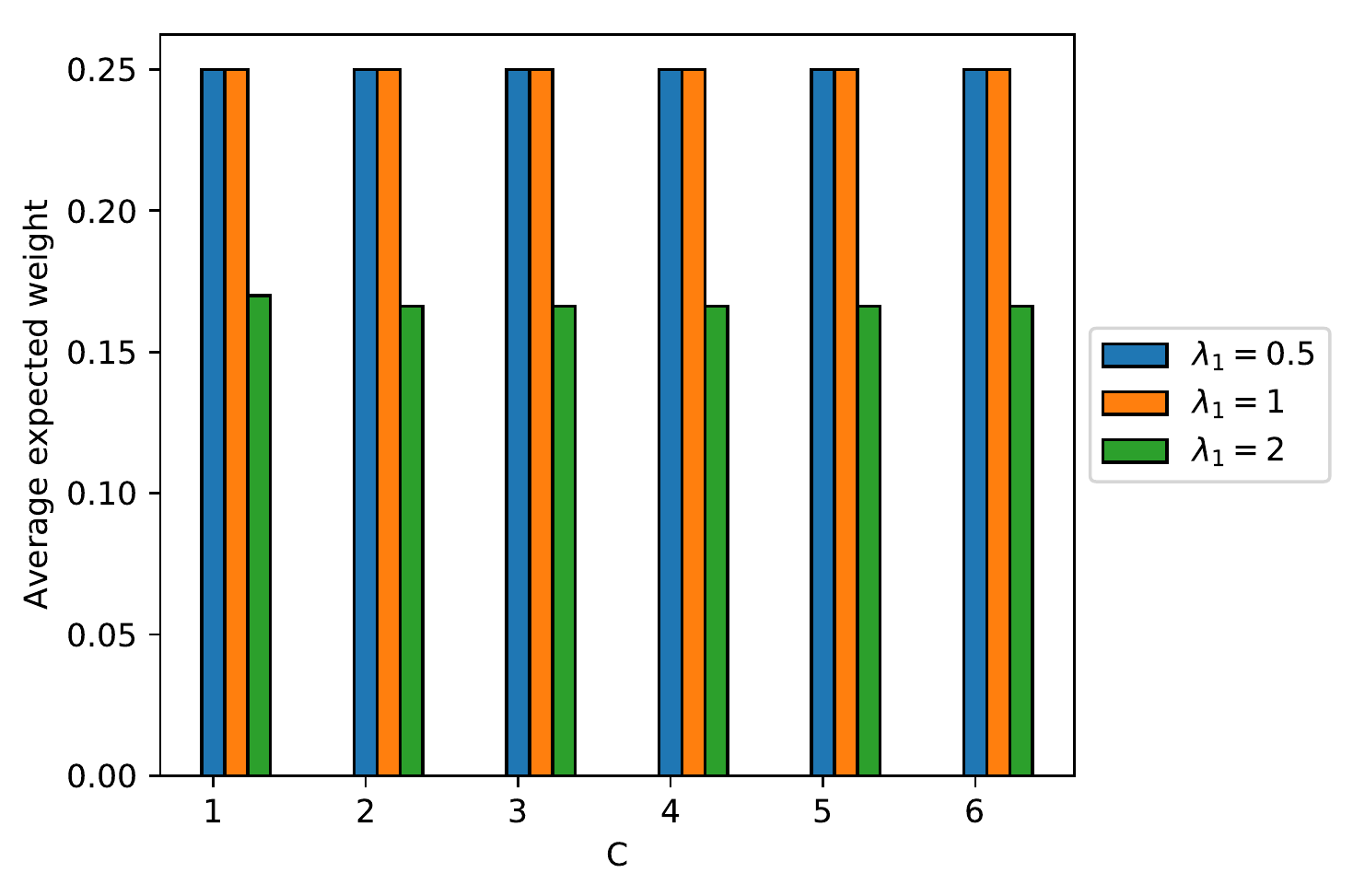} & \includegraphics[width=0.33\textwidth]{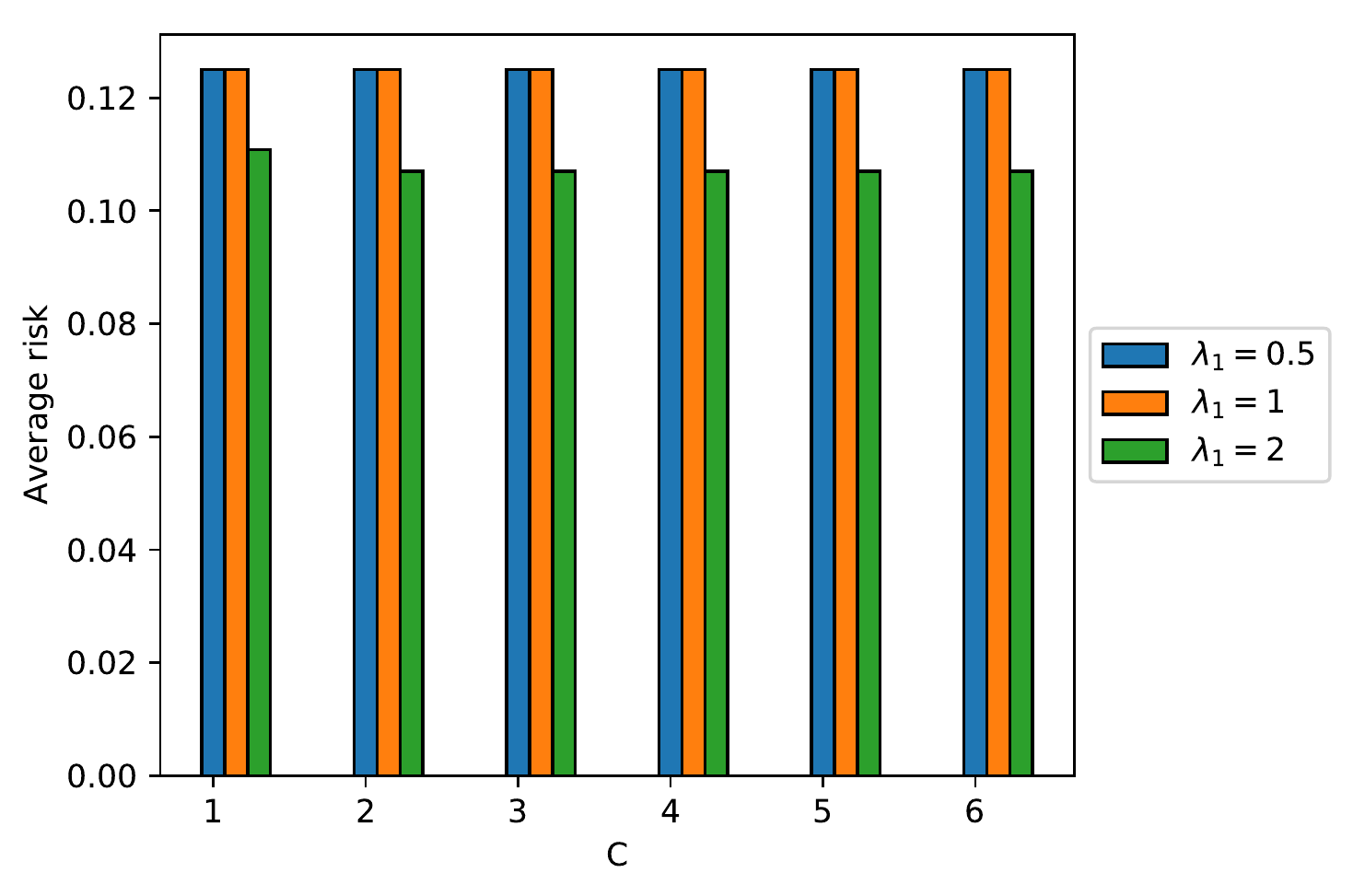}    & \includegraphics[width=0.33\textwidth]{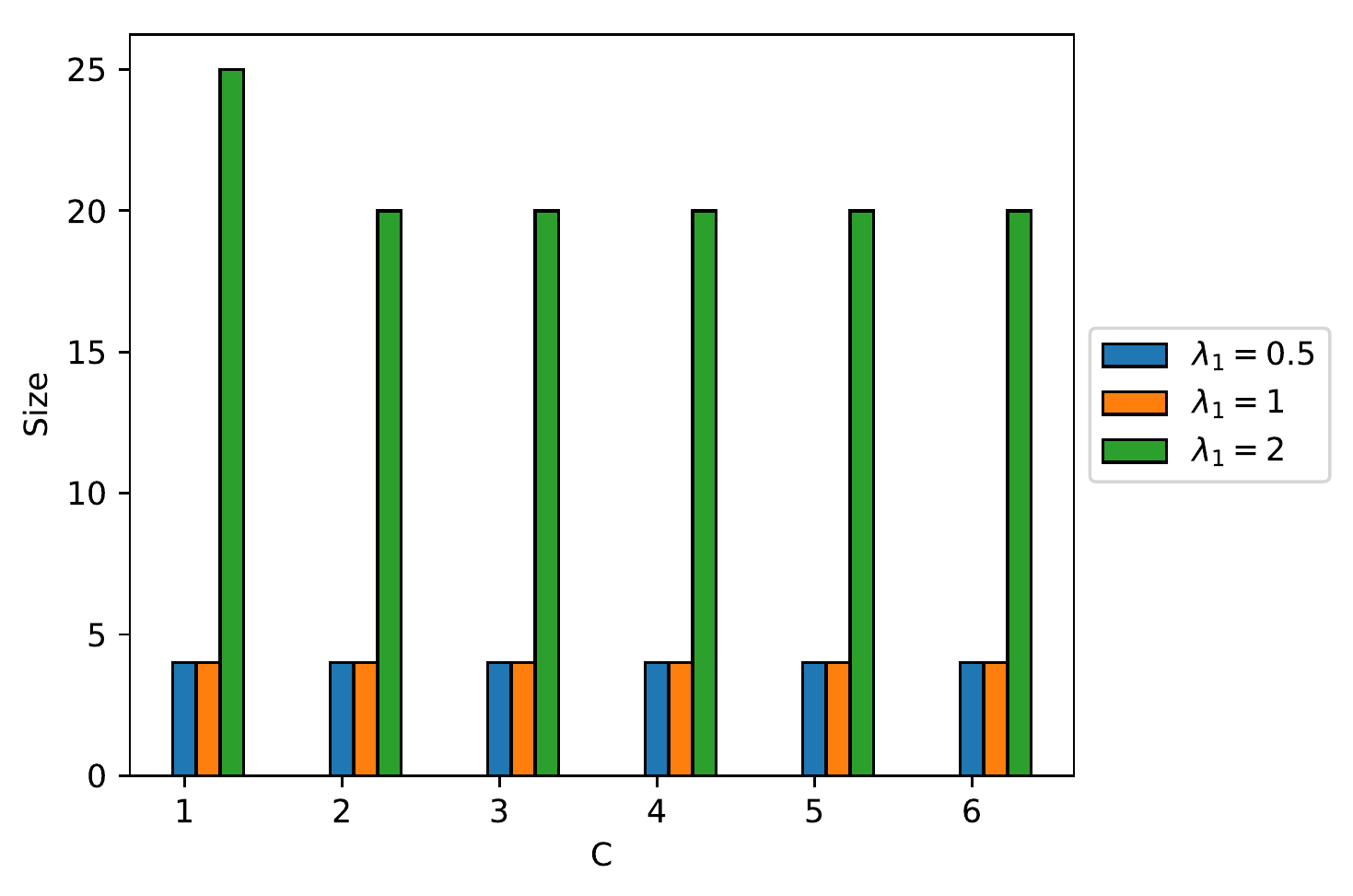}    
\\
($\alpha$) & ($\beta$) & ($\gamma$) \\ 
\includegraphics[width=0.3\textwidth]{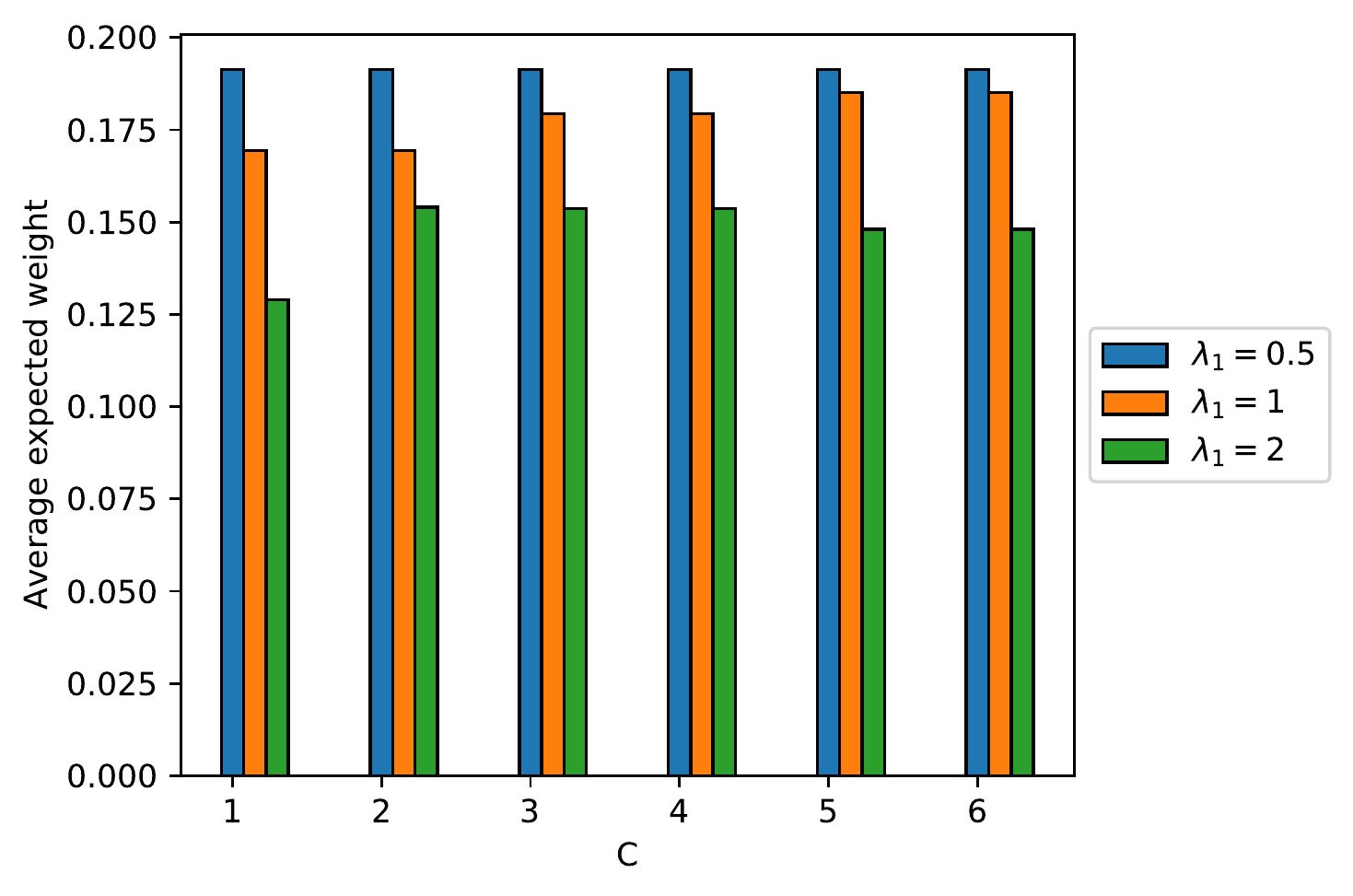} & \includegraphics[width=0.3\textwidth]{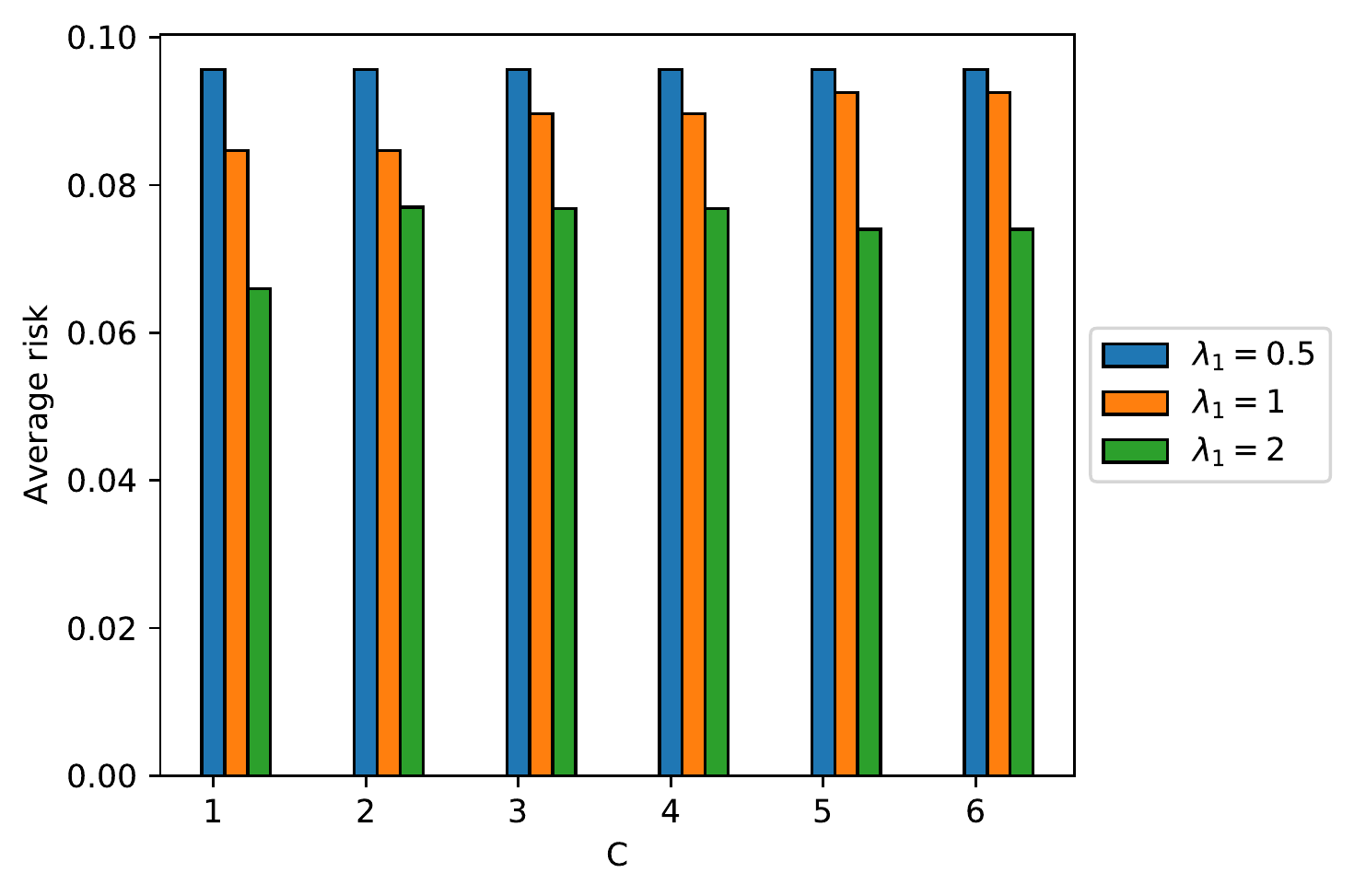}    & \includegraphics[width=0.3\textwidth]{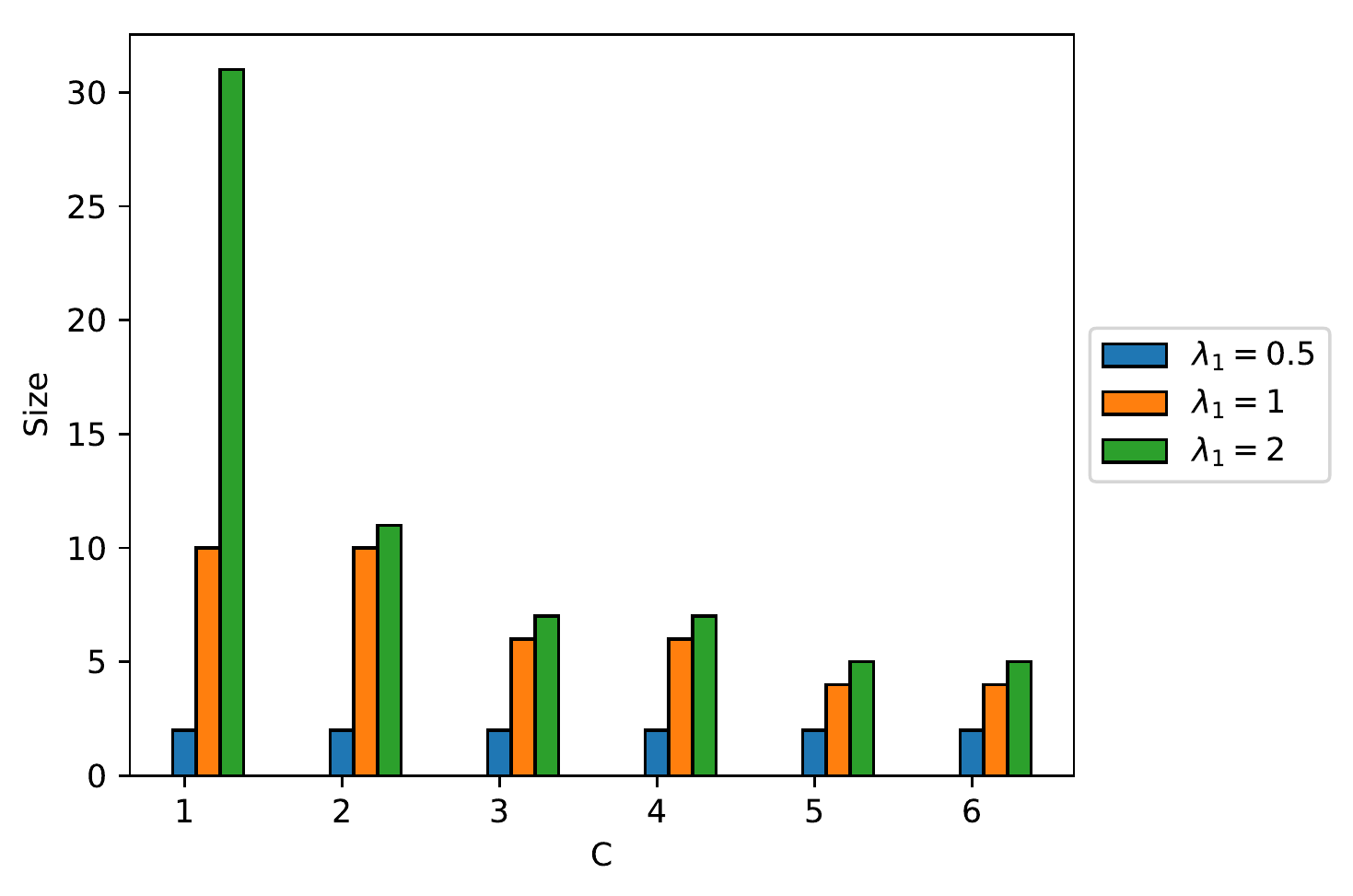}       \\
($\delta$) & ($\epsilon$) & ($\sigma\tau$) \\ 
\includegraphics[width=0.3\textwidth]{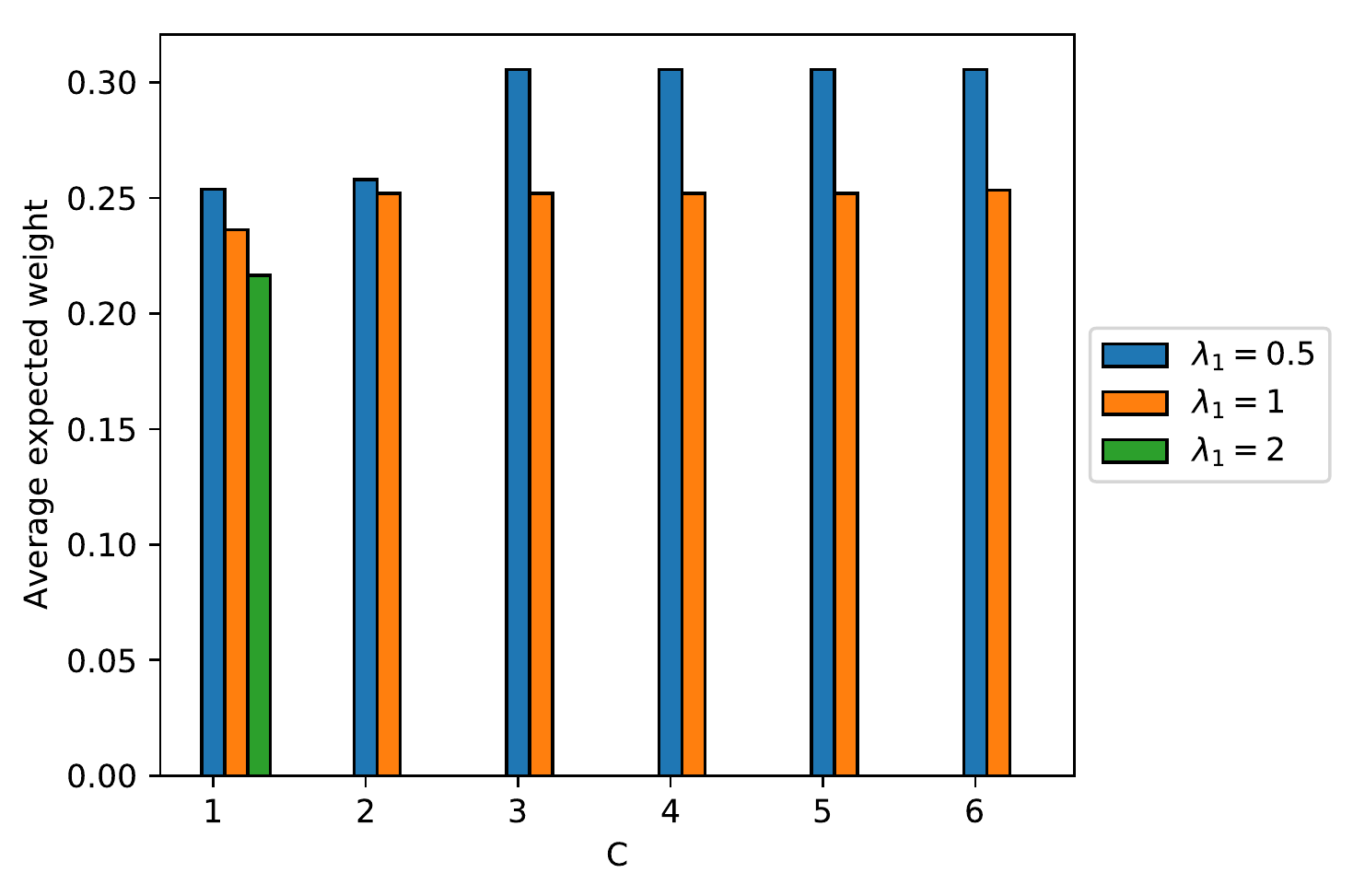} & \includegraphics[width=0.3\textwidth]{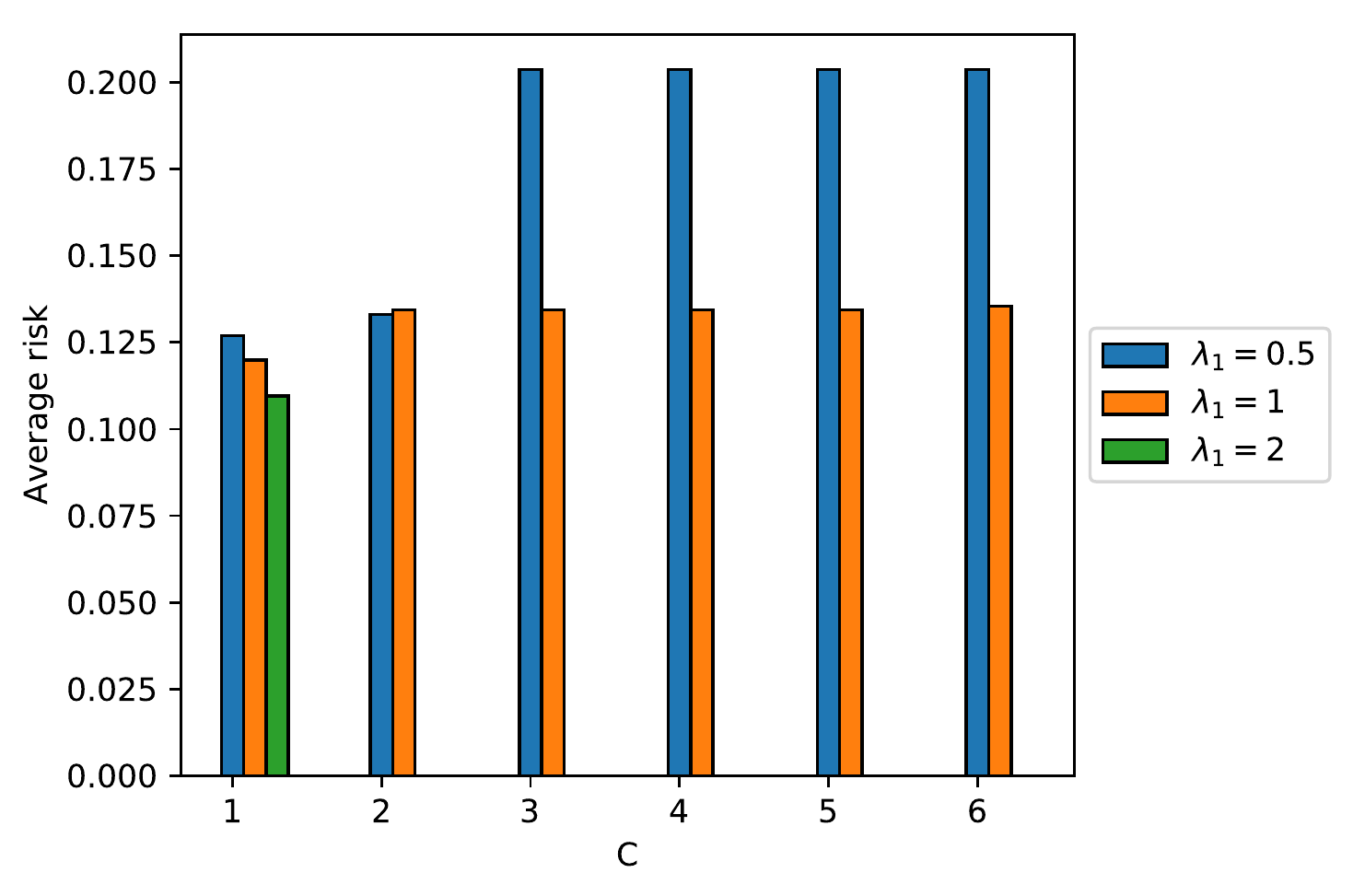}    & \includegraphics[width=0.3\textwidth]{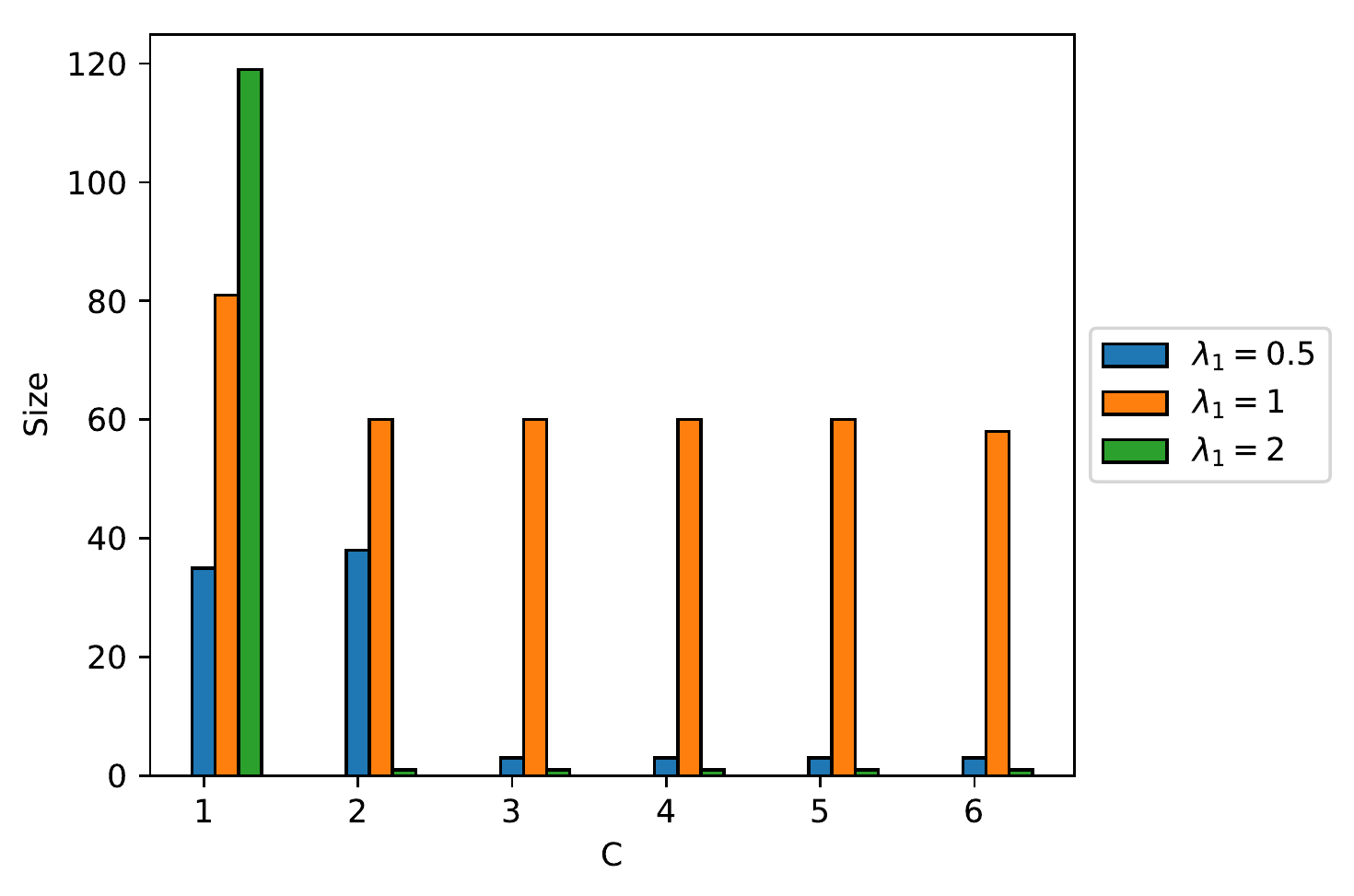}       \\
($\zeta$) & ($\eta$) & ($\theta$) \\ 
\includegraphics[width=0.3\textwidth]{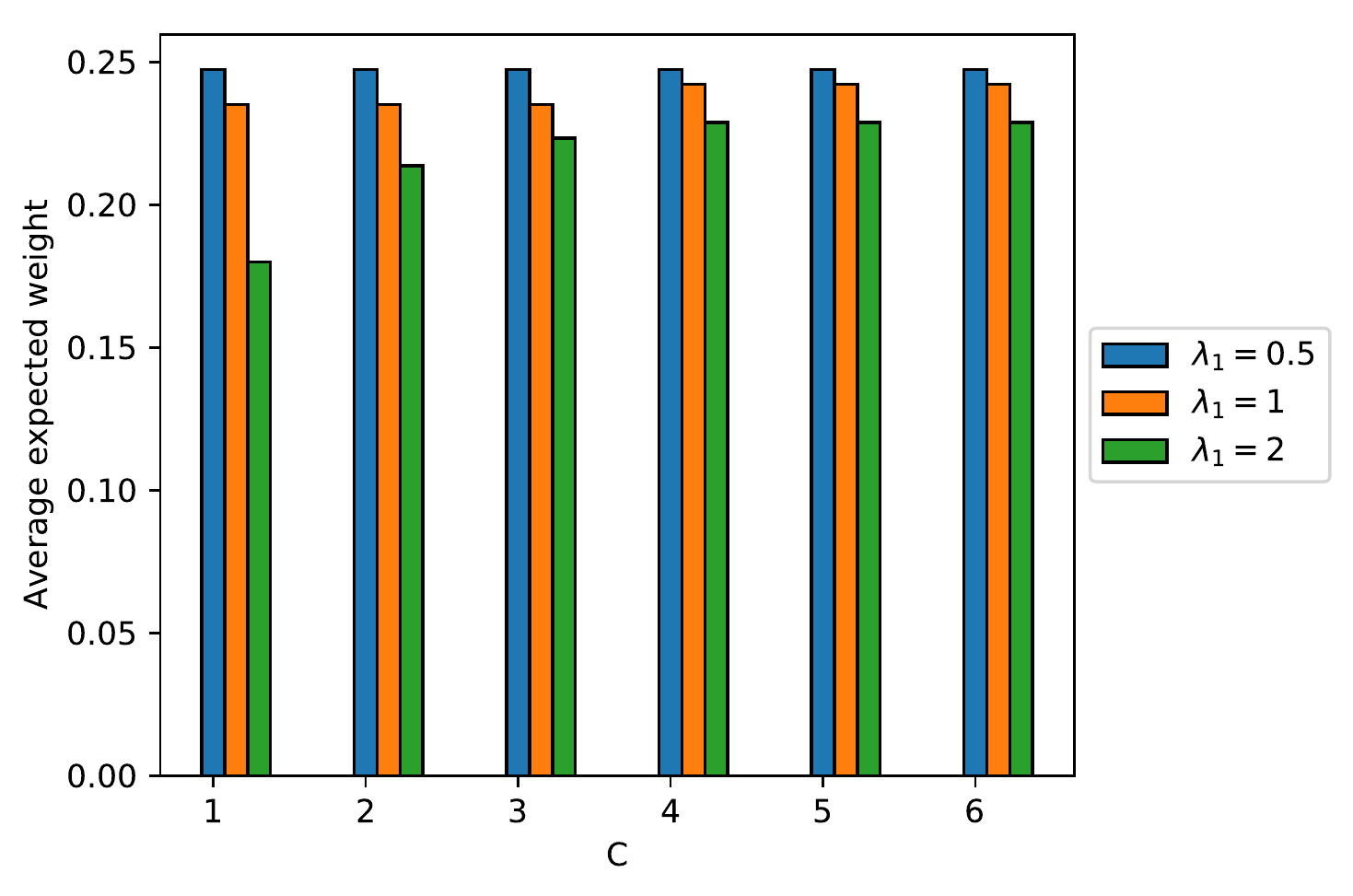} & \includegraphics[width=0.3\textwidth]{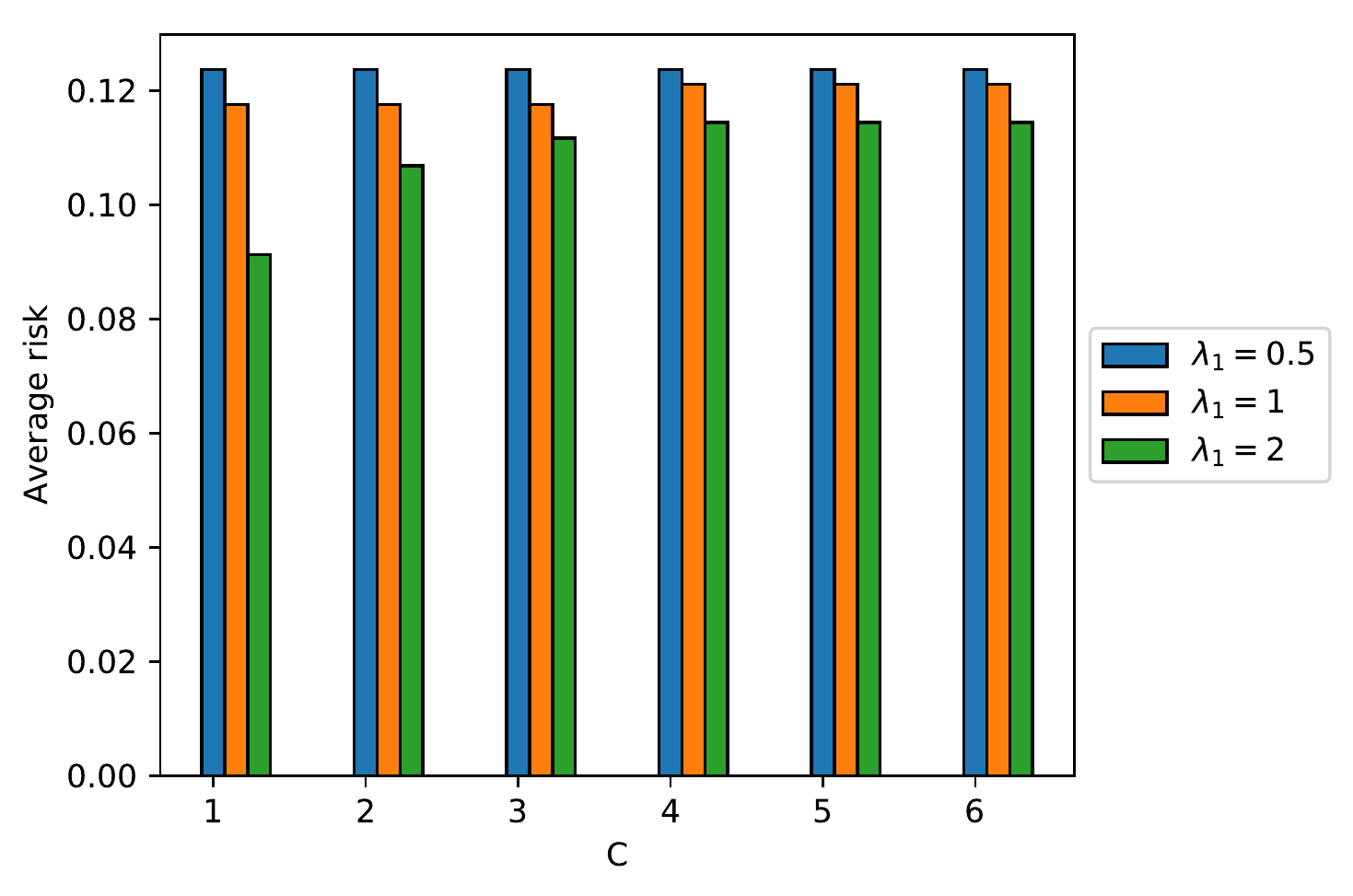}    & \includegraphics[width=0.3\textwidth]{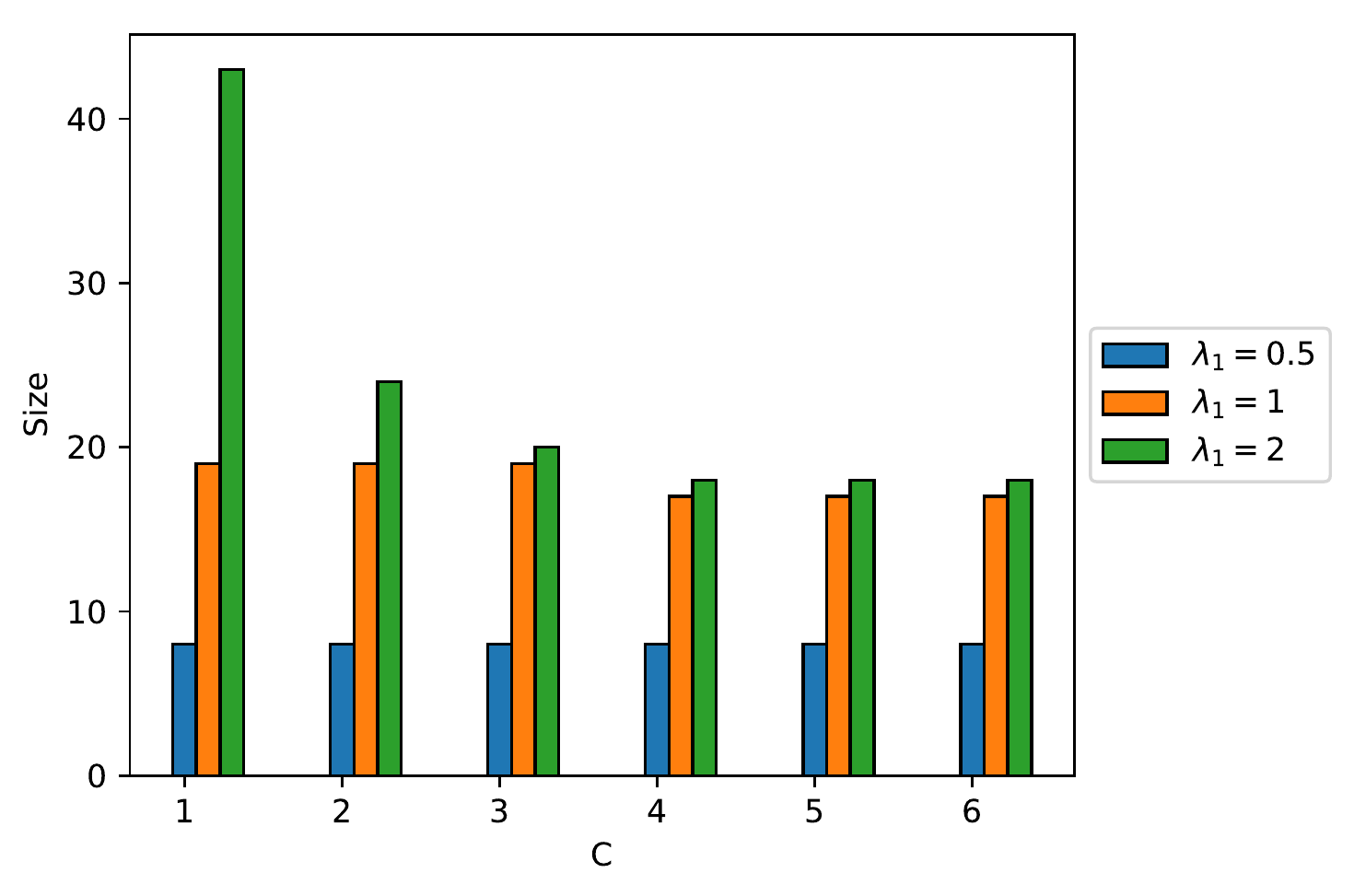}       \\
($\iota$) & ($\iota\alpha$) & ($\iota\beta$) \\  
\includegraphics[width=0.3\textwidth]{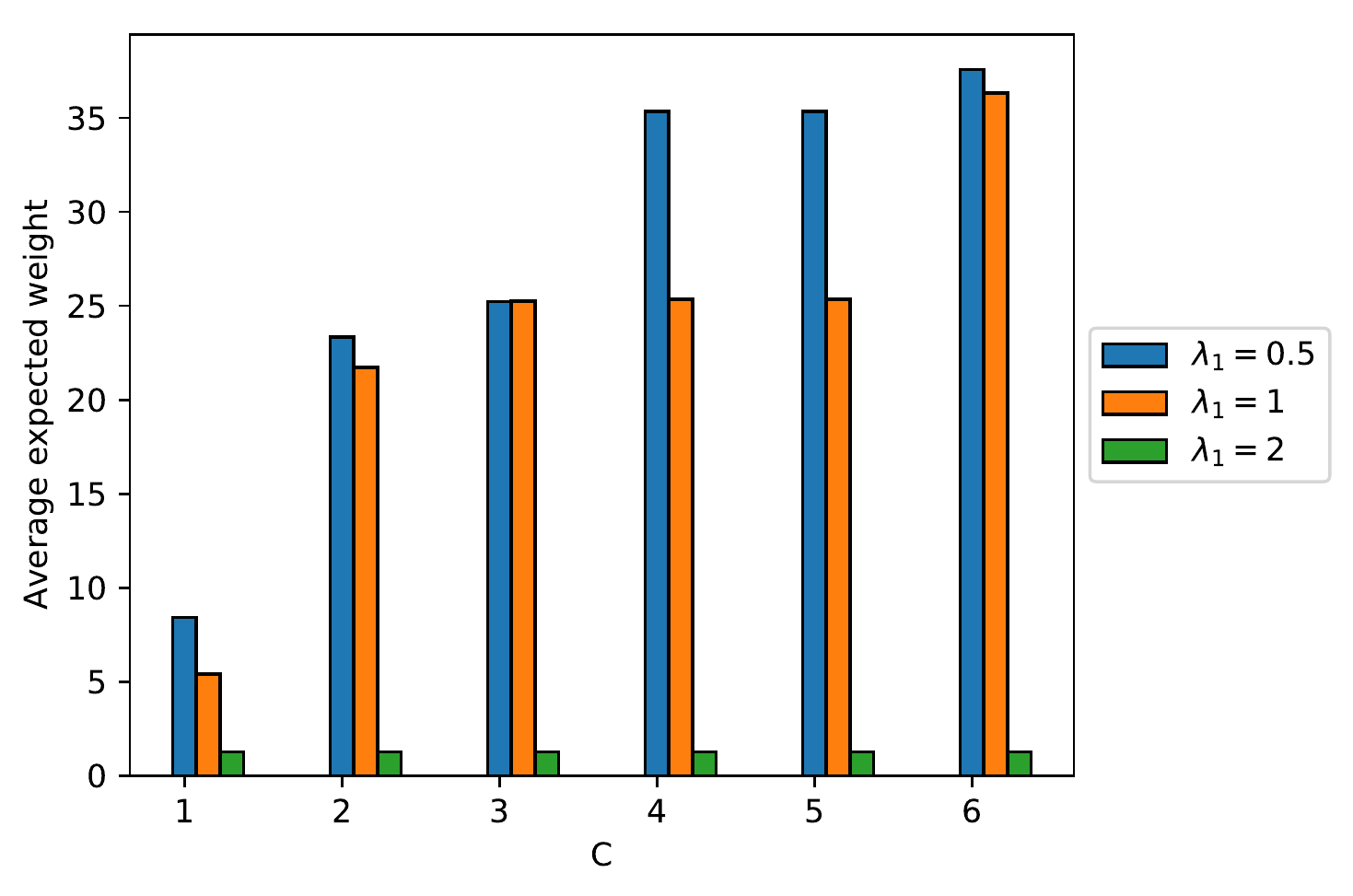} & \includegraphics[width=0.3\textwidth]{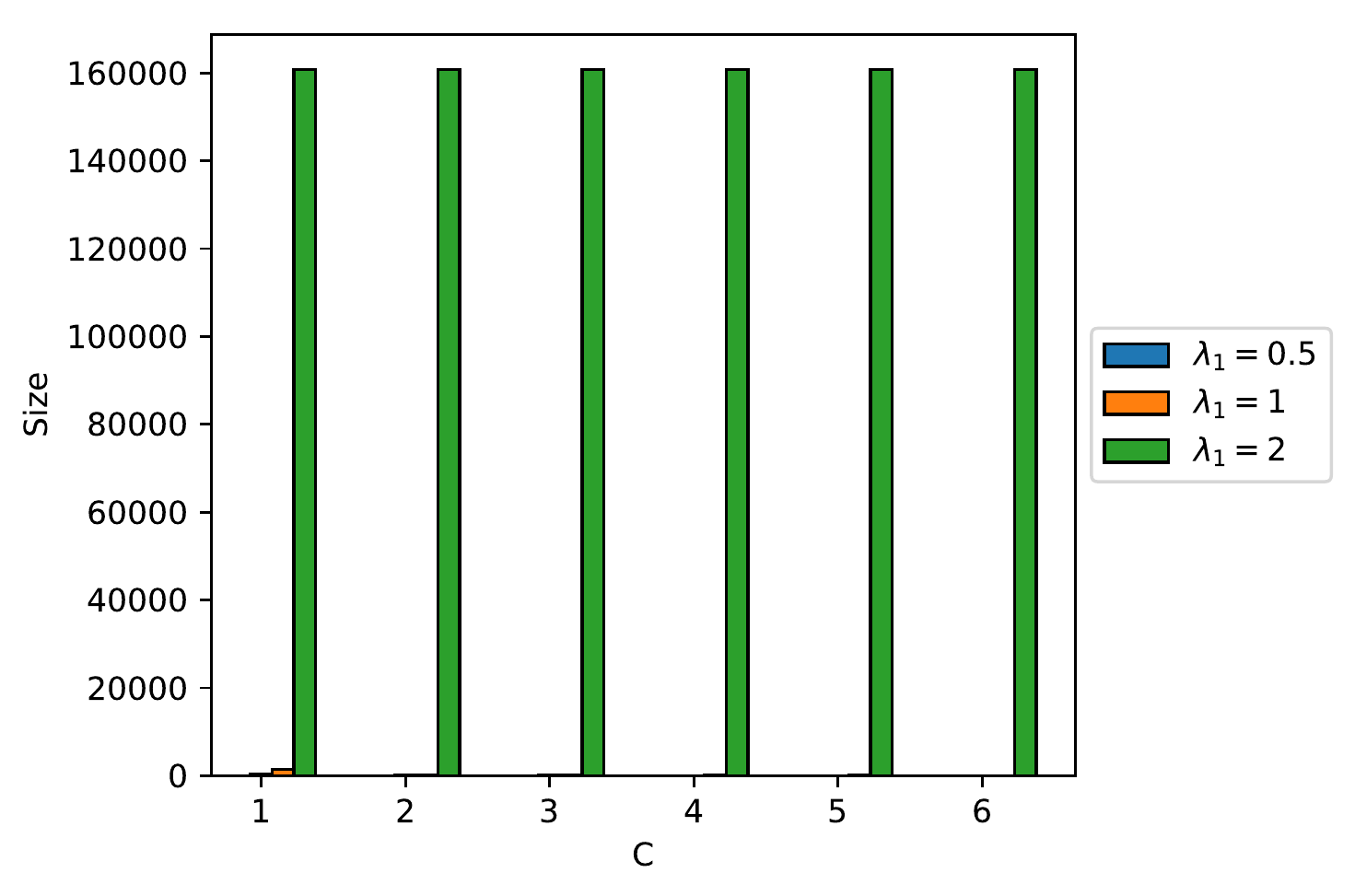}    & \includegraphics[width=0.3\textwidth]{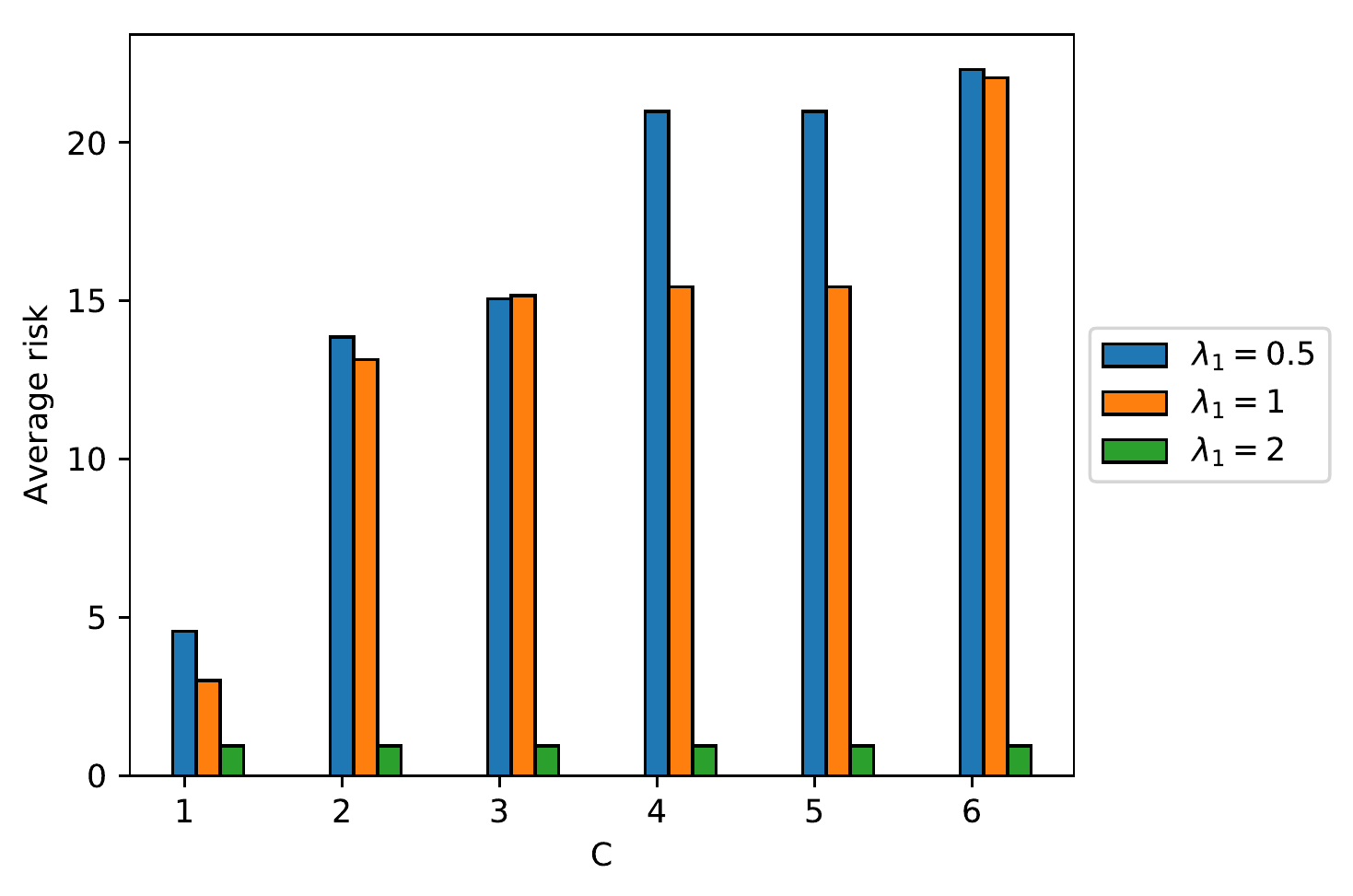}       \\
($\iota\gamma$) & ($\iota\delta$) & ($\iota\epsilon$) \\   
\end{tabular}
\caption{\label{fig:riskaversefull} Effect of $C$, $(\lambda_1,\lambda_2)$ on risk averse DSD. For details, see text.   }
\end{figure*}

\end{document}